\documentclass[preprint,a4paper,10pt,authoryear]{elsarticle}

\usepackage{mathtools}
\usepackage{amssymb}
\usepackage{amsmath}
\usepackage{amsthm}
\usepackage{amsfonts}
\usepackage{xspace}
\usepackage{paralist}
\usepackage{hyperref}
\usepackage{wrapfig}
\usepackage{subcaption}
\usepackage{booktabs}
\usepackage{trimclip}
\usepackage[dvipsnames]{xcolor}

\usepackage[capitalise]{cleveref}
\Crefname{algocf}{Algorithm}{Algorithms}
\crefname{figure}{Figure}{Figures}
\crefalias{AlgoLine}{line}

\usepackage[linesnumbered,ruled,noend,lined]{algorithm2e}
\usepackage{stmaryrd}

\usepackage{tikz}
\usetikzlibrary{arrows,decorations,decorations.pathmorphing,decorations.pathreplacing,shapes,calc,automata}
\usetikzlibrary{arrows.meta}
\usetikzlibrary{backgrounds}
\usetikzlibrary{calc}
\usetikzlibrary{fit}
\usetikzlibrary{shapes.multipart}

\newcommand{\hide}[1]{}

\newcommand{\ol}[1]{\textcolor{blue}{\ifmmode \text{[OL: #1]}\else [OL: #1] \fi}}
\newcommand{\vh}[1]{\textcolor{orange}{\ifmmode \text{[VH: #1]}\else [VH: #1] \fi}}
\iftrue
\usepackage{todonotes}
\newcommand{\at}[1]{\todo[inline,color=teal!10,caption={AT}]{\textbf{AT:} #1}}
\else
\newcommand{\at}[1]{}
\fi

\newcommand{\nat}{\mathbb{N}}
\newcommand{\bbB}{\mathbb{B}}
\newcommand{\vars}{\mathbb{X}}
\newcommand{\varsV}{\mathbb{V}}
\newcommand{\svars}{\Sigma_{\vars}}
\newcommand{\wordvars}{\mathbb{W}}
\newcommand{\posvars}{\mathbb{P}}
\newcommand{\strend}[0]{\$}

\newcommand{\A}[0]{\mathcal{A}}
\newcommand{\weq}[0]{\mathit{eq}}
\newcommand{\seq}[0]{\wedge\mathit{eq}}
\newcommand{\sqeq}[0]{\wedge\mathit{eq}}
\newcommand{\sct}[0]{\wedge\vee\mathit{eq}}
\newcommand{\sweq}[0]{\Phi}

 \newcommand{\lhs}[0]{\mathsf{t}_{\ell}}
 \newcommand{\rhs}[0]{\mathsf{t}_{r}}
\newcommand{\deli}[0]{\#}
\newcommand{\msoalph}[0]{\Sigma_{\vars}}
\newcommand{\msoalphpad}[0]{\Sigma_{\vars,\pad}}

\newcommand{\encode}[0]{\mathsf{encode}}

\newcommand{\eqencode}[0]{\mathsf{eqencode}}

\newcommand{\encodeof}[1]{\encode(#1)}

\newcommand{\eqencodeof}[1]{\eqencode(#1)}

\newcommand{\rewrite}[0]{\mathop{\rightarrowtail}}
\newcommand{\rewriteto}[2]{#1 \rewrite #2}
\newcommand{\xtoyx}[0]{\rewriteto x {yx}}
\newcommand{\ytozy}[0]{\rewriteto y {zy}}
\newcommand{\ytoay}[0]{\rewriteto y {ay}}
\newcommand{\xsubyx}[0]{x\mapsto y x}
\newcommand{\xsubalphax}[0]{x\mapsto \alpha x}
\newcommand{\emptyword}[0]{\epsilon}
\newcommand{\xsubepsilon}[0]{x\mapsto\emptyword}
\newcommand{\xtoax}[0]{\rewriteto x {ax}}
\newcommand{\xtoalphax}[0]{\rewriteto x {\alpha x}}
\newcommand{\ytoalphay}[0]{\rewriteto y {\alpha y}}
\newcommand{\ytoxy}[0]{\rewriteto y {xy}}

\newcommand{\xtoepsilon}[0]{\rewriteto x \emptyword}

\newcommand{\ytoepsilon}[0]{\rewriteto y \emptyword}
\newcommand{\trim}[0]{\mathit{trim}}
\newcommand{\step}[0]{\mathit{step}}
\newcommand{\pass}[0]{\mathit{skip}}
\newcommand{\msoint}[0]{\sigma}
\newcommand{\lenofint}[1]{|#1|}
\newcommand{\ordered}[1]{\mathit{ordered}(#1)}

\newcommand{\alleq}[3]{\mathit{alleq}_{#1}^{#2}(#3)}
\newcommand{\substof}[3]{#1[#2 \mapsto #3]}
\newcommand{\forallp}[0]{\forall^\posvars}
\newcommand{\forallw}[0]{\forall^\wordvars}
\newcommand{\existsp}[0]{\exists^\posvars}
\newcommand{\existsw}[0]{\exists^\wordvars}
\newcommand{\sing}[0]{\mathit{Sing}}
\newcommand{\occur}[3]{\mathit{occur}_{#1}^{#2}(#3)}

\newcommand{\pad}[0]{{\scriptscriptstyle\Box}}
\newcommand{\padfnc}[0]{\mathsf{pad}}

\newcommand{\twotrackpadsmall}[0]{\mbox{\scalebox{0.8}{$\twotrack \pad \pad$}}}

\newcommand{\false}{\mbox{$\mathsf{false}$}}
\newcommand{\model}{\mathit{I}}

\newcommand{\transduct}[0]{\tau}
\newcommand{\transductof}[1]{\transduct_{#1}}
\newcommand{\bigtransof}[1]{\T_{#1}}
\newcommand{\compose}[0]{\mathop{\circ}}
\newcommand{\concat}[0]{.}
\newcommand{\substrof}[2]{#1[#2\!:]}

\newcommand{\T}{\mathcal{T}}

\newcommand{\I}{\mathcal{I}}
\newcommand{\D}{\mathcal{D}}

\newcommand{\lang}[0]{\mathcal{L}}
\newcommand{\langof}[1]{\lang(#1)}

\newcommand{\twotrack}[2]{{\textstyle{{#1}\choose{#2}}}}
\newcommand{\occurof}[2]{|#2|_{#1}}

\newcommand{\encs}[0]{E}
\newcommand{\encsof}[1]{\encs_{#1}}
\newcommand{\smallencsof}[1]{e_{#1}}
\newcommand{\invar}[0]{\mathit{inv}}

\newcommand{\sidesep}[0]{\tikz[baseline,anchor=base,inner sep=-0.4mm]{\node[draw,circle] {=\vphantom{ql}};}}

\newcommand{\limpl}[0]{\mathrel{\rightarrow}}

\newcommand{\defiff}[0]{\mathrel{\triangleq}}

\newcommand{\msostr}[0]{\textsc{MSO(Str)}\xspace}

\newcommand{\runto}[0]{\mathrel{\leadsto}}

\newcommand{\retro}{\textsc{Retro}\xspace}
\newcommand{\sloth}{\textsc{Sloth}\xspace}
\newcommand{\slog}{\textsc{Slog}\xspace}
\newcommand{\slent}{\textsc{Slent}\xspace}
\newcommand{\norn}{\textsc{Norn}\xspace}
\newcommand{\zthreestrfour}{\textsc{Z3str4}\xspace}
\newcommand{\trau}{\textsc{Trau}\xspace}
\newcommand{\ostrich}{\textsc{Ostrich}\xspace}

\newcommand{\bintrackin}[2]{\resizebox{!}{3.5mm}{$\Big[\!\!\begin{array}{c} #1\\[-1mm] #2\end{array}\!\!\Big]\!$}}
\newcommand{\bintrackfst}[4]{\resizebox{!}{3.5mm}{$\begin{array}{r}#1\colon\\[-1mm]#2\colon\end{array}$}\!\!\!\bintrackin{#3}{#4}}
\newcommand{\bintrackelip}[8]{\resizebox{!}{3.5mm}{$\Big[\!\!\begin{array}{cccccc}#1&\!\!\cdots\!\!&#2&\!\!#3&\!\!\cdots\!\!&#4\\[-1mm] #5\!&\!\!\cdots\!\!&#6&\!\!#7&\!\!\cdots\!\!&#8\end{array}\!\!\Big]\!$}}
\newcommand{\lensep}[0]{\#_{\mathit{len}}}
\newcommand{\len}[0]{\mathit{len}}

\newcommand{\psireg}{\psi^{\mathit{reg}}}
\newcommand{\filter}{\mathit{filter}}

\newcommand{\union}{\mathit{union}}

\newcommand{\set}{\mathit{set}}
\newcommand{\ordset}{\mathit{ordSet}}
\newcommand{\expfunc}{\mathit{exp}}
\newcommand{\expfuncof}[2]{\expfunc_{#1,#2}}

\newcommand{\allsymreg}{\Gamma_{\vars, \A_r}}
\newcommand{\allsymregx}{\Gamma_{\vars\setminus\{x\}, \A_r}}

\newcommand{\regsep}[0]{\#_{\mathit{reg}}}
\newcommand{\reg}[0]{\mathit{reg}}

\newcommand{\NielsenTransformationST}[2]{\dfrac{#1}{#2}}

\newcommand{\kepler}[0]{\texttt{Kepler}$_{22}$\xspace}
\newcommand{\pyexhard}[0]{\textsc{PyEx-Hard}\xspace}

\newcommand{\trimtrans}[0]{(trim)\xspace}
\newcommand{\transclof}[1]{#1^{rt}}
\newcommand{\nstepof}[2]{#1^{(#2)}}
\newcommand{\regsym}[1]{\langle #1 \rangle}

\newcommand{\cut}[0]{\mathcal{C}}
\newcommand{\autof}[1]{\A_{#1}}
\newcommand{\ltr}[1]{\labeledto{#1}}
\newcommand{\langinof}[2]{\lang_{#1}(#2)}

\newcommand{\nil}{\mathsf{noop}}
\newcommand{\saturate}{\mathit{satur}}
\newcommand{\saturatenew}{\saturate^\bullet}

\newcommand{\lsbf}[0]{\mathit{LSBF}}
\newcommand{\lsbfp}[0]{\mathit{LSBF_p}}
\newcommand{\lsbfof}[1]{\lsbf(#1)}

\newcommand{\clNP}[0]{\textsc{NP}\xspace}

\newcommand{\zstrre}{\textsc{Z3Str3RE}\xspace}

\renewcommand{\twotrack}[2]{\bintrackin{#1}{#2}}

\newtheorem{theorem}{Theorem}
\newtheorem{proposition}[theorem]{Proposition}
\newtheorem{lemma}[theorem]{Lemma}

\newcounter{example}
\newenvironment{example}[1][]{\refstepcounter{example}\par\medskip
   \noindent \textbf{Example~\theexample. #1} \rmfamily}{\medskip}

\makeatletter
\DeclareRobustCommand{\shortto}{%
  \mathrel{\mathpalette\short@to\relax}%
}

\DeclareRobustCommand{\shortminus}{%
  \mathrel{\mathpalette\short@minus\relax}%
}

\newcommand{\short@to}[2]{%
  \mkern2mu
  \clipbox{{.5\width} 0 0 0}{$\m@th#1\vphantom{+}{\rightarrow}$}%
}

\newcommand{\short@minus}[2]{%
  \mkern2mu
  \clipbox{{.5\width} 0 0 0}{$\m@th#1\vphantom{+}{-}$}%
}
\makeatother

\newcommand{\labeledto}[1]{\raisebox{-0.2pt}{\scalebox{1.2}{\ensuremath{{\shortminus}\hspace{-2.1pt}\raisebox{0.16ex}{$\scriptstyle\{#1\hspace{-0.28pt}\}$}\hspace{-2.4pt}{\shortto}}}}}

\newcommand{\setnocond}[1]{\{#1\}}

\renewcommand{\epsilon}{\varepsilon}

\newcommand{\aut}[1][A]{\mathcal{#1}}
\newcommand{\transducer}{\aut[T]}

\newcommand{\seTransform}[2]{{#1}\hookrightarrow{#2}}

\tikzset{every picture/.style={->,shorten >=0pt, >=stealth', cap=rect, auto}}

\tikzstyle{update}=[font=\scriptsize, right, very near end]
\tikzstyle{state}=[draw, shape=ellipse, initial text={}]
\tikzstyle{loop}+=[min distance=10mm]
\tikzstyle{loop above}=[above, out=105, in=75, loop]
\tikzstyle{loop left}=[left, out=165, in=195, loop]
\tikzstyle{loop below}=[below, out=255, in=285, loop]

\title{A Symbolic Algorithm for the Case-Split Rule in Solving Word
Constraints with Extensions\\ (Technical Report)}
\author[1]{Yu-Fang Chen}
\ead{yfc@iis.sinica.edu.tw}
\affiliation[1]{organization={Institute of Information Science, Academia Sinica},
                addressline={128~Academia Road, Section~2, Nangang},
                postcode={11500},
                city={Taipei},
                country={Taiwan}}

\author[2]{Vojt\v{e}ch Havlena}
\ead{ihavlena@fit.vutbr.cz}

\author[2]{Ond\v{r}ej Leng\'{a}l\corref{cor1}}
\ead{lengal@fit.vutbr.cz}
\affiliation[2]{organization={Faculty of Information Technology, Brno University of Technology},
                addressline={Bozetechova~2},
                postcode={61200},
                city={Brno},
                country={Czech Republic}}

\author[3,4]{Andrea Turrini}
\ead{turrini@ios.ac.cn}

\affiliation[3]{organization={State Key Laboratory of Computer Science, Institute of Software, Chinese Academy of Sciences},
                addressline={Haidian District, Zhongguancun 4\# South Fourth Street},
                postcode={100190},
                city={Beijing},
                country={China}}

\affiliation[4]{organization={Institute of Intelligent Software},
                addressline={221, Nansha Street West},
                postcode={511458},
                city={Guangzhou},
                country={China}}

\cortext[cor1]{Corresponding author}

\begin{document}

\begin{abstract}
  \emph{Case split} is a core proof rule in current decision procedures for the
  theory of string constraints.
  Its use is the primary cause of the state space explosion in string
  constraint solving, since it is the only rule that creates branches in the
  proof tree.
  Moreover, explicit handling of the \emph{case split} rule may cause recomputation of
  the same tasks in multiple branches of the proof tree.
  In this paper, we propose a~symbolic algorithm that significantly reduces
  such a~redundancy.
  In particular, we encode a~string constraint as a~regular language and proof
  rules as rational transducers.
  This allows us to perform similar steps in the proof tree only once, alleviating
  the state space explosion.
  We also extend the encoding to handle arbitrary Boolean combinations of
  string constraints, length constraints, and regular constraints.
  In our experimental results, we validate that our technique
  works in many practical cases where other state-of-the-art solvers fail to
  provide an answer;
  our Python prototype implementation solved over 50\,\% of string constraints
  that could not be solved by the other tools.
\end{abstract}


\maketitle

%%%%%%%%%%%%%%%%%%%%%%%%%%%%%%%%%%%%%%%%%%%%%%%%%%%%%%%%%%%%%%%%%%%%%%%%%%%%%%%%
\vspace{-0.0mm}
\section{Introduction} \label{section:introduction}
\vspace{-0.0mm}
%%%%%%%%%%%%%%%%%%%%%%%%%%%%%%%%%%%%%%%%%%%%%%%%%%%%%%%%%%%%%%%%%%%%%%%%%%%%%%%%

Constraint solving is a~technique used as an enabling technology in many areas
of formal verification and analysis, such as
symbolic execution (\cite{CadarGPDE06Exe,GodefroidKS05DART,King76SymbolicExecution,SenKBG13Jalangi}),
static analysis (\cite{8115706,gulwani2008program}), or
synthesis (\cite{gulwani2011,osera2019,knoth2019}).
For instance, in symbolic execution, feasibility of a~path in a~program is
tested by creating a~constraint that encodes the evolution of the values of
variables on the given path and checking if it is satisfiable.
Due to the features used in the analyzed programs,
checking satisfiability of the constraint can be a complex task.
For instance, the solver has to deal with different data types, such as
Boolean, Integer, Real, or String.
Theories for the first three data types are well known, widely developed, and
implemented in tools, while the theory for the String data type has started to
be investigated only
recently in~\cite{abdulla2014string,BerzishGZ17Z3str3,BjornerTV09PathFeasibility,ChenHLRW19DecisionProcedures,ChenCHLW18WhatIsDecidable,HolikJLRV18,LinM18QuadraticWordEquations,LiangRTBD14DPLLStrings,WangTLYJ16StringAutomata,YuABI14AutomataString,AbdullaACDHRR17FlattenConquer,abdulla2015norn,KiezunGAGHE12HAMPI,lin2016string,berzish2021smt,reynolds2019high,blotsky2018stringfuzz,stanford2021symbolic,loring2019sound,trinh2020inter,chen2019decision},
despite having been considered already by A.\@ A.\@ Markov in the late 1960s in
connection with Hilbert's 10th
problem (see, e.g., \cite{Matiyasevich68Hilbert,DurnevZ09EquationsFreeGroups,Kosovskii76EquationsFreeGroups}).

Most current decision procedures for string constraints involve the so-called
\emph{case-split} rule.
This rule performs a~case split with respect to the possible alignment of the variables.
The case-split rule is used in most, if not all, (semi-)decision procedures for
string constraints, including Makanin's algorithm in~(\cite{makanin1977problem}), Nielsen
transformation~(\cite{Nielsen17}) (also known as the Levi's lemma~(\cite{Levi44})), and the
procedures implemented in most state-of-the-art solvers such as
Z3~(\cite{BjornerTV09PathFeasibility}), CVC4~(\cite{LiangRTBD14DPLLStrings}),
Z3Str3~(\cite{BerzishGZ17Z3str3}), Norn~(\cite{abdulla2014string}), and many more.
In this paper, we will explain the general idea of our symbolic approach using 
the Nielsen transformation, which is the simplest of the approaches;
nonetheless, we believe that the approach is applicable also to other procedures.

Consider the \emph{word equation} $xz = yw$, the primary type of \emph{atomic
string constraints} considered in this paper, where $x$, $z$, $y$, and~$w$ are
\emph{word variables}.
When establishing satisfiability of the word equation,
the Nielsen transformation (introduced in~\cite{Nielsen17}) proceeds by first performing a~case
split based on the possible alignments of the variables~$x$ and~$y$, the first
symbol of the left and right-hand sides of the equation, respectively.
More precisely, it reduces the satisfiability problem for $xz = yw$ into
satisfiability of (at least) one of the following four cases
\begin{inparaenum}[(1)]
\item
	$y$~is a prefix of $x$,\label{itm:NTxgeqy}
\item
	$x$~is a prefix of $y$,\label{itm:NTygeqx}
\item
	$x$~is an empty string, and\label{itm:NTxzero}
\item
	$y$~is an empty string.\label{itm:NTyzero}
\end{inparaenum}
Note that these cases are not disjoint:
for instance, the empty string is a prefix of every variable.
For these cases, the Nielsen transformation generates the following
equations.

For the case~\eqref{itm:NTxgeqy}, i.e., $y$ is a~prefix of~$x$, all occurrences
of~$x$ in $xz = yw$ are replaced with~$yx'$,
where $x'$ is a~fresh word variable (we denote this case as
$\seTransform{x}{yx'}$), i.e., we obtain the equation $yx'z = yw$, which can be
simplified to $x'z = w$.
In fact, since the transformation $\seTransform{x}{yx'}$ removes all occurrences
of the variable~$x$, we can just reuse the variable~$x$ and perform the
transformation $\seTransform{x}{yx}$ instead (and take this into account when
constructing a~model later).

Case~\eqref{itm:NTygeqx} of the Nielsen transformation is just a~symmetric
counterpart of case~\eqref{itm:NTxgeqy} discussed above.
For cases~\eqref{itm:NTxzero} and~\eqref{itm:NTyzero}, $x$ and $y$,
respectively, are replaced by empty strings.
Taking into account all four possible transformations of the equation $xz =
yw$, we obtain the following equations:
\begin{align*}
	\eqref{itm:NTxgeqy}\ xz & = w & \eqref{itm:NTygeqx}\ z & = yw & \eqref{itm:NTxzero}\ z & = yw & \eqref{itm:NTyzero}\ xz & = w
\end{align*}
(Note that the results for \eqref{itm:NTxgeqy} and \eqref{itm:NTyzero} coincide, as well as the results for \eqref{itm:NTygeqx} and \eqref{itm:NTxzero}.)
If $xz = yw$ has a solution, then at least one of the above equations has a~solution, too.
The Nielsen transformation keeps applying the transformation rules on the obtained
equations, building a~proof tree and searching for a~tautology of the form
$\emptyword = \emptyword$.

\begin{figure}[t]
	\centering
\scalebox{0.8}{
\hspace*{-2mm}
\begin{tikzpicture}
\tikzstyle{node_style} = [rectangle,draw,rounded corners=.5cm, minimum height=1.2cm]

\node[node_style] (s_0) at (-1,0) {\begin{tabular}{c} $xz = ab$ \\ $wabyx = awbzy$ \end{tabular}};
\node[node_style] (s_1) at (3,1) {\begin{tabular}{c} $xz = b$ \\ $wabyax = awbzy$ \end{tabular}};
\node[node_style] (s_2) at (3,-1) {\begin{tabular}{c} $z = ab$ \\ $waby = awbzy$ \end{tabular}};

\node[node_style] (s_3) at (7,2) {\begin{tabular}{c} $xz = \emptyword$ \\ $wabyabx = awbzy$ \end{tabular}};
\node[node_style] (s_4) at (7, 0) {\begin{tabular}{c} $z = b$ \\ $wabya = awbzy$ \end{tabular}};

\node[node_style] (s_5) at (11,2) {$wabyab = awby$};

\node[node_style] (s_6) at (11, 0) {$wabya = awbby$};
\node[node_style] (s_7) at (11,-2) {$waby = awbaby$};

\draw (s_0) --++ (0,1) -- node[above] {$\seTransform x {ax}$}  (s_1);
\draw (s_0) --++ (0,-1) -- node[above] {$\seTransform x \emptyword$}  (s_2);
\draw (s_1) --++ (0,1) -- node[above] {$\seTransform x {bx}$}  (s_3);
\draw (s_1) --++ (0,-1) -- node[above] {$\seTransform x \emptyword$}  (s_4);
\draw (s_2) --++ (0,-1) -- node[above] {$\seTransform z  {az} ; \seTransform z {bz};\seTransform z \emptyword$}  (s_7);

\draw (s_3) edge node[above] {$\begin{array}{c}\seTransform z \emptyword \\
  \seTransform x \emptyword \end{array}$} (s_5);
\draw (s_4) edge node[above] {$\begin{array}{c}\seTransform z {bz}\\ \seTransform z \emptyword \end{array}$} (s_6);

\end{tikzpicture}}
\caption{A~partial proof tree of applying the Nielsen transformation on the
  string constraint $xz = ab
  \land wabyx = awbzy$. The leaves are the outcome of completely processing the first word
  equation $xz = ab$. Branches leading to contradictions are omitted.}
\label{fig:nielsen_problem}
\end{figure}
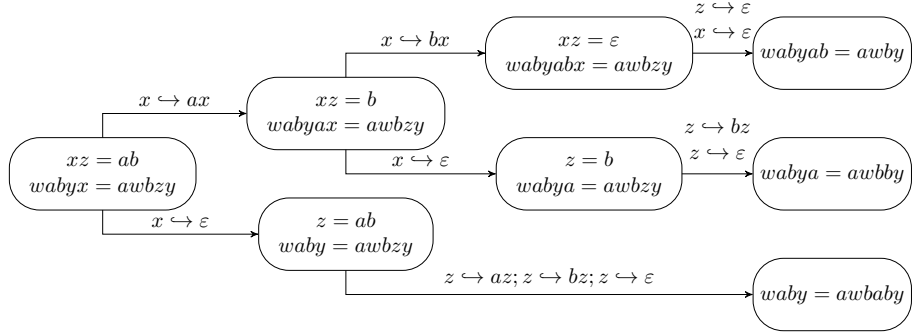
Treating each of the obtained equations separately can cause some redundancy
(as we could already see above).
Let us consider the example in \cref{fig:nielsen_problem}, where we apply
the Nielsen transformation to solve the string constraint $xz = ab \land wabyx =
awbzy$, where $x$, $z$, $w$, and $y$ are word variables and~$a$ and~$b$ are
constant symbols.
After processing the first word equation $xz = ab$, we obtain a~proof tree with
three very similar leaf nodes $wabyab = awby$, $wabya = awbby$, and $waby =
awbaby$, which share the prefixes $waby$ and $awb$ on the left and right-hand
side of the equations, respectively.
If we continue applying the Nielsen transformation on the three leaf nodes, we will
create three very similar subtrees, with almost identical operations.
In particular, the nodes near the root of such subtrees, which transform
$waby\dotso = awb\dotso$, are going to be essentially the same.
The resulting proof trees will therefore start to differ only after processing
such a~common part.
Therefore, handling those equations separately will cause that some operations
will be performed multiple times.
If~the proof tree of each word equation has $n$~leaves and the string
constraint is a conjunction of $k$~word equations, we might need to create~$n^k$
similar subtrees.

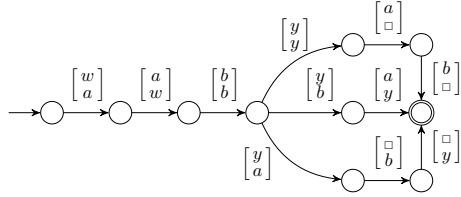
\begin{figure}[t]
	\centering
\scalebox{0.9}{
	\begin{tikzpicture}
    [node distance=1cm]

	\node[state, initial] (s_0) at (0,0) {};
	\node[state] (s_1) [right of=s_0] {};
	\node[state] (s_2) [right of=s_1] {};
	\node[state] (s_3) [right of=s_2] {};
	\node[state] (s_5) [right of=s_3,xshift=4mm] {};
	\node[state] (s_4) [above of=s_5] {};
	\node[state] (s_6) [below of=s_5] {};
	\node[state] (s_7) [right of=s_4] {};
	\node[accepting,state] (s_8) [right of=s_5] {};
	\node[state] (s_9) [right of=s_6] {};
	\draw (s_0) to node[above] {$\twotrack{w}{a}$} (s_1);
	\draw (s_1) to node[above] {$\twotrack{a}{w}$} (s_2);
	\draw (s_2) to node[above] {$\twotrack{b}{b}$} (s_3);
	\draw (s_3) to[bend left] node[above] {$\twotrack{y}{y}$} (s_4);
	\draw (s_3) to node[above,xshift=2mm] {$\twotrack{y}{b}$} (s_5);
	\draw (s_3) to[bend right] node[below left,xshift=-2mm,yshift=4mm] {$\twotrack{y}{a}$} (s_6);
	\draw (s_4) to node[above] {$\twotrack{a}{\pad}$} (s_7);
	\draw (s_7) to node[auto] {$\twotrack{b}{\pad}$} (s_8);

	\draw (s_5) to node[auto] {$\twotrack{a}{y}$} (s_8);

	\draw (s_6) to node[above] {$\twotrack{\pad}{b}$} (s_9);
	\draw (s_9) to node[right] {$\twotrack{\pad}{y}$} (s_8);

	\end{tikzpicture}}
	\caption{A~finite automaton encoding the three equations $wab yab = awb y$, $wab ya = awb by$, and $wab y = awb aby$.}
	\label{fig:fsa_leaves}
\end{figure}
The case split can be performed more efficiently if we process the common part
of the said leaves together using a~symbolic encoding.
In this paper, we use an encoding of a~set of equations as a~regular language,
which is represented by a~\emph{finite automaton}.
An example is given in \cref{fig:fsa_leaves},
which shows a~finite automaton over a~2-track alphabet, where each of the two
tracks represents one side of the equation.
For instance, the equation $wabyab = awby$ is represented by the word
$\twotrack{w}{a}\twotrack{a}{w}\twotrack{b}{b}\twotrack{y}{y}\twotrack{a}{\pad}\twotrack{b}{\pad}$
accepted by the automaton, where the $\pad$ symbol is a~padding used to make
sure that both tracks are of the same length.

Given our regular language-based symbolic encoding, we need a~mechanism to
perform the Nielsen transformation steps on a~set of equations encoded as
a~regular language.
We show that the transformations can be encoded as \emph{rational relations},
represented using \emph{finite transducers}, and the whole satisfiability
checking problem can be encoded within the framework of \emph{regular model
checking}.
We will provide more details on how this is done in
\cref{section:quadratic,section:conjunction,section:full,sec:rmc-extension}
stepwise.
In \cref{section:quadratic}, we describe the approach for a~simpler case where
the input is a~\emph{quadratic word equation}, i.e., a~word equation with at
most two occurrences of every variable.
In this case, the Nielsen transformation is sound and complete, that is, the solution it returns is correct and it returns a solution whenever a~solution exists.
In \cref{section:conjunction}, we extend the technique to support
the \emph{conjunction} of \emph{non-quadratic} word equations.
In \cref{section:full}, we extend our approach to support arbitrary Boolean
combination of string constraints.
\cref{sec:rmc-extension} extends our framework with two additional types of
atomic string constraints---\emph{length} and \emph{regular constraints}---which
can constrain the length of values assigned to word variables and their
membership in a~regular language, respectively.

We have implemented our approach in a~prototype Python tool called \retro and
evaluated its performance on two benchmark sets:
\kepler obtained from~\cite{LeH18} and \pyexhard obtained by running the PyEx
symbolic execution engine on Python programs from \cite{reynolds2017scaling} and
collecting examples on which CVC4 or Z3 fail.
\retro solved most of the problems in \kepler (on which CVC4 and Z3 do not
perform well).
Moreover, it solved over 50\,\% of the benchmarks in \pyexhard that could be
solved by neither CVC4 nor Z3.

This paper is an extended version of the paper that appeared in the proceedings
of APLAS'20 (\cite{ChenHLT20}), containing complete proofs of the presented
lemmas and theorems and further extending the presented technique to handle
\begin{inparaenum}[(i)]
  \item  arbitrary Boolean combination of string constraints (\cref{section:full}),
  \item  length constraints (\cref{sec:length-constraints}), and
  \item  regular constraints (\cref{sec:regular-constraints}).
\end{inparaenum}

%%%%%%%%%%%%%%%%%%%%%%%%%%%%%%%%%%%%%%%%%%%%%%%%%%%%%%%%%%%%%%%%%%%%%%%%%%%%%%%%
\vspace{-0.0mm}
\section{Preliminaries}
\label{sec:preliminaries}
\vspace{-0.0mm}
%%%%%%%%%%%%%%%%%%%%%%%%%%%%%%%%%%%%%%%%%%%%%%%%%%%%%%%%%%%%%%%%%%%%%%%%%%%%%%%%
An \emph{alphabet} $\Sigma$ is a finite set of \emph{characters} and a
\emph{word} over~$\Sigma$ is a sequence $w = a_{1} \ldots a_{n}$ of characters
from $\Sigma$, with $\emptyword$ denoting the \emph{empty word}.
We use $w_{1} \concat w_{2}$ (and often just $w_1 w_2$) to denote the
\emph{concatenation} of words $w_{1}$ and $w_{2}$.
$\Sigma^{*}$ is the set of all words over~$\Sigma$, $\Sigma^{+} = \Sigma^{*}
\setminus \setnocond{\emptyword}$, and $\Sigma_{\emptyword} = \Sigma \cup \setnocond{\emptyword}$.
A~\emph{language} over $\Sigma$ is a~subset~$L$ of~$\Sigma^{*}$.
Given a~word~$w = a_1 \ldots a_n$, we use $|w|$ to denote the length~$n$ of~$w$
and $\occurof a w$ to denote the number of occurrences of the character $a \in
\Sigma$ in~$w$.
Further, we use $w[i]$ to denote $a_{i}$, the $i$-th character of~$w$, and
$\substrof w i$ to denote the word $a_{i} \ldots a_n$.
When $i > n$, the value of $w[i]$ and $\substrof w i$ is in both cases $\bot$,
a~special \emph{undefined} value, which is different from all other
values and also from itself (i.e., $\bot \neq \bot$).
We use $\Sigma^k$ for $k \geq 2$ to denote the \emph{stacked alphabet} consisting
of $k$-tuples of symbols from~$\Sigma$, e.g., $\twotrack a b \in \Sigma^2$ for
$a,b \in \Sigma$.

%------------------------------------------------------------------------------
\subsection{Automata and transducers}

A \emph{(finite) $k$-tape transducer} is a~quintuple $\transducer =
(Q,\Sigma,\Delta,Q_{i},Q_f)$ such that
$Q$ is a~finite set of \emph{states},
$\Sigma$~is an~alphabet,
$\Delta \subseteq Q \times \Sigma_\emptyword^k \times Q$ is a~set of
\emph{transitions} of the form $q \ltr{a^1, \ldots, a^k} s$ for $a^1, \ldots,
a^k \in \Sigma_\emptyword$,
$Q_{i} \subseteq Q$ is a~set of \emph{initial states}, and
$Q_f \subseteq Q$ is a~set of \emph{final states}.
A~run $\pi$ of~$\transducer$ over a~$k$-tuple of words $(w_1, \ldots, w_k)$ is a~sequence
of transitions
$q_0 \ltr{a^1_1, \ldots, a^k_1} q_1$,
$q_1 \ltr{a^1_2, \ldots, a^k_2} q_2$,
$\ldots,
q_{n-1} \ltr{a^1_n, \ldots, a^k_n} q_n \in \Delta$
such that for each $i \in [1,k]$ we have
$w_{i} = a^{i}_1 a^{i}_2 \ldots a^{i}_n$
(note that $a^{i}_m$ can be $\emptyword$, so $w_{i}$ and $w_j$ may be
of a~different length, for $i \neq j$).
The run~$\pi$ is \emph{accepting} if $q_0 \in Q_{i}$ and $q_{n} \in Q_f$, and
a~$k$-tuple $(w_1, \ldots, w_k)$ is \emph{accepted} by~$\transducer$ if there
exists an accepting run of~$\transducer$ over $(w_1, \ldots, w_k)$.
The \emph{language} $\langof \transducer$ of~$\transducer$ is defined as the $k$-ary
relation $\langof \transducer = \{\,(w_1, \ldots, w_k) \in (\Sigma^{*})^k \mid  (w_1, \ldots,
w_k) \text{ is accepted by } \transducer\,\}$.
We call the class of relations accepted by transducers \emph{rational
relations}.
$\transducer$~is \emph{length-preserving} if no transition in $\Delta$ contains
$\emptyword$;
the class of relations accepted by length-preserving transducers
is named as \emph{regular relations}. 
For a 2-tape transducer $\transducer$ and $(w_1, w_2) \in \langof\transducer$,
we denote $w_1$ as an input and $w_2$ as an output of $\transducer$. 
A~\emph{finite automaton} (FA) is a~1-tape finite transducer;
languages accepted by finite automata are called \emph{regular
languages}.
See, e.g., \cite{Pin2021} for more details on automata and transducers.

Given two $k$-ary relations $R_1, R_2$, we define their \emph{concatenation}
$R_1 \concat R_2 = \{\,(u_1v_1, \dots, u_kv_k) \in (\Sigma^{*})^k \mid (u_1,\dots,u_k)\in R_1 \land
(v_1,\dots,v_k)\in R_2 \,\}$ and given two binary relations $R_1, R_2$, we
define their \emph{composition} $R_1 \compose R_2 = \{\,(x, z) \in (\Sigma^{*})^2 \mid \exists y \in \Sigma^{*}\colon
(x, y) \in R_2 \land (y,z) \in R_1\,\}$.
Given a~$k$-ary relation $R$ we define $R^0 = \{\emptyword\}^k$, $R^{i+1} =
R.R^{i}$ for $i\geq 0$.
\emph{Iteration} of $R$ is then defined as $R^{*} = \bigcup_{i\geq 0}R^{i}$.
Given a~language $\lang \subseteq \Sigma^*$ 
and a~binary relation $R$, we use $R(\lang)$ to denote the language
$\{\,y \in \Sigma^{*} \mid \exists x \in \lang\colon (x, y) \in R\,\}$, called the $R$-\emph{image}
of~$\lang$. We also use $R^{-1}(w)$ to denote the language $\{ u \mid (w,u) \in R \}$, called the \emph{preimage} of a~word~$w$.

\begin{proposition}[\cite{Berstel79}]\label{prop:rat-prop}
  The following propositions hold:
	\begin{enumerate}[(i)]
    \item  The class of binary rational relations is closed under (finite) union,
      composition, concatenation, and iteration;
      it is not closed under intersection and complement.
    \item For a~binary rational relation $R$, a~regular language $\lang$, and a word $w$, the
      languages $R(\lang)$ and $R^{-1}(w)$ are also effectively regular (i.e., they can be
      computed).
    \item The class of regular relations is closed under Boolean operations.
	\end{enumerate}
\end{proposition}

%------------------------------------------------------------------------------
\subsection{String constraints}
Let $\Sigma$ be an alphabet and $\vars$ be a~set of \emph{word variables}
ranging over~$\Sigma^{*}$ s.t. $\vars \cap \Sigma = \emptyset$.
We use $\svars$ to denote the extended alphabet~$\Sigma\cup\vars$.
An \emph{assignment of $\vars$} is a mapping $I \colon \vars\rightarrow\Sigma^{*}$.
A~\emph{word term} is a~string over the alphabet $\svars$.
We lift an assignment~$I$ to word terms by defining $I(\emptyword) = \emptyword$,
$I(a)=a$, and $I(x.w) = I(x).I(w)$, for $a \in\Sigma$, $x \in \svars$, and $w
\in \svars^{*}$.
A~\emph{word equation}~$\varphi_e$ is of the form $t_{1} = t_{2}$ where $t_{1}$
and $t_{2}$ are word terms.
$I$ is a~\emph{model} of~$\varphi_e$ if $I(t_{1})=I(t_{2})$.
We call a~word equation an \emph{atomic string constraint}.
A~\emph{string constraint} is obtained from atomic string constraints using
Boolean connectives ($\wedge, \vee, \neg$), with the semantics defined in the
standard manner.
A~string constraint is \emph{satisfiable} if it has a~model.
Given a~word term $t \in \svars^{*}$, a variable $x \in \vars$, and
a~word term $u \in \svars^{*}$, we use $\substof t x u$ to denote the
word term obtained from~$t$ by replacing all occurrences of~$x$ by~$u$,
e.g.~$\substof{(abxcxy)} x {cy} = abcyccyy$.
We~call a~string constraint~$\psi$ \emph{quadratic} if each variable has at most two
occurences in~$\psi$, and \emph{cubic} if each variable has at most three
occurences in~$\psi$.

We use the following terminology.
Let~$\Phi$ be a~string constraint.
A~\mbox{(semi-)algo}\-ritm~\textbf{A} for solving string constraints is
\begin{enumerate}
  \item  \emph{sound} if it holds that if~\textbf{A} returns an assignment $I$
    to the variables of~$\Phi$, then~$I$ is a~model of~$\Phi$,
  \item  \emph{complete} if it holds that if~$\Phi$ is satisfiable,
    then~\textbf{A} returns a~model of~$\Phi$ in a~finite number of steps, and
  \item  \emph{terminating} if it holds that~\textbf{A} always returns an
    assignment or $\false$ in a~finite number of steps.
\end{enumerate}

%*******************************************************************************
\vspace{-0.0mm}
\subsection{Monadic Second-Order Logic on Strings (\msostr)}
\label{ssec:mso_str}
\vspace{-0.0mm}
%*******************************************************************************

We define \emph{monadic second-order logic on strings} (\msostr) (\cite{buchi1960weakSOAautomata}) over the
alphabet $\Gamma$ as follows.
Let $\wordvars$ be a~countable set of \emph{string variables} whose values range over
$\Gamma^{*}$ and
$\posvars$ be a~countable set of \emph{set (second-order) position variables} whose
values range over finite subsets of $\nat_1 = \nat \setminus \{0\}$ such that $\wordvars \cap \posvars =
\emptyset$.
A~formula~$\varphi$ of~\msostr is defined as
\begin{align*}
  \varphi ::=
  P \subseteq R \mid
  P = R + 1 \mid
  w[P] = a \mid
  \varphi_{1} \land \varphi_{2} \mid
  \neg \varphi \mid
  \forallp P (\varphi) \mid
  \forallw w (\varphi)
\end{align*}
where $P,R \in \posvars$,
$w \in \wordvars$, and
$a \in \Gamma$.
We use $\varphi(w_1, \ldots, w_k)$ to denote that the free variables of $\varphi$
are contained in $\{w_1, \ldots, w_k\}$.

\begin{figure}[t]
\centering
\begin{tabular}{l@{~~}l}
  $\msoint \models P \subseteq R$ &
  iff $\msoint(P)$ is a subset of $\msoint(R)$
  \\
  $\msoint \models P = R + 1$ &
  iff $\msoint(P) = \{r + 1 \mid r \in \msoint(R)$ and $r + 1 \leq \lenofint
  \msoint\}$
  \\
  $\msoint \models w[P] = a$ &
  iff for all $p \in P$ it holds that $\msoint(w)[p]$ is $a$
  \\
  $\msoint \models \varphi_{1} \land \varphi_{2}$ &
  iff $\msoint \models \varphi_{1}$ and $\msoint \models \varphi_{2}$
  \\
  $\msoint \models \neg \varphi$ &
  iff not $\msoint \models \varphi$
  \\
  $\msoint \models \forallp P (\varphi)$ &
  iff for all $v \subseteq \{1, \ldots, \lenofint \msoint\}$ it
  holds that $\substof \msoint P v \models \varphi$
  \\
  $\msoint \models \forallw w (\varphi)$ &
  iff for all $v \in \Gamma^{\lenofint \msoint}$ it
  holds that $\substof \msoint w v \models \varphi$
\end{tabular}
\caption{Semantics of \msostr}
\label{tab:semantics}
\end{figure}

\begin{figure}[t]
	\centering
\begin{align*}
  \varphi \lor \psi ~~\defiff~~ {}&
   \neg (\varphi \land \psi)
   &
  \existsp P (\varphi) ~~\defiff~~ {}&
  \neg \forallp P (\neg \varphi)
  \\
  \varphi \limpl \psi ~~\defiff~~ {}&
   \neg \varphi \lor \psi
  &
  \existsw w (\varphi) ~~\defiff~~ {} &
  \neg \forallw w (\neg \varphi)
  \\
  P = R ~~\defiff~~ {} &
  P \subseteq R ~\land~ R \subseteq P
  &
  P = \emptyset ~~\defiff~~ {} &
  \forallp R (P \subseteq R)
  \\
  \sing (P) ~~\defiff~~ {} &
  \neg (P = \emptyset) ~\land~ \forallp R(R \subseteq P
  \limpl (R = \emptyset ~\lor~ R = P))
  \span\omit\span\omit
  \\
  p \in R ~~\defiff~~ {} &
   \sing(p) \land p \subseteq R
  \\
  p \leq r ~~\defiff~~ {} &
  \forallp T((p \in T  \land \forallp u (u \in T
  \limpl \existsp v(v = u+1 \land v \in T))) \limpl r \in T)
  \span\omit\span\omit
  \\
  p < r ~~\defiff~~ {} &
  p \leq r \land \neg (p = r)
  \\
  x = 1 ~~\defiff~~ {} &
  \forallp u(u \leq x \rightarrow u = x)
  &
  x = \strend ~~\defiff~~ {} &
  \forallp u(x \leq u \rightarrow u = x)
  \\
  w_1[P] = w_2[R] ~~\defiff~~ {} &
  \bigvee_{a \in \Gamma} (w_1[P] = a \land w_2[R] = a)
  \span\omit\span\omit
\end{align*}
\begin{flushleft}
for all $j \geq 1$ and any term $\gamma$:
\end{flushleft}
\begin{align*}
  P = R + j ~~\defiff~~ {} &
  \existsp S_1 \ldots \existsp S_j(S_1 = P + 1 \land S_2 = S_1 + 1 \land \ldots \land R = S_j + 1)
  \\
  p = j ~~\defiff~~ {} &
  \existsp z (z = 1 \land p = z + (j - 1))
  \\
  w[P + j] = \gamma ~~\defiff~~ {} &
  \existsp S(S = P + j \land w[S] = \gamma)
  \\
  w[P - j] = \gamma ~~\defiff~~ {} &
  \existsp S(P = S + j \land w[S] = \gamma)
  \\
  w[\strend - j] = \gamma ~~\defiff~~ {} &
  \existsp r\existsp s(r = \strend \land r = s + j \land w[s] = \gamma)
  \\
  w[j] = \gamma ~~\defiff~~ {} &
  \existsp p (p = j \land w[p] = \gamma)
  \\
  p \geq j ~~\defiff~~ {} &
  \existsp s (s = j \land p \geq s)
\end{align*}
\caption{Syntactic sugar for \msostr}
\label{fig:mso_sugar}
\end{figure}
The semantics of \msostr is defined in \cref{tab:semantics}.
An \emph{\msostr variable assignment} is an assignment $\msoint \colon \wordvars \cup \posvars
\to (\Gamma^{*} \cup 2^{\nat_1})$ that respects the types of variables with the
additional requirement that for every $u, v \in \wordvars$ we have $|\msoint(u)|
= |\msoint(v)|$.
(We often omit unused variables in $\msoint$.)
We use $\lenofint{\msoint}$ to denote the value $|\msoint(w)|$ of any $w \in
\wordvars$.
Note that $\lenofint{\msoint}$ is well defined since we assume that all $w \in
\wordvars$ are mapped to strings of the same length. 
The notation $\substof \msoint x v$ denotes a~variant of $\msoint$ where the
assignment of variable $x$ is changed to the value~$v$.

We call an \msostr formula a~\emph{string formula} if it contains no free
position variables.
Such a~formula (with $k$~free string variables) denotes a~$k$-ary relation
over~$\Gamma^{*}$.
In particular, given an \msostr string formula~$\varphi(w_1, \ldots, w_k)$ with
$k$~free string variables $w_1, \ldots, w_k$, we use $\langof \varphi$ to denote
the relation $\{\,(x_1, \ldots, x_k) \in (\Gamma^{*})^k \mid \{w_1 \mapsto x_1,
\ldots, w_k \mapsto x_k\} \models \varphi \,\}$.
In the special case of $k = 1$, $\varphi$ denotes a~language $\langof \varphi
\subseteq \Gamma^{*}$.

\begin{proposition}[\cite{ThatcherW68}]\label{prop:msostr-power}
  The class of languages denoted by \msostr string formulae with $1$~free string
  variable is exactly the class of regular languages.
  Furthermore, the class of relations denoted by \msostr string formulae with
  $k$~free string variables, for $k > 1$, is exactly the class of regular
  relations.
\end{proposition}

%%%%%%%%%%%%%%%%%%%%%%%%%%%%%%%%%%%%%%%%%%%%%%%%%%%%%%%%%%%%%%%%%%%%%%%%%%%%%%%%
\vspace{-0.0mm}
\paragraph{Syntactic sugar for \msostr{}}
\vspace{-0.0mm}
%%%%%%%%%%%%%%%%%%%%%%%%%%%%%%%%%%%%%%%%%%%%%%%%%%%%%%%%%%%%%%%%%%%%%%%%%%%%%%%%

In \cref{fig:mso_sugar}, we define the standard syntactic sugar to
allow us to write more concise \msostr formulae.
Most of the sugar is standard,
let us, however, explain some of the less standard notation:
$\sing(P)$ denotes that~$P$ is a~singleton set of positions,
$p \leq r$ denotes that~$p$ and~$r$ are single positions and that~$p$
is less than or equal to~$r$,
$x = 1$ and $x = \strend$ denote that $x$ is the first and the last position
respectively, and
$P = R+j$ denotes that~$P$ is equal to~$R$ with all positions incremented by~$j$.
We also extend our syntax to allow first-order variables (we abuse
notation and use the same quantifier notation as for second-order variables, but
denote the first-order variable with a~lowercase letter):
\begin{align*}
  \forallp p (\varphi) & {}\defiff \forallp P (\sing(P) \limpl \substof \varphi p P) \\
  \existsp p (\varphi) & {}\defiff \existsp P (\sing(P) \land \substof \varphi p P)
\end{align*}
where $\substof \varphi p P$ denotes the substitution of all free occurrences
of $p$ in $\varphi$ by~$P$.

%%%%%%%%%%%%%%%%%%%%%%%%%%%%%%%%%%%%%%%%%%%%%%%%%%%%%%%%%%%%%%%%%%%%%%%%%%%%%%%%
\vspace{-0.0mm}
\subsection{Nielsen Transformation}
\label{ssec:nielsen}
\vspace{-0.0mm}
%%%%%%%%%%%%%%%%%%%%%%%%%%%%%%%%%%%%%%%%%%%%%%%%%%%%%%%%%%%%%%%%%%%%%%%%%%%%%%%%
\begin{figure}[t]
  \centering
	\begin{tabular}{c}
		$\NielsenTransformationST{\alpha u = \alpha v}{u = v}$ \trimtrans \qquad\qquad
		$\NielsenTransformationST{xu = v}{\substof{u}{x}{\emptyword} = \substof{v}{x}{\emptyword}}$ ($\seTransform{x}{\emptyword}$)
		\\
		\\[1mm]
		$\NielsenTransformationST{xu = \alpha v}{x(\substof{u}{x}{\alpha x}) =
    \substof{v}{x}{\alpha x}}$ ($\seTransform{x}{\alpha x}$) \qquad
	\end{tabular}
	\caption{Rules of the Nielsen transformation, with $x~\in~\vars$,
    $\alpha~\in~\svars$, and $u,v \in \svars^{*}$.
	 Symmetric rules are omitted.}
	\label{tab:NielsenTransformations}
\end{figure}

As already briefly mentioned in the introduction, the Nielsen transformation can be
used to check satisfiability of a~conjunction of word equations.
We use the three rules shown in \cref{tab:NielsenTransformations};
besides the rules $\seTransform{x}{\alpha x}$ and $\seTransform{x}{\emptyword}$ that we
have seen in the introduction, there is also the \trimtrans rule, used to remove
a~shared prefix from both sides of the equation.

Given a system of word equations, multiple Nielsen transformations might be
applicable to it, resulting in different transformed equations on which other
Nielsen transformations can be performed, as shown in
\cref{fig:nielsen_problem}.
Trying all possible transformations generates a tree (or a~graph in general)
whose nodes contain conjunctions of word equations and whose edges are labelled
with the applied transformation.
The conjunction of word equations in the root of the tree is satisfiable if and only if at
least one of the leaves in the graph is a tautology, i.e., it contains
a~conjunction of the form $\emptyword = \emptyword \land \cdots \land
\emptyword = \emptyword$.
As an example, consider the satisfiable equation $xy=ax$ where
$x,y$ are word variables and $a$ is a symbol with the proof graph in
\cref{fig:nielsen-proof-graph}.

\begin{figure}[t]
  \centering
	\begin{tikzpicture}
  [node distance=3cm,->,>=stealth,transform shape, scale=0.8]

  \tikzstyle{state} = [rectangle,draw,rounded corners=.2cm,draw=black,minimum height=0.7cm, inner sep=3pt]

  \node[state, initial, initial text={}] (s0) at (0,0) {$xy = ax$};
  \node[state, right of=s0] (s2) {$y = a$};
  \node[state, yshift=0-10mm, above of=s2] (s3) {$\emptyword = a$};
  \node[state, right of=s2] (s4) {$y = \emptyword$};
  \node[state, right of=s4] (s5) {$\emptyword = \emptyword$};

  \draw (s0) edge [loop above] node[above] {$\seTransform{x}{ax}$} (s0);
  \draw (s0) to node[above] {$\seTransform{x}{\emptyword}$} (s2);
  \draw (s2) to node[right, yshift=0mm] {$\seTransform{y}{\emptyword}$} (s3);
  \draw (s2) to node[above] {$\seTransform{y}{ay}$} (s4);
  \draw (s4) to node[above] {$\seTransform{y}{\emptyword}$} (s5);

  % \draw[orange] (s5) edge [bend right] node[above] {$y = \emptyword$} (s4);
  % \draw[orange] (s4) edge [bend right] node[above] {$y = a$} (s2);
  % \draw[orange] (s2) edge [bend right] node[above,xshift=3mm] {$y = a, x = \emptyword$} (s0);
  % \draw[orange] (s0) edge [loop] node[above] {$y = a, x = a^n$} (s0);

\end{tikzpicture}
  \caption{Proof graph of the equation $xy=ax$ generated by the Nielsen transformation.}
\label{fig:nielsen-proof-graph}
\end{figure}
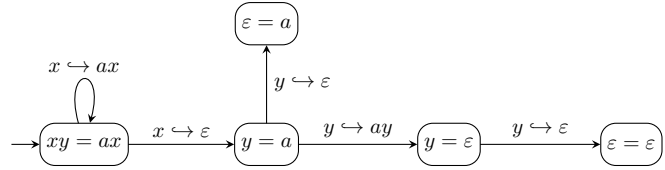

\begin{lemma}[cf.\@ \cite{makanin1977problem,Diekert02Makanin}]
\label{lem:Nielsen}
	The Nielsen transformation is sound and complete. Moreover, if the system of word
	equations is quadratic, the proof graph is finite.
\end{lemma}
\cref{lem:Nielsen} is correct even if we construct the proof tree using the
following strategy:
every application of $\seTransform{x}{\alpha x}$ or $\seTransform{x}{\emptyword}$
is followed by as many applications of the \trimtrans rule as possible.
We use $\rewriteto{x}{\alpha x}$ to denote the application of one
$\seTransform{x}{\alpha x}$ rule followed by as many applications of \trimtrans as
possible, and $\rewriteto{x}{\emptyword}$ for the application of
$\seTransform{x}{\emptyword}$ followed by \trimtrans.

%%%%%%%%%%%%%%%%%%%%%%%%%%%%%%%%%%%%%%%%%%%%%%%%%%%%%%%%%%%%%%%%%%%%%%%%%%%%%%%%
\vspace{-0.0mm}
\subsection{Regular Model Checking}
\label{ssec:rmc}
\vspace{-0.0mm}
%%%%%%%%%%%%%%%%%%%%%%%%%%%%%%%%%%%%%%%%%%%%%%%%%%%%%%%%%%%%%%%%%%%%%%%%%%%%%%%%
\emph{Regular model checking}
(RMC) (cf. \cite{KestenMMPS01,WolperB98,BouajjaniJNT00,Abdulla12,BouajjaniHRV12}) is a~framework for
verifying infinite state systems.
In~RMC, each \emph{system configuration} is represented as a~word
over an~alphabet~$\Sigma$.
The set of \emph{initial configurations} $\I$ and \emph{destination
configurations} $\D$
are captured as regular languages over~$\Sigma$.
The \emph{transition relation} $\T$ is captured as a~binary rational relation over~$\Sigma^{*}$.
A~regular model checking \emph{reachability problem} is represented by the triplet $(\I,\T,\D)$
and asks whether $\transclof{\T} (\I) \cap \D \neq \emptyset$, where $\transclof
\T$ represents the reflexive and transitive closure of~$\T$.
One way how to solve the problem is to start computing the sequence
$\nstepof \T 0 (\I), \nstepof \T 1(\I), \nstepof \T 2(\I), \ldots$
where $\nstepof \T 0(\I) = \I$ and $\nstepof \T {n+1} (\I) = \T(\nstepof \T n
(\I))$.
During the computation of the sequence, we can check whether we find $\nstepof \T i(\I)$
that overlaps with~$\D$, and if yes, we can deduce that $\D$ is reachable.
On the other hand,
if we obtain a~sequence such that $\bigcup_{0 \leq i < n} \T^{i}(\I) \supseteq \T^{n}
(\I)$, we know that we have explored all possible system configurations without
reaching $\D$, so $\D$ is unreachable.
The RMC reachability problem is in general \emph{undecidable} (this can be
easily shown, e.g., by a~reduction from Turing machine configuration
reachability).

%%%%%%%%%%%%%%%%%%%%%%%%%%%%%%%%%%%%%%%%%%%%%%%%%%%%%%%%%%%%%%%%%%%%%%%%%%%%%%%%
\vspace{-0.0mm}
\section{Solving Word Equations using RMC}
\label{section:we_rmc}
\label{section:quadratic}
\vspace{-0.0mm}
%%%%%%%%%%%%%%%%%%%%%%%%%%%%%%%%%%%%%%%%%%%%%%%%%%%%%%%%%%%%%%%%%%%%%%%%%%%%%%%%

In this section, we describe a~symbolic RMC-based framework for solving string
constraints.
The framework is based on encoding a~string constraint into a~regular language
and encoding steps of the Nielsen transformation as a~rational relation.
Satisfiability of a~string constraint is then reduced to a~reachability problem
of RMC.

%*******************************************************************************
\vspace{-0.0mm}
\subsection{Nielsen Transformation as Word Operations}
\label{ssec:nielsenWordTransform}
\vspace{-0.0mm}
%*******************************************************************************

In the following, we describe how the Nielsen transformation of a~single word
equation can be expressed as operations on words.
We view a word equation~$\weq\colon \lhs = \rhs$ as a~pair of word terms $\smallencsof \weq =
(\lhs, \rhs)$ corresponding to the left and right hand sides of the equation respectively;
therefore $\smallencsof \weq \in \msoalph^{*} \times \msoalph^{*}$.
Without loss of generality we assume that $\lhs[1] \neq \rhs[1]$;
if this is not the case, we pre-process the equation by applying the \trimtrans
Nielsen transformation (cf.\@ \cref{tab:NielsenTransformations})
to trim the common prefix of $\lhs$ and $\rhs$.

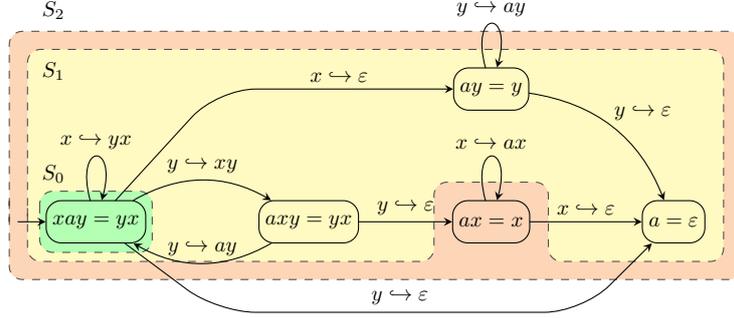
\begin{figure}[t]
	\centering
	\begin{tikzpicture}
  [node distance=3cm,->,>=stealth,transform shape, scale=0.8]

  \tikzstyle{state} = [rectangle,draw,rounded corners=.2cm,draw=black,minimum height=0.7cm, inner sep=3pt]
  \tikzstyle{empty}=[]

  \node[state, initial, initial text={}] (s0) at (0,0) {$xay = yx$};
  \node[state, right of=s0, xshift=5mm] (s2) {$axy = yx$};
  \node[state, right of=s2] (s3) {$ax = x$};
  \node[state, above of=s3,yshift=-8mm] (s4) {$ay = y$};
  \node[state, right of=s3] (s5) {$a = \emptyword$};

  \draw (s0) edge [loop above] node[above] {$\seTransform{x}{yx}$} (s0);
  \draw (s0) edge[bend left] node[above] {$\seTransform{y}{xy}$} (s2);
  \draw (s2) edge[bend left] node[above] {$\seTransform{y}{ay}$} (s0);
  \draw (s2) to node[above] {$\seTransform{y}{\emptyword}$} (s3);
  \draw (s3) edge [loop above] node[above] {$\seTransform{x}{ax}$} (s3);
  \draw[->,rounded corners=5mm] (s0) -- ++(2,2.2) node[above,xshift=20mm]{$\seTransform{x}{\emptyword}$} -- (s4);
  \draw (s4) edge [loop above] node[above] {$\seTransform{y}{ay}$} (s4);
  \draw[->,rounded corners=5mm] (s0) -- ++(2,-1.5) node[above,xshift=30mm]{$\seTransform{y}{\emptyword}$} -- ++(6,0) -- (s5);
  \draw (s4) edge[bend left] node[above,xshift=5mm] {$\seTransform{y}{\emptyword}$} (s5);
  \draw (s3) to node[above] {$\seTransform{x}{\emptyword}$} (s5);

  \begin{pgfonlayer}{background}
    \draw[-,dashed,rectangle,fill=Orange!30,draw=black!70,rounded corners=5pt,inner sep=10pt]
      ([xshift=-6mm]s0.west) |- ([yshift=6mm]s4.north) -| ([xshift=6mm,yshift=-6mm]s5.south east) -| ([xshift=-6mm]s0.north west);

    \draw[-,dashed,rectangle,fill=Yellow!30,draw=black!70,rounded corners=5pt,inner sep=10pt]
      ([xshift=-3mm]s0.west) |- ([yshift=3mm]s4.north) -| ([xshift=3mm,yshift=-3mm]s5.south east) -| ([xshift=3mm]s3.east) |- ([yshift=3mm]s3.north) -|
      ([xshift=-3mm,yshift=-3mm]s3.south west) -| ([xshift=-3mm]s0.north west);

    \draw[-,dashed,rectangle,fill=green!30,draw=black!70,rounded corners=5pt,inner sep=10pt]
      ([xshift=-1mm]s0.west) |- ([yshift=1.5mm]s0.north) -| ([xshift=1mm]s0.east)
      |- ([yshift=-1.5mm,xshift=-1mm]s0.south west) -- cycle;
  \end{pgfonlayer}

  \node[empty, above of=s0, node distance=8mm, text=black, xshift=-7mm] (lab1) {$S_0$};
  \node[empty, above of=s0, node distance=25mm, text=black, xshift=-7mm] (lab1) {$S_1$};
  \node[empty, above of=s0, node distance=35mm, text=black, xshift=-7mm] (lab2) {$S_2$};

\end{tikzpicture}
	\caption{
    Proof graph for a~run of the Nielsen transformation on the equation $x a y = y x$.
    The sets~$S_0$, $S_1$, and~$S_2$ are the sets of nodes explored in
    0, 1, and 2 steps of our algorithm, respectively.
  }
	\label{fig:nielsen-unsat-ex}
\end{figure}

\begin{example}\label{ex:init_weq}%
	The word equation $\weq_{1}\colon x a y = y x$ is represented by the pair of
	word terms $\smallencsof 1 = (xay, yx)$. 
	The full proof graph generated by applying the Nielsen
	transformation is depicted in~\cref{fig:nielsen-unsat-ex}.
  \qed
\end{example}

A rule of the Nielsen transformation (cf.\@ \cref{ssec:nielsen}) is represented
using a~(partial) function $\transduct \colon (\msoalph^{*} \times \msoalph^{*})
\to (\msoalph^{*} \times \msoalph^{*})$. Given a~pair of word terms $(\lhs,
\rhs)$ of a~word equation~$\weq$, the function~$\transduct$ transforms it into a
pair of word terms of a~word equation~$\weq'$ that would be obtained by
performing the corresponding step of the Nielsen transformation on~$\weq$. 
Before we express the rules of the Nielsen transformation, we define functions 
performing the corresponding substitution. 
For $x \in \vars$ and $\alpha \in \msoalph$ we define
\begin{align}
  \begin{split}
  \transductof{\xsubalphax}&= \{\,(\lhs, \rhs)\mapsto (\lhs',
  \rhs') \mid \lhs' = \substof \lhs x {\alpha x} \land  \rhs' = \substof \rhs x
  {\alpha x} \,\} \text{ and}\\
  \transductof{\xsubepsilon} &= \{\,(\lhs, \rhs) \mapsto (\lhs',
  \rhs') \mid \lhs' = \substof \lhs x \emptyword \land  \rhs' = \substof \rhs x
  \emptyword \,\} .
  \end{split}
\end{align}
The function $\transductof{\xsubalphax}$ performs a~substitution $\xsubalphax$ while the
function $\transductof{\xsubepsilon}$ performs a~substitution $\xsubepsilon$.

\begin{example}
Consider the pair of word terms $\smallencsof 1 = (xay, yx)$ from \cref{ex:init_weq}.
The application $\transductof \xsubyx (\smallencsof 1)$ would produce the pair
$\smallencsof 2 = (yxay,yyx)$ and the application
$\transductof \xsubepsilon (\smallencsof 1)$
would produce the pair $\smallencsof 3 = (ay,y)$.
\qed
\end{example}

The functions introduced above do not take into account the first symbols of
each side and do not remove a common prefix of the two sides of the equation,
which is a~necessary operation for the Nielsen transformation to terminate.
Let us, therefore, define the following function, which trims (the longest)
matching prefix of word terms of the two sides of an equation:
\begin{align}
  \begin{split}
    \transductof{\trim} = \{\,(\lhs, \rhs) \mapsto (\lhs',\rhs') \mid \exists i \geq 1 \forall j < i
    & \big(\lhs[i] \neq \rhs[i]  \land \lhs[j] = \rhs[j] \\
    & {} \land \lhs' = \substrof \lhs i \land \rhs' = \substrof \rhs i \big)
    \,\} .
  \end{split}
\end{align}

\begin{example}
Continuing in our running example, the application $\transductof \trim
  (\smallencsof 2)$ produces the pair $\smallencsof 2 ' = (xay,yx)$ and,
  furthermore, $\transductof \trim (\smallencsof 3)$ produces the pair
  $\smallencsof 3 ' = (ay, y)$.
\qed
\end{example}

Now we are ready to define functions corresponding to the rules of the Nielsen
transformation. In particular, the rule $\xtoalphax$ and its symmetric variant (i.e., $x$ is 
the first symbol of either left or right side of an equation) for $x \in \vars$ and
$\alpha \in \msoalph$ (cf.\@ \cref{ssec:nielsen}) can be expressed using the function
\begin{align}
  \begin{split}
 \transductof{\xtoalphax} = \transductof{\trim} \circ \{\, (\lhs, \rhs)\mapsto
  \transductof{x\mapsto \alpha x}(\lhs, \rhs)\mid {} &(\rhs[1] = x \land
  \lhs[1] = \alpha) \lor {}
  \\
  &
  (\rhs[1] = \alpha \land \lhs[1] = x) \,\}
  \end{split}
\end{align}
while the rule $\xtoepsilon$ and its symmetric variant for $x \in \vars$ can be expressed as the function
\begin{align}
	\transductof{\xtoepsilon} &=\transductof{\trim} \circ \{\, (\lhs, \rhs)\mapsto
		\transductof{x\mapsto\emptyword}(\lhs, \rhs)\mid (\lhs[1] = x \vee
		\rhs[1] = x)\}.
\end{align}
If we keep applying the functions defined above on individual pairs of
word terms, while searching for the pair $(\emptyword, \emptyword)$---which
represents the case when a~solution to the original equation~$\weq$ exists---, we would obtain
the Nielsen transformation graph (cf.\@ \cref{ssec:nielsen}).
In the following, we show how to perform the steps symbolically on
a~representation of a~\emph{set} of word equations at once.

%*******************************************************************************
\vspace{-0.0mm}
\subsection{Symbolic Algorithm for Word Equations}
\label{ssec:symbolicWordEquations}
\vspace{-0.0mm}
%*******************************************************************************

In this section, we describe the main idea of our symbolic algorithm for
solving word equations.
We first focus on the case of a~single word equation and in
subsequent sections extend the algorithm to a~richer class.

Our algorithm is based on applying the transformation rules not on a~single
equation, but on a~whole \emph{set of equations} at once.
For this, we define the relations $\bigtransof \xtoalphax$ and $\bigtransof
\xtoepsilon$ that aggregate the versions of $\transductof{\xtoalphax}$ and
$\transductof{\xtoepsilon}$ for all possible $x \in \vars$ and $\alpha \in
\svars$.
The signature of these relations is $(\msoalph^{*} \times \msoalph^{*}) \times
(\msoalph^{*} \times \msoalph^{*})$ and they are defined as follows:
\begin{align}\label{eq:trans-x-yx}
		\bigtransof \xtoalphax = {} & \hspace*{-3mm}\bigcup_{y\in\vars,
    \alpha\in\svars}\hspace*{-3mm}\transductof\ytoalphay &
	  \bigtransof \xtoepsilon	= {} & \bigcup_{y\in\vars} \transductof\ytoepsilon
\end{align}
Note the following two properties of the relations:
\begin{inparaenum}[(i)]
  \item  they produce outputs of all possible Nielsen transformation steps
    applicable with the first symbols on the two sides of the equations and
  \item  they include the \emph{trimming} operation.
\end{inparaenum}

We compose the introduced relations into a single one, denoted as $\bigtransof
\step$ and defined as $\bigtransof \step = \bigtransof \xtoalphax \cup \bigtransof
\xtoepsilon$. The relation $\bigtransof \step$ can then be used to compute
\emph{all successors} of a set of word terms of equations in one step.
For a~set
of word terms $S$ we can compute the $\bigtransof \step$-image of~$S$ to obtain all
successors of pairs of word terms in~$S$.
The initial configuration, given a~word equation~$\weq\colon \lhs = \rhs$, is the set
$\encsof \weq = \{(\lhs, \rhs)\}$.

\begin{example} \label{ex:lift_to_sets}
  Lifting our running example to the introduced notions over sets, we start with
  the set $\encsof \weq = S_0 = \{\smallencsof 1 = (xay, yx)\}$. After applying
  $\bigtransof \step$ on~$\encsof \weq$, we obtain the set $S_1 = \{\smallencsof
  2 ' = (xay,yx), \smallencsof 3 ' = (ay, y), (axy, yx), (a, \emptyword)\}$. The
  pairs~$\smallencsof 2 '$ and~$\smallencsof 3 '$ were described earlier, the
  pair $(axy, yx)$ is obtained by the transformation $\transductof \ytoxy$, and
  the pair $(a, \emptyword)$ is obtained by the transformation~$\transductof
  \ytoepsilon$. If~we continue by computing $\bigtransof \step (S_1)$, we obtain
  the set $S_2 = S_1 \cup \{ (ax, x) \}$, as shown
  in~\cref{fig:nielsen-unsat-ex} (the pair $(ax,x)$ was obtained from
  $(axy, yx)$ using the transformation $\transductof \ytoepsilon$).
  \qed
\end{example}

Using the symbolic representation, we can formulate the problem of checking
satisfiability of a~word equation~$\weq$ as the task of
\begin{itemize}
  \item  either testing whether $(\emptyword, \emptyword) \in
    \transclof{\bigtransof \step} (\encsof \weq)$; if the membership holds, it
    means that the constraint~$\weq$ is satisfiable, or
  \item  finding a set (called \emph{unsat-invariant}) $\encsof \invar$
    such that $\encsof \weq \subseteq \encsof \invar$, $(\emptyword, \emptyword) \notin
    \encsof \invar$, and $\bigtransof \step (\encsof \invar) \subseteq \encsof
    \invar$, implying that $\weq$ is unsatisfiable.
\end{itemize}
In the following sections, we show how to encode the problem into the
RMC framework.

\begin{example}
	To proceed in our running example, when we apply $\bigtransof\step$ on~$S_2$, we
	get $\bigtransof\step(S_2)\subseteq S_2$. Since $e_1\in S_2$ and
	$(\emptyword, \emptyword)\notin S_2$, the set~$S_2$ is our unsat-invariant,
	which means that~$\weq_1$ is unsatisfiable.
  \qed
\end{example}

%*******************************************************************************
\vspace{-0.0mm}
\subsection{Towards Symbolic Encoding}
\label{ssec:towardsSymbolicEncoding}
\vspace{-0.0mm}
%*******************************************************************************

Let us now discuss some possible encodings of the word equations satisfiability
problem into RMC. Recall that our task is to find an encoding such that the
encoded equation (corresponding to initial configurations in RMC) and
satisfiability condition (corresponding to destination configurations) are regular
languages and the transformation (transition) relation is a~rational relation.
We~start by describing two possible methods of encodings that do not work,
analyze why they cannot be used, and then describe a~working encoding that we do use.

The first idea about how to encode a~set of word equations as a~regular language is to
encode a~pair $\smallencsof \weq = (\lhs, \rhs)$ as a~word $\lhs \cdot \sidesep
\cdot \rhs$, where $\sidesep \notin \msoalph$.
One immediately finds out that although the transformations $\transductof
\xtoalphax$
and $\transductof \xtoepsilon$ are rational (i.e., expressible using
a~transducer), the transformation $\transductof \trim$,
which removes the longest matching prefix from both sides, is not (a~transducer
with an unbounded memory to remember the prefix would be required).

The second attempt of an encoding might be to encode $\smallencsof \weq = (\lhs, \rhs)$ as a~rational
binary relation, represented, e.g., by a~(not necessarily length-preserving)
2-tape transducer (with one tape for $\lhs$ and the other tape for $\rhs$) and use
four-tape transducers to represent the transformations (with two tapes
for $\lhs$ and $\rhs$ and two tapes for $\lhs'$ and $\rhs'$). 
The transducers implementing $\transductof \xtoyx$ and $\transductof
\xtoepsilon$ can be constructed easily and so can be the transducer
implementing~$\transductof \trim$, so this solution looks appealing.
One, however, quickly realizes that there is an issue in computing
$\bigtransof \step (\encsof \weq)$.
In particular, since~$\encsof \weq$ and~$\bigtransof \step$ are both represented
as rational relations, the intersection $(\encsof \weq \times \msoalph^{*} \times
\msoalph^{*}) \cap \bigtransof \step$, which needs to be computed first, may not
be rational any more.
Why?
Suppose $\encsof \weq = \{\,(a^m b^{n}, c^m) \mid m,n \geq 0\,\}$ and
$\bigtransof \step = \{\,(a^m b^{n}, c^{n}, \emptyword, \emptyword) \mid m,n \geq 0 \,\}$.
Then the intersection $(\encsof \weq \times \msoalph^{*} \times
\msoalph^{*}) \cap \bigtransof \step = \{\,(a^{n} b^{n}, c^{n}, \emptyword, \emptyword) \mid n
\geq 0\,\}$ is clearly not rational any more.

%-------------------------------------------------------------------------------
\vspace{-0.0mm}
\subsection{Symbolic Encoding of Quadratic Equations into RMC}
\label{ssec:symbolicEncodingQuadratictoRMC}
\vspace{-0.0mm}
%-------------------------------------------------------------------------------
%
We therefore converge on the following method of representing word equations by
a~regular language. A~set of pairs of word terms is represented as a~regular
language over a 2-track alphabet with padding $\msoalphpad^2$, where
$\msoalphpad = \msoalph \cup \{\pad\}$, using an FA. For instance, $\smallencsof
1 = (x a y, y x)$ would be represented by the regular language
$\twotrack{x~a~y}{y~x~\pad} \twotrack \pad \pad ^{*}$. 
In other words, 
the equation $e_1$ has many encodings that differ by the padding, with $\twotrack{x~a~y}{y~x~\pad}$ 
being the shortest encoding. The valid representation of the equation contains all of these encodings.
On the other hand, Nielsen transformations are represented by (in general,
length non-preserving) \emph{binary} rational relations over the 2-track alphabet $\msoalphpad^2$
(the first item of each pair refers to an encoding of an equation and the second one refers to the particular transformation applied to the encoding). 
For instance, the transformation $\transductof\xtoepsilon$
is represented by a~rational relation containing, e.g., the pair
$\left(\twotrack {x~a~y~\pad}{y~x~\pad~\pad} ~,~
\twotrack{a~y~\pad}{y~\pad~\pad} \right)$ and
the transformation $\transductof{\rewriteto y {x y}}$
is represented by a~rational relation containing, e.g., the pair
$\left(\twotrack{x~a~y~\pad}{y~x~\pad~\pad} ~,~ \twotrack{a~x~y}{y~x~\pad}\right)$.

Formally, we first define the \emph{equation encoding function} $\eqencode
\colon (\msoalph^{*})^2 \to (\msoalphpad^2 \setminus \{\twotrack \pad \pad\})^{*}$
such that for a~pair of word terms $\lhs = a_1 \ldots a_n$ and $\rhs = b_1
\ldots b_m$ (without loss of generality we assume that $n \geq m$), we have
$\eqencodeof{\lhs, \rhs} =
\bintrackelip{a_1}{a_m}{a_{m+1}}{a_n}{b_1}{b_m}{\pad}{\pad}$.
We also lift $\eqencode$ to sets of pairs of word terms $S \subseteq \msoalph^{*}
\times \msoalph^{*}$ as
$\eqencodeof S = \{\eqencodeof {\lhs, \rhs} \mid (\lhs, \rhs) \in S\}$.

Let $\sigma$ be a~symbol. We define the \emph{padding} of a~word~$w$ with
respect to $\sigma$ as the language $\padfnc_\sigma =
\{ (w, w')\mid w' \in \{w\}.\{\sigma\}^{*} \}$, i.e., it is a~set of words obtained from~$w$ by extending
it by an arbitrary number of~$\sigma$'s.
Moreover, we also create a~(length non-preserving) transducer~$T_{\trim}$ that
performs trimming of its input; this is easy to implement by a~two-state
transducer that replaces a~prefix of symbols of
the form~$\twotrack \beta \beta$ with~$\epsilon$, for $\beta \in \msoalphpad$.
We define the function $\encode$, used for encoding word
equations into regular languages, as
$\encode = T_\trim \circ \padfnc_{\twotrackpadsmall}\circ \eqencode$,
i.e., it takes an encoding of the equation, adds padding, and trims the maximum
shared prefix of the two sides of the equation.
For example, $\encodeof{b x a y, b y x} = \twotrack{x~a~y}{y~x~\pad} \twotrack
\pad \pad ^{*}$. Moreover, for an equation $e\in\msoalph^{*}
\times \msoalph^{*} $, a word $w \in \encodeof{e}$, and a
Nielsen rule $\rho$, we use $w[\rho]$ to denote the set $\encodeof{\transductof{\rho}(e)}$. 

\begin{lemma}
	\label{lem:eq-enc}
	Given a~word equation $\weq\colon \lhs = \rhs$ for $\lhs, \lhs \in \msoalph^{*}$,
	the set $\encodeof \weq$ is regular.
\end{lemma}
\begin{proof}
Without loss of generality we assume that $|\rhs| \leq |\lhs|$.
We give the following \msostr formula that encodes~$\weq$:
\begin{align}
  \begin{split}
    \varphi_\weq(w, w') \defiff {} &
    \bigwedge_{1 \leq k \leq |\lhs|}\hspace*{-3mm} w[k] = \lhs[k] ~\land
    \hspace*{-2mm}\bigwedge_{1 \leq k \leq |\rhs|} \hspace*{-3mm} w'[k] = \rhs[k] ~\land {} 
	\hspace*{-2mm}\bigwedge_{|\rhs| < k \leq |\lhs|} \hspace*{-4mm} w'[k] = \pad ~\land\\
    & \forallp p ((p > |\lhs|) \limpl (w[p] = \pad \land w'[p] = \pad))
  \end{split}
\end{align}
  From \cref{prop:msostr-power}, it follows that $\langof{
  \varphi_{\weq}}$ is a regular binary relation and, moreover, it can be interpreted
  as a~regular language over the composed alphabet~$\msoalphpad^2$.
  Since the image of a~regular language with respect to a rational relation (realizing $T_\trim$) is also
  regular (cf.\@ \cref{prop:rat-prop}), it
  follows that $\encodeof{\langof{\varphi_{\weq}}}$ is also regular.
\end{proof}

Using the presented encoding,
when trying to express the $\transductof \xtoalphax$ and $\transductof \xtoepsilon$
transformations, we, however, encounter an issue with the need of an unbounded
memory.
For instance, for the language $L = \twotrack x y ^{*}$, the transducer
implementing $\transductof \xtoyx$ would need to remember how many times it has
seen $x$ on the first track of its input (indeed, the image of~$L$ with respect to
$\transductof \xtoyx$, i.e., the set $\{\,\encodeof{u, v} \mid \exists n \geq
0\colon u = (yx)^{n} \land v = y^{n} \pad^{n}\,\}$, is no longer regular).

We address this issue in several steps:
first, we give a~rational relation that correctly represents the transformation
rules for cases where the equation $\weq$ is quadratic, and extend our algorithm to
equations with more occurrences of variables in \cref{section:conjunction}.
Let us define the following, more general, restriction of $\transductof
\xtoalphax$
to equations with at most $i \in \nat$ occurrences of variable~$x$:
\begin{equation}
\transductof \xtoalphax ^{\leq i} = \transductof \xtoalphax \cap \{\,((\lhs, \rhs),
  (w, w'))
  \mid w, w' \in \msoalph^{*}, \occurof x {\lhs \concat \rhs} \leq i\,\} .
\end{equation}
We define $\transductof \xtoepsilon ^{\leq i}$, $\transductof \xsubalphax ^{\leq i}$, and $\transductof \xsubepsilon ^{\leq i}$ similarly.
%

%-------------------------------------------------------------------------------
\vspace{-0.0mm}
\subsubsection{Encoding Nielsen Transformations as Rational Relations}
\label{sssec:enc_rational}
\vspace{-0.0mm}
%-------------------------------------------------------------------------------

Next, in order to be able to perform the operations given by~$\transductof
\xtoepsilon ^{\leq i}$ and~$\transductof \xtoalphax ^{\leq i}$ on our encoding
within the RMC framework, we need to encode them as rational relations.
In this section, we define rational relations~$\T_{\xtoepsilon}^{\leq i}$
and~$\T_{\xtoalphax}^{\leq i}$ that do exactly this encoding.
We obtain these relations in successive steps, by defining several intermediate 
formulae with the transformation as subscript and the number of variables as 
superscript, e.g., $\varphi_{\xtoepsilon}^{\leq n}$ and $\psi_{\xtoalphax}^{n}$.

We begin with defining some useful \msostr predicates for an \msostr
string variable~$w$, a~word constraint variable~$x$, and positions $k_1, \dots, k_m$.%
\begin{align}
  \ordered{k_{1}, \ldots, k_m} \defiff {}& \bigwedge_{1 \leq i < m} k_{i} < k_{i+1}
  \\
  \alleq x w {k_{1}, \ldots, k_m} \defiff {}& \bigwedge_{1 \leq i \leq m} w[k_{i}] = x \\
	\begin{split}
	\occur x w {k_{1}, \ldots, k_m} \defiff {}& \ordered{k_{1}, \ldots, k_m} \wedge \alleq x w {k_{1}, \ldots, k_m} \\
	& \hspace*{10.1mm} \wedge \forallp j\Big(w[j] = x \rightarrow \bigvee_{1\leq i \leq m} j = k_{i}\Big)
\end{split}
\end{align}

We use the following \msostr formula to define the transformation $x \mapsto \epsilon$
for~$n$ occurrences of~$x$ in a~single string. 
The formula guesses $n$ positions of $x$ and then ensures that all symbols in $w'$ are 
correctly shifted. In particular, the symbols on positions smaller than $i_1$ are copied from $w$ 
without change. Symbols in $w$ on positions between $i_\ell$ and $i_{\ell+1}$ are shifted $\ell$ positions 
to the left in $w'$. And the remaining positions in $w'$ are filled with $\pad$'s.
\begin{align}
  \begin{split}
    \psi^{n}_{x \mapsto \epsilon}(w, w') \defiff
    \existsp i_{1}, \ldots, i_{n}
    \Big(&
       \occur x w {i_{1}, \ldots, i_{n}} \land {}
         \forallp j (j < i_{1} \limpl w'[j] = w[j]) \land {} \\
        \bigwedge_{1 \leq k < n} &
          \forallp j((i_k < j < i_{k+1}) \limpl w'[j-k] = w[j]) \land {} \\
        & \forallp j (i_{n} < j \limpl w'[j-n] = w[j]) \land {} \\
        \bigwedge_{1 \leq k \leq n} & w'[\strend - k] = \pad
    \Big),
  \end{split}
\end{align}
where $w'[\strend - k] = \pad$ stands for the formula
$\existsp r \existsp s (r = \strend \land r = s + k \land w'[s] = \pad)$ 
(cf.~\cref{fig:mso_sugar}).
We extend $\psi^{n}_{x \mapsto \epsilon}$ to describe the relation on pairs of strings:
\begin{align}
  \psi'^{n}_{x\mapsto \epsilon}(\lhs, \rhs, \lhs', \rhs') \defiff {} &
		\bigvee_{0 \leq k \leq n } \psi^k_{x\mapsto \epsilon}(\lhs, \lhs') \land
    \psi^{n-k}_{x\mapsto \epsilon}(\rhs, \rhs')
  \\
  \psi'^{\leq n}_{x\mapsto \epsilon}(\lhs, \rhs, \lhs', \rhs') \defiff {} &
		\bigvee_{0 \leq k \leq n } \psi'^k_{x\mapsto \epsilon}(\lhs, \rhs, \lhs', \rhs')
  \\
    \varphi^{\leq n}_{x \mapsto \epsilon}(\lhs, \rhs, \lhs', \rhs') \defiff {} &
		(\lhs[1] = x \vee \rhs[1] = x) \land \psi'^{\leq n}_{x\mapsto \epsilon}(\lhs, \rhs, \lhs', \rhs')
\end{align}
Next, we define the transformation $x \mapsto \alpha x$ for $n$ occurrences
of~$x$ in a~single string.
The formula guesses $n$ positions of $x$ in $w$. Then, it ensures that on $i_\ell + \ell$ position in $w'$
there
is the symbol $\alpha$ and all other symbols from $w$ are copied to the correct positions in $w'$.
In particular, symbols in $w$ on positions between $i_\ell$ and $i_{\ell+1}$ are shifted $\ell$ positions 
to the right in $w'$.
\begin{align}
  \begin{split}
    \psi^{n}_{x \mapsto \alpha x} &(w, w') \defiff {}\\
    & \hspace*{-4mm} \existsp i_{1}, \ldots, i_{n}
    \Big(
      \occur x w {i_{1}, \ldots, i_{n}} \land {} \forallp j (j \leq i_{1} \limpl w'[j] = w[j]) \land {} \\
        &\bigwedge_{1 \leq k < n}
          w'[i_k+k] = \alpha \land \forallp j((i_k < j \leq i_{k+1}) \limpl w'[j+k] = w[j]) \land {} \\
        & \hspace*{10.1mm} w'[i_{n}+n] = \alpha \land \forallp j (i_{n} < j \limpl w'[j+n] = w[j]) \land {} \\
        & \bigwedge_{1 \leq k \leq n} w[\strend - k] = \pad
    \Big)
  \end{split}
\end{align}
We extend $\psi^{n}_{x \mapsto \alpha x}$ to describe the relation on pairs of strings:
\begin{align}
  \psi'^{n}_{x \mapsto \alpha x}(\lhs, \rhs, \lhs', \rhs') \defiff {} &
  \bigvee_{0 \leq k \leq n } \psi^k_{x \mapsto \alpha x}(\lhs, \lhs') \land
    \psi^{n-k}_{x \mapsto \alpha x}(\rhs, \rhs')
  \\
  \psi'^{\leq n}_{x \mapsto \alpha x}(\lhs, \rhs, \lhs', \rhs') \defiff {} &
		\bigvee_{0 \leq k \leq n} \psi'^k_{x \mapsto \alpha x}(\lhs, \rhs, \lhs', \rhs')
  \\
  \begin{split}
    \varphi^{\leq n}_{x \mapsto \alpha x}(\lhs, \rhs, \lhs', \rhs') \defiff {} &
    \big((\lhs[1] = x \land \rhs[1] = \alpha) \lor (\lhs[1] = \alpha \land \rhs[1] = x)\big)\\
		&{}\land \psi'^{\leq n}_{x \mapsto \alpha x}(\lhs, \rhs, \lhs', \rhs')
  \end{split}
\end{align}

The constructed formulae~$\varphi^{\leq n}_{x \mapsto \epsilon}$ and~$\varphi^{\leq
n}_{x \mapsto \alpha x}$ describe regular relations with arity~$4$ over the
alphabet~$\msoalphpad$.
Since they are regular (i.e., length-preserving), they can also be interpreted
as binary relations over the composed alphabet~$\msoalphpad^2$.
This interpretation can easily be done by modifying a length-preserving transducer corresponding to $\varphi^{\leq n}_{x \mapsto \epsilon}$ and~$\varphi^{\leq
n}_{x \mapsto \alpha x}$ respectively in a way that each transition $q \ltr{a_1, a_2, a_3, a_4} r$ 
is replaced by the transition $q \ltr{\scalebox{0.8}{\twotrack{a_1}{a_2}}, \scalebox{0.8}{\twotrack{a_3}{a_4}}} r$ 
(with the same sets of states).
When interpreted as binary relations in this way, they denote the relations
containing pairs (without loss of generality, we show this only for $\varphi^{\leq n}_{x
\mapsto \epsilon}$)
\[
\left(
\bintrackelip{u_1}{u_m}{u_{m+1}}{u_n}{v_1}{v_m}{\pad}{\pad}\twotrack \pad \pad ^{i}
~,~
\bintrackelip{w_1}{w_k}{w_{k+1}}{w_\ell}{z_1}{z_k}{\pad}{\pad}\twotrack \pad \pad ^j
\right)
\]
such that
\begin{itemize}
  \item  $\tau_{x \mapsto \epsilon}^{\leq n}((u_1 \ldots u_n, v_1 \ldots v_m))
    = (w_1 \ldots w_\ell, z_1 \ldots z_k)$ for $\varphi^{\leq n}_{x \mapsto
    \epsilon}$,
  \item  $\tau_{x \mapsto \alpha x}^{\leq n}((u_1 \ldots u_n, v_1 \ldots v_m))
    = (w_1 \ldots w_\ell, z_1 \ldots z_k)$ for $\varphi^{\leq n}_{x \mapsto
    \alpha x}$, and
  \item $\max(m,n) + i = \max(k, \ell) + j$.
\end{itemize}

Let us now consider~$\varphi^{\leq n}_{x \mapsto \epsilon}$.
We note that for every $((u,v),(w,z)) \in \tau_{x \mapsto \epsilon}^{\leq n}$,
there will indeed be a~corresponding pair in~$\varphi^{\leq n}_{x \mapsto
\epsilon}$ (actually, there will be infinitely many such pairs that differ in
the number of used padding symbols).

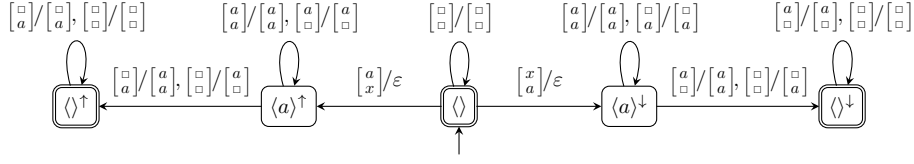
\begin{figure}[t]
	\begin{tikzpicture}
    [node distance=2.8cm,->,>=stealth,transform shape,scale=0.8]
  
    \tikzstyle{state} = [draw=black,minimum height=0.5cm, inner sep=4pt,rectangle, rounded corners=1mm]
  
    \node[state, initial below, accepting, initial text={}] (s0) at (0,0) {$\langle \rangle$};
    \node[state, right of=s0] (s1) {$\langle {a} \rangle^\downarrow$};
    \node[state, accepting, node distance=3.5cm, right of=s1] (s2) {$\langle \rangle^\downarrow$};
    \node[state, left of=s0] (s3) {$\langle {a} \rangle^\uparrow$};
    \node[state, accepting, node distance=3.5cm, left of=s3] (s4) {$\langle \rangle^\uparrow$};

    \draw (s0) edge [loop above] node[above] {$\twotrack{\pad}{\pad}/\twotrack{\pad}{\pad}$} (s0);
    \draw (s0) to node[above] {$\twotrack{x}{{a}}/\emptyword$} (s1);
    \draw (s0) to node[above] {$\twotrack{a}{x}/\emptyword$} (s3);

    \draw (s1) edge [loop above] node[above] {$\twotrack{a}{a}/\twotrack{a}{a}, \twotrack{\pad}{a}/\twotrack{\pad}{a}$} (s1);
    \draw (s1) to node[above] {$\twotrack{a}{\pad}/\twotrack{a}{a}, \twotrack{\pad}{\pad}/\twotrack{\pad}{a}$} (s2);
    \draw (s2) edge [loop above] node[above] {$\twotrack{a}{\pad}/\twotrack{a}{\pad}, \twotrack{\pad}{\pad}/\twotrack{\pad}{\pad}$} (s2);

    \draw (s3) edge [loop above] node[above] {$\twotrack{a}{a}/\twotrack{a}{a}, \twotrack{a}{\pad}/\twotrack{a}{\pad}$} (s3);
    \draw (s3) to node[above] {$\twotrack{\pad}{a}/\twotrack{a}{a}, \twotrack{\pad}{\pad}/\twotrack{a}{\pad}$} (s4);
    \draw (s4) edge [loop above] node[above] {$\twotrack{\pad}{a}/\twotrack{\pad}{a}, \twotrack{\pad}{\pad}/\twotrack{\pad}{\pad}$} (s4);
  
  \end{tikzpicture}
	\vspace*{-5mm}
	\caption{Example of a transducer realizing the relation $T_{x\mapsto \epsilon}^{\leq 1}$ for the set of variables $\vars = \{ x \}$ and the alphabet $\Sigma = \{ a \}$.
		The names of the states of the transducers denote symbols that the transducers
		remember to output on the tape given by the symbols~$\uparrow$ and~$\downarrow$.}
	\label{fig:transducer-ex}
\end{figure}

In order to get closer to~$\transductof \xtoepsilon$,
we need to modify the relation of~$\varphi^{\leq n}_{x \mapsto \epsilon}$ to
also perform trimming of the shared prefix.
We do this modification by taking the (length-preserving) two-track
transducer~$T_{x\mapsto \epsilon}^{\leq n}$ that
recognizes $\varphi^{\leq n}_{x \mapsto \epsilon}$ (it can be constructed due
to \cref{prop:msostr-power}).
Moreover, we also create a~(length non-preserving) transducer~$T_{\trim}$ that
performs trimming of its input; this is easy to implement by a~two-state
transducer that replaces a~prefix of symbols of
the form~$\twotrack \beta \beta$ with~$\epsilon$, for $\beta \in \msoalphpad$.
By composing the two transducers, we~obtain~$T_{\xtoepsilon}^{\leq n} =
T_{\trim} \circ T_{x\mapsto \epsilon}^{\leq n}$. 
An example of a transducer realizing $T_{x\mapsto \epsilon}^{\leq 1}$ is 
shown in \cref{fig:transducer-ex}. 

We can repeat the previous reasoning for~$\varphi^{\leq n}_{x \mapsto \alpha
x}$ in a~similar way to obtain the (length non-preserving)
transducer~$T_{\xtoalphax}^{\leq n}$.
\begin{lemma}
It holds that
  $\tau_{\xtoepsilon}^{\leq n}((u_1 \ldots u_n, v_1 \ldots v_m)) = (w_1
    \ldots w_\ell, z_1 \ldots z_k)$ iff
\[
\left(
\bintrackelip{u_1}{u_m}{u_{m+1}}{u_n}{v_1}{v_m}{\pad}{\pad}\twotrack \pad \pad ^{i}
~,~
\bintrackelip{w_1}{w_k}{w_{k+1}}{w_\ell}{z_1}{z_k}{\pad}{\pad}\twotrack \pad \pad ^j
\right)
\in \langof{T^{\leq n}_{\xtoepsilon}}
\]
for some $i,j \in \nat$.

  Further, it holds that
  $\tau_{\xtoalphax}^{\leq n}((u_1 \ldots u_n, v_1 \ldots v_m)) = (w_1
    \ldots w_\ell, z_1 \ldots z_k)$ iff
\[
\left(
\bintrackelip{u_1}{u_m}{u_{m+1}}{u_n}{v_1}{v_m}{\pad}{\pad}\twotrack \pad \pad ^{i}
~,~
\bintrackelip{w_1}{w_k}{w_{k+1}}{w_\ell}{z_1}{z_k}{\pad}{\pad}\twotrack \pad \pad ^j
\right)
\in \langof{T^{\leq n}_{\xtoalphax}}
\]
for some $i,j \in \nat$.
\end{lemma}
\begin{proof}
	The proof follows from the above described construction of transducers $T_{\xtoalphax}^{\leq n}$ and
	$T_{\xtoepsilon}^{\leq n}$.
\end{proof}

\newcommand{\reachset}[0]{\mathit{reach}}
\newcommand{\oldset}[0]{\mathit{old}}
\newcommand{\procset}[0]{\mathit{processed}}
\newcommand{\extractmodel}[0]{\mathit{ExtractModel}}
\begin{algorithm}[t]
  \caption{Solving a~string constraint $\varphi$ using RMC}
  \label{alg:rmc}
	\SetKwComment{ctriang}{\scalebox{1.3}{\textcolor{gray}{$\triangleright$\,}}}{}
	\KwIn{Encoding $\I$ of a formula $\varphi$ (the initial set),\\
  \hspace*{13mm}transformers $\T_{\xtoalphax}$, $\T_{\xtoepsilon}$, and \\
  \hspace*{13mm}the destination set $\D$}
	\KwOut{A model of $\varphi$ if $\varphi$ is satisfiable, $\false$ otherwise}
  $\reachset_0 := \I$\;
  $\procset := \emptyset$\;
	$\T := \T_{\xtoalphax} \cup \T_{\xtoepsilon}$\;\label{algln:sett}
  $i := 0$\;
  \While{$\reachset_{i} \not\subseteq \procset$}
  {
    \If{$\D \cap \reachset_{i} \neq \emptyset$\label{algln:touch}}
    {
      \Return{$\extractmodel(\{\T\}_{j\in\{0,\dots, i\}}, \D, \reachset_0, \ldots, \reachset_{i})$}\;
    }
    $\procset := \procset \cup \reachset_{i}$\;
    $\reachset_{i+1} := \saturate\circ\T(\reachset_{i})$\label{algln:tr}\tcp*{$\saturatenew\circ\T(\reachset_{i})$}
    $i$++\;
  }
  \Return{$\false$}\;
\end{algorithm}

%-------------------------------------------------------------------------------
\vspace{-0.0mm}
\subsubsection{RMC for Quadratic Equations}
\label{sssec:rmcForQuadraticEquations}
\vspace{-0.0mm}
%-------------------------------------------------------------------------------

In \cref{alg:rmc}, we give a~high-level algorithm for solving string constraints
using RMC.
The algorithm is parameterized by the following:
\begin{inparaenum}[(i)]
  \item  a~regular language~$\I$ encoding a~formula~$\varphi$ (the initial set),
  \item  rational relations given by the transducers $\T_{\xtoalphax}$ and $\T_{\xtoepsilon}$, and
  \item  the destination set~$\D$ (also given as a~regular language).
\end{inparaenum}
The algorithm tries to solve the RMC problem $(\I, \T_{\xtoalphax}\cup
\T_{\xtoepsilon}, \D)$ by an iterative unfolding of the transition relation $\T$
computed in Line~\ref{algln:sett}, looking for an element~$w_{i}$ from~$\D$.
If such an element is found in the set~$\reachset_{i}$, we call Function
\textit{\ref{alg:extract-model-func}} to extract a~model of the
original word equation by starting a backward run from~$w_{i}$, computing
pre-images $w_{i-1}, \ldots, w_1$ over transfomers~$\T_{\xtoalphax}$
and~$\T_{\xtoepsilon}$ (restricting them to $\reachset_j$ for every~$w_j$),
while updating values of the variables according to the transformations that were
performed. 
Note that $\extractmodel$ uses a more general interface allowing to specify
a~transducer for each backward step (Line~\ref{alg-line:back-step} of Function
\textit{\ref{alg:extract-model-func}}).
This is utilized later in~\cref{section:conjunction}; here, we just pass~$i$
copies of~$\T$.
\cref{algln:sett} also employs \emph{saturation} of the sets of reachable
configurations defined as:
\begin{equation}
	\saturate(L) = \left\{ u ~\middle|~ w \in L, w \in u.\twotrack \pad \pad^*  \right\}.
\end{equation}
Intuitively, $\saturate(L)$ removes some occurrences (possibly none of them) of
the padding symbol at the end of all words from~$L$.
If $L$ is a regular language, $\saturate(L)$ is regular as well (from an FA
representing~$L$ we can get $\saturate(L)$ by saturating its transitions over
the padding symbol).
We saturate the sets of reachable configurations, because we want to keep the shortest words
(i.e., words without padding symbols)---e.g., the transformer
$\T^{\weq}_{\xtoepsilon}$ need not generate all shortest words.

\SetKwProg{Fn}{Function}{:}{}
\SetProcNameSty{texttt}
\vspace{-0mm}
\hspace*{-7mm}
\begin{function}[t]
	$\model(a) := a$ for each $a\in\Sigma$, $\model(x) := \emptyword$ for each $x\in\vars$\;
  let $w_{i} \in \D \cap \reachset_{i}$\;
	\For{$\ell = i \mathbf{~downto~} 1$} {
		let $w_{\ell - 1} \in \T_\ell^{-1}(w_\ell) \cap \reachset_{\ell - 1}$\;\label{alg-line:back-step}
		let $\rho$ be a rule s.t. $w_{\ell} \in w_{\ell - 1}[\rho]$\;
		\If{$\rho = \ytoalphay$} {
			$\model(y) := \model(\alpha).\model(y)$\;
		}
	}
	\Return $\model$\;
  \caption{ExtractModel($\{\T_j\}_{j\in\{0,\dots, i\}}, \D, \reachset_0, \ldots, \reachset_{i}$)}
	\label{alg:extract-model-func}
\end{function}

\cref{alg:rmc} follows a \emph{breadth-first search} (BFS) strategy:
from the initial set $\I$, we apply both transformers $\T_{\xtoalphax}$ and $\T_{\xtoepsilon}$ on all elements of $\I$ at the same time, before repeatedly applying the transformers on the result.
This corresponds to a~breadth-first application of the transformers if we applied them one element of~$\I$ at a~time.

Our first instantiation of the algorithm is for checking satisfiability of
a~single quadratic word equation $\weq\colon \lhs = \rhs$.
We instantiate the RMC problem as the tuple
$(\I^{\weq},\T^{\weq}_{\xtoalphax}\cup\T^{\weq}_{\xtoepsilon},\D^{\weq})$ where

\begin{align*}
  \I^{\weq} &= \encodeof{\lhs, \rhs} &
  \T^{\weq}_{\xtoalphax} &= \hspace*{-3.5mm}\bigcup_{y\in\vars,
  \alpha\in\svars}\hspace*{-3.5mm} T_{\ytoalphay}^{\leq 2}  &
  \D^{\weq} &= \big\{\twotrack \pad \pad\big\}^{*} \\
  && \T^{\weq}_{\xtoepsilon} &= \bigcup_{y\in\vars} T_{\ytoepsilon}^{\leq 2}
\end{align*}

\begin{lemma}
  \cref{alg:rmc} instantiated with
  $(\I^{\weq},\T^{\weq}_{\xtoalphax}\cup\T^{\weq}_{\xtoepsilon},\D^{\weq})$ is
  sound, complete, and terminating if $\weq\colon \lhs = \rhs$ is quadratic.
\end{lemma}
\begin{proof}
	We encode the nodes of a Nielsen proof graph as strings. The
	initial node $\lhs = \rhs$ corresponds to a string from $\I^{\weq}$. Due to the
	padding, the final node $\epsilon = \epsilon$ corresponds to a string from
	$\D^{\weq}$. The relations $\T^{\weq}_{\xtoalphax}$ and
	$\T^{\weq}_{\xtoepsilon}$ implement Nielsen rules $\xtoalphax$ and
	$\xtoepsilon$. From~\cref{lem:Nielsen} we have that the Nielsen proof graph is
	finite. Since our approach implements the \emph{breadth-first search} (BFS)
	strategy, it is sound, complete, and terminating.
\end{proof}

\newcommand{\figExampleQuadraticSetup}[0]{
\begin{figure}[t]
  \centering
	\begin{subfigure}[b]{0.3\textwidth}
    \centering
    \begin{tikzpicture}
  [node distance=1cm,->,>=stealth,transform shape, scale=0.8]

  \tikzstyle{empty}=[]

  \node[state,initial] (s0) {};
  \node[state,right of=s0] (s1) {};
  \node[state,right of=s1] (s4) {};
  \node[state,right of=s4,accepting] (s5) {};

  \draw (s0) edge node[above] {$\twotrack{x}{y}$} (s1);
  \draw (s1) edge node[above] {$\twotrack{a}{x}$} (s4);
  \draw (s4) edge node[above] {$\twotrack{y}{\pad}$} (s5);
  %\draw (s2) edge[bend left] node[above] {$\twotrack{x}{\pad}$} (s5);
  \draw (s5) edge[loop above] node[above] {$\twotrack{\pad}{\pad}$} (s5);

\end{tikzpicture}
    \vspace{13mm}
		\caption{$\I = \aut_{\reachset_0}$}
		\label{fig:aut-ex-proc0}
  \end{subfigure}
  \hfill
	\begin{subfigure}[b]{0.65\textwidth}
    \centering
    \begin{tikzpicture}
  [node distance=2.8cm,->,>=stealth,transform shape,scale=0.8]

  \tikzstyle{state} = [draw=black,minimum height=0.5cm, inner sep=4pt,rectangle, rounded corners=1mm]

  \node[state, initial left, accepting, initial text={}] (s0) at (0,0) {$\langle \rangle_1$};
  \node[state, right of=s0] (s1) {$\langle {y} \rangle^\downarrow$};
  \node[state, accepting, node distance=4cm, right of=s1] (s2) {$\langle \rangle_2$};
  % \node[state, accepting, node distance=4.5cm, right of=s2] (s3) {$\langle \rangle^\downarrow_3$};
  \node[state, below of=s1, node distance=1.5cm] (s5) {$\langle a \rangle^\downarrow$};

  \node[below of=s0,node distance=1.2cm] (s0bel) {};
  \node[below of=s1,node distance=1.2cm, xshift=-12mm] (s1bel) {};
  \node[left of=s5,node distance=1.2cm] (s5bel) {};

  \draw (s0) edge [loop above] node[above] {$\twotrack{\pad}{\pad}/\twotrack{\pad}{\pad}$} (s0);
  \draw (s0) to node[above] {$\twotrack{x}{y}/\emptyword$} (s1);

  \draw[dotted] (s0) to (s0bel);
  \draw[dotted] (s1) to (s1bel);
  \draw[dotted] (s5) to (s5bel);

  \draw (s1) edge [loop above] node[above]
    {$\alpha/\alpha$; $\alpha \in \left\{\twotrack y y, \twotrack a y\right\}$} (s1);
  \draw (s1) to node[above] {$\twotrack{a}{x}/\twotrack{a}{y}, \twotrack{y}{x}/\twotrack{y}{y}$} (s2);
  \draw (s2) edge [loop above] node[above]
    {$\twotrack \alpha \beta / \twotrack \alpha \beta$; $\alpha, \beta \in \left\{ a, y, \pad \right\}$} (s2);
  % \draw (s2) to node[above] {$\alpha/\alpha$; $\alpha \in \left\{\twotrack{y}{\pad}, \twotrack{a}{\pad}\right\}$} (s3);
  % \draw (s3) edge [loop above] node[above] {$\twotrack{a}{\pad}/\twotrack{a}{\pad}, \twotrack{\pad}{\pad}/\twotrack{\pad}{\pad}$} (s3);

  \draw (s1) edge node[auto] {$\twotrack y a / \twotrack y y, \twotrack a a / \twotrack a y$} (s5);

\end{tikzpicture}
    \vspace{-5mm}
		\caption{A part of $T_{\xtoepsilon}^{\leq 2}$}
		\label{fig:aut-ex-trans}
  \end{subfigure}\\[2mm]
	\begin{subfigure}[b]{0.3\textwidth}
    \centering
    \begin{tikzpicture}
  [node distance=2cm,->,>=stealth,transform shape, scale=0.8]

  \tikzstyle{empty}=[]

  \node[state,initial,accepting] (s0) {};

  \draw (s0) edge[loop above] node[above] {$\twotrack{\pad}{\pad}$} (s0);

\end{tikzpicture}
    \vspace{7mm}
		\caption{$\D$}
		\label{fig:aut-ex-dest}
  \end{subfigure}
  \hfill
	\begin{subfigure}[b]{0.65\textwidth}
    \centering
    \begin{tikzpicture}
  [node distance=2.8cm,->,>=stealth,transform shape,scale=0.8]

  \tikzstyle{state} = [draw=black,minimum height=0.5cm, inner sep=4pt,rectangle, rounded corners=1mm]

  \node[state, initial left, accepting, initial text={}] (s0) at (0,0) {$\langle \rangle_1$};
  \node[state, right of=s0,node distance=1.9cm] (s1) {$\langle {x} \rangle^\uparrow$};
  \node[state, right of=s1,node distance=2.5cm] (s3) {$\langle {a} \rangle^\uparrow, \langle x \rangle^\downarrow$};
  \node[state, accepting, node distance=2.3cm, right of=s3] (s2) {$\langle \rangle_2$};
  % \node[state, accepting, node distance=4.5cm, right of=s2] (s3) {$\langle \rangle^\downarrow_3$};
  \node[below of=s1, node distance=1.2cm] (s5) {};

  \node[below of=s0,node distance=1.2cm] (s0bel) {};

  \draw (s0) edge [loop above] node[above] {$\twotrack{\pad}{\pad}/\twotrack{\pad}{\pad}$} (s0);
  \draw (s0) to node[above] {$\twotrack{x}{y}/\emptyword$} (s1);

  \draw[dotted] (s0) to (s0bel);
  \draw[dotted] (s1) to (s5);

  % \draw (s1) edge [loop above] node[above]
  %   {$\alpha/\alpha$; $\alpha \in \left\{\twotrack y y, \twotrack a y\right\}$} (s1);
  \draw (s1) to node[above] {$\twotrack{a}{x}/\twotrack{x}{y}$} (s3);
  \draw (s3) to node[above] {$\emptyword/\twotrack a x$} (s2);
  \draw (s2) edge [loop above] node[above]
    {$\twotrack \alpha \beta / \twotrack \alpha \beta$; $\alpha, \beta \in \left\{ a, y, \pad \right\}$} (s2);
  % \draw (s2) to node[above] {$\alpha/\alpha$; $\alpha \in \left\{\twotrack{y}{\pad}, \twotrack{a}{\pad}\right\}$} (s3);
  % \draw (s3) edge [loop above] node[above] {$\twotrack{a}{\pad}/\twotrack{a}{\pad}, \twotrack{\pad}{\pad}/\twotrack{\pad}{\pad}$} (s3);

\end{tikzpicture}
    \vspace{-3mm}
		\caption{A part of $T_{\xtoyx}^{\leq 2}$}
		\label{fig:aut-ex-trans2}
  \end{subfigure}
  \caption{Input of \cref{alg:rmc} for solving the word equation $x a y = y x$.
  The size of $\T^{\weq}_{\xtoepsilon}$ and $\T^{\weq}_{\xtoyx}$ would be
  prohibitively large, so we only give relevant parts of transducers
  $T_{\xtoepsilon}^{\leq 2}$ and $T_{\xtoyx}^{\leq 2}$.
  Some transitions use shorthand notation with the obvious meaning.
  }
  \label{fig:rmc_setup}
\end{figure}
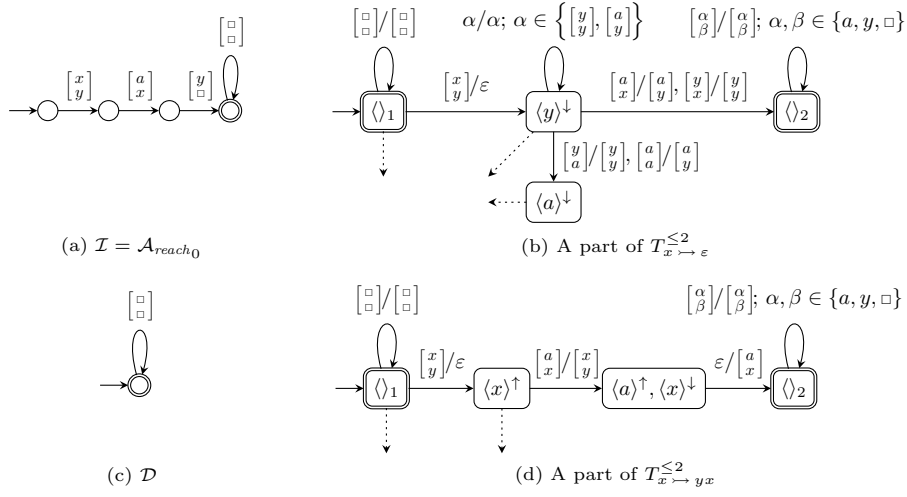
}

\newcommand{\figExampleQuadraticRun}[0]{
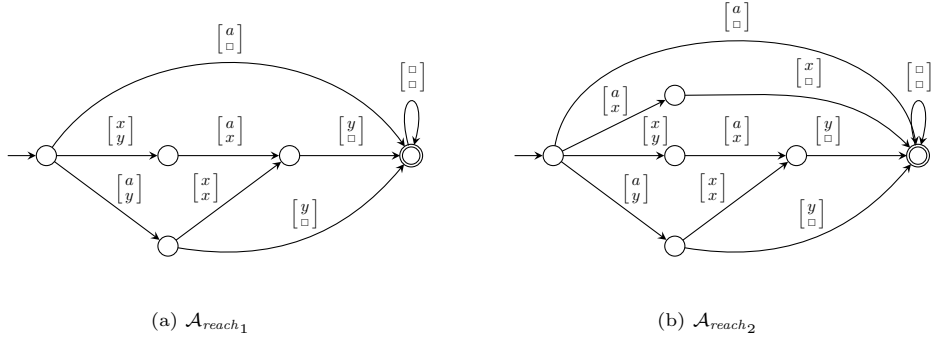
\begin{figure}[t]
\centering
\begin{subfigure}[b]{0.45\textwidth}
  \centering
  \begin{tikzpicture}
  [node distance=2cm,->,>=stealth,transform shape, scale=0.8]

  \tikzstyle{empty}=[]

  \node[state,initial] (s0) {};
  \node[state,right of=s0] (s1) {};
  %\node[state,above of=s1,yshift=-5mm] (s2) {};
  \node[state,below of=s1,yshift=5mm] (s3) {};
  \node[state,right of=s1] (s4) {};
  \node[state,right of=s4,accepting] (s5) {};

  \draw (s0) edge[in=125, out=55] node[above] {$\twotrack{a}{\pad}$} (s5);
  \draw (s0) edge node[above,xshift=2mm] {$\twotrack{x}{y}$} (s1);
  \draw (s0) edge node[above,right,yshift=2mm] {$\twotrack{a}{y}$} (s3);
  \draw (s1) edge node[above] {$\twotrack{a}{x}$} (s4);
  \draw (s3) edge node[above,yshift=-2mm,xshift=-4mm] {$\twotrack{x}{x}$} (s4);
  \draw (s3) edge[bend right] node[above] {$\twotrack{y}{\pad}$} (s5);
  \draw (s4) edge node[above] {$\twotrack{y}{\pad}$} (s5);
  %\draw (s2) edge[bend left] node[above] {$\twotrack{x}{\pad}$} (s5);
  \draw (s5) edge[loop above] node[above] {$\twotrack{\pad}{\pad}$} (s5);

\end{tikzpicture}
  \caption{$\aut_{\reachset_1}$}
  \label{fig:aut-ex-proc1}
\end{subfigure}
\hfill
\begin{subfigure}[b]{0.45\textwidth}
  \centering
  \begin{tikzpicture}
  [node distance=2cm,->,>=stealth,transform shape, scale=0.8]

  \tikzstyle{empty}=[]

  \node[state,initial] (s0) {};
  \node[state,right of=s0] (s1) {};
  \node[state,above of=s1,yshift=-10mm] (s2) {};
  \node[state,below of=s1,yshift=5mm] (s3) {};
  \node[state,right of=s1] (s4) {};
  \node[state,right of=s4,accepting] (s5) {};

  \draw (s0) edge[in=100, out=80] node[above] {$\twotrack{a}{\pad}$} (s5);
  \draw (s0) edge node[above] {$\twotrack{a}{x}$} (s2);
  \draw (s0) edge node[above,pos=0.9] {$\twotrack{x}{y}$} (s1);
  \draw (s0) edge node[above,right,yshift=2mm] {$\twotrack{a}{y}$} (s3);
  \draw (s1) edge node[above] {$\twotrack{a}{x}$} (s4);
  \draw (s3) edge node[above,yshift=-2mm,xshift=-4mm] {$\twotrack{x}{x}$} (s4);
  \draw (s3) edge[bend right] node[above] {$\twotrack{y}{\pad}$} (s5);
  \draw (s4) edge node[above,pos=0.2] {$\twotrack{y}{\pad}$} (s5);
  \draw (s2) edge[out=0] node[above] {$\twotrack{x}{\pad}$} (s5);
  \draw (s5) edge[loop above] node[above] {$\twotrack{\pad}{\pad}$} (s5);

\end{tikzpicture}
  \caption{$\aut_{\reachset_2}$}
  \label{fig:aut-ex-proc2}
\end{subfigure}
\caption{Examples of finite automata $\aut_{\reachset_1}$,
$\aut_{\reachset_2}$ representing sets of configurations reachable in one and
two steps, respectively, when solving the word equation $x a y = y x$
using~\cref{alg:rmc}.
  }
\end{figure}
}

\begin{figure}[t]
  \centering
	\begin{subfigure}[b]{0.3\textwidth}
    \centering
    \begin{tikzpicture}
  [node distance=1cm,->,>=stealth,transform shape, scale=0.8]

  \tikzstyle{empty}=[]

  \node[state,initial] (s0) {};
  \node[state,right of=s0] (s1) {};
  \node[state,right of=s1] (s4) {};
  \node[state,right of=s4,accepting] (s5) {};

  \draw (s0) edge node[above] {$\twotrack{x}{y}$} (s1);
  \draw (s1) edge node[above] {$\twotrack{a}{x}$} (s4);
  \draw (s4) edge node[above] {$\twotrack{y}{\pad}$} (s5);
  %\draw (s2) edge[bend left] node[above] {$\twotrack{x}{\pad}$} (s5);
  \draw (s5) edge[loop above] node[above] {$\twotrack{\pad}{\pad}$} (s5);

\end{tikzpicture}
    \vspace{13mm}
		\caption{$\I = \aut_{\reachset_0}$}
		\label{fig:aut-ex-proc0}
  \end{subfigure}
  \hfill
	\begin{subfigure}[b]{0.65\textwidth}
    \centering
    \begin{tikzpicture}
  [node distance=2.8cm,->,>=stealth,transform shape,scale=0.8]

  \tikzstyle{state} = [draw=black,minimum height=0.5cm, inner sep=4pt,rectangle, rounded corners=1mm]

  \node[state, initial left, accepting, initial text={}] (s0) at (0,0) {$\langle \rangle_1$};
  \node[state, right of=s0] (s1) {$\langle {y} \rangle^\downarrow$};
  \node[state, accepting, node distance=4cm, right of=s1] (s2) {$\langle \rangle_2$};
  % \node[state, accepting, node distance=4.5cm, right of=s2] (s3) {$\langle \rangle^\downarrow_3$};
  \node[state, below of=s1, node distance=1.5cm] (s5) {$\langle a \rangle^\downarrow$};

  \node[below of=s0,node distance=1.2cm] (s0bel) {};
  \node[below of=s1,node distance=1.2cm, xshift=-12mm] (s1bel) {};
  \node[left of=s5,node distance=1.2cm] (s5bel) {};

  \draw (s0) edge [loop above] node[above] {$\twotrack{\pad}{\pad}/\twotrack{\pad}{\pad}$} (s0);
  \draw (s0) to node[above] {$\twotrack{x}{y}/\emptyword$} (s1);

  \draw[dotted] (s0) to (s0bel);
  \draw[dotted] (s1) to (s1bel);
  \draw[dotted] (s5) to (s5bel);

  \draw (s1) edge [loop above] node[above]
    {$\alpha/\alpha$; $\alpha \in \left\{\twotrack y y, \twotrack a y\right\}$} (s1);
  \draw (s1) to node[above] {$\twotrack{a}{x}/\twotrack{a}{y}, \twotrack{y}{x}/\twotrack{y}{y}$} (s2);
  \draw (s2) edge [loop above] node[above]
    {$\twotrack \alpha \beta / \twotrack \alpha \beta$; $\alpha, \beta \in \left\{ a, y, \pad \right\}$} (s2);
  % \draw (s2) to node[above] {$\alpha/\alpha$; $\alpha \in \left\{\twotrack{y}{\pad}, \twotrack{a}{\pad}\right\}$} (s3);
  % \draw (s3) edge [loop above] node[above] {$\twotrack{a}{\pad}/\twotrack{a}{\pad}, \twotrack{\pad}{\pad}/\twotrack{\pad}{\pad}$} (s3);

  \draw (s1) edge node[auto] {$\twotrack y a / \twotrack y y, \twotrack a a / \twotrack a y$} (s5);

\end{tikzpicture}
    \vspace{-5mm}
		\caption{A part of $T_{\xtoepsilon}^{\leq 2}$}
		\label{fig:aut-ex-trans}
  \end{subfigure}\\[2mm]
	\begin{subfigure}[b]{0.3\textwidth}
    \centering
    \begin{tikzpicture}
  [node distance=2cm,->,>=stealth,transform shape, scale=0.8]

  \tikzstyle{empty}=[]

  \node[state,initial,accepting] (s0) {};

  \draw (s0) edge[loop above] node[above] {$\twotrack{\pad}{\pad}$} (s0);

\end{tikzpicture}
    \vspace{7mm}
		\caption{$\D$}
		\label{fig:aut-ex-dest}
  \end{subfigure}
  \hfill
	\begin{subfigure}[b]{0.65\textwidth}
    \centering
    \begin{tikzpicture}
  [node distance=2.8cm,->,>=stealth,transform shape,scale=0.8]

  \tikzstyle{state} = [draw=black,minimum height=0.5cm, inner sep=4pt,rectangle, rounded corners=1mm]

  \node[state, initial left, accepting, initial text={}] (s0) at (0,0) {$\langle \rangle_1$};
  \node[state, right of=s0,node distance=1.9cm] (s1) {$\langle {x} \rangle^\uparrow$};
  \node[state, right of=s1,node distance=2.5cm] (s3) {$\langle {a} \rangle^\uparrow, \langle x \rangle^\downarrow$};
  \node[state, accepting, node distance=2.3cm, right of=s3] (s2) {$\langle \rangle_2$};
  % \node[state, accepting, node distance=4.5cm, right of=s2] (s3) {$\langle \rangle^\downarrow_3$};
  \node[below of=s1, node distance=1.2cm] (s5) {};

  \node[below of=s0,node distance=1.2cm] (s0bel) {};

  \draw (s0) edge [loop above] node[above] {$\twotrack{\pad}{\pad}/\twotrack{\pad}{\pad}$} (s0);
  \draw (s0) to node[above] {$\twotrack{x}{y}/\emptyword$} (s1);

  \draw[dotted] (s0) to (s0bel);
  \draw[dotted] (s1) to (s5);

  % \draw (s1) edge [loop above] node[above]
  %   {$\alpha/\alpha$; $\alpha \in \left\{\twotrack y y, \twotrack a y\right\}$} (s1);
  \draw (s1) to node[above] {$\twotrack{a}{x}/\twotrack{x}{y}$} (s3);
  \draw (s3) to node[above] {$\emptyword/\twotrack a x$} (s2);
  \draw (s2) edge [loop above] node[above]
    {$\twotrack \alpha \beta / \twotrack \alpha \beta$; $\alpha, \beta \in \left\{ a, y, \pad \right\}$} (s2);
  % \draw (s2) to node[above] {$\alpha/\alpha$; $\alpha \in \left\{\twotrack{y}{\pad}, \twotrack{a}{\pad}\right\}$} (s3);
  % \draw (s3) edge [loop above] node[above] {$\twotrack{a}{\pad}/\twotrack{a}{\pad}, \twotrack{\pad}{\pad}/\twotrack{\pad}{\pad}$} (s3);

\end{tikzpicture}
    \vspace{-3mm}
		\caption{A part of $T_{\xtoyx}^{\leq 2}$}
		\label{fig:aut-ex-trans2}
  \end{subfigure}
  \caption{Input of \cref{alg:rmc} for solving the word equation $x a y = y x$.
  The size of $\T^{\weq}_{\xtoepsilon}$ and $\T^{\weq}_{\xtoyx}$ would be
  prohibitively large, so we only give relevant parts of transducers
  $T_{\xtoepsilon}^{\leq 2}$ and $T_{\xtoyx}^{\leq 2}$.
  Some transitions use shorthand notation with the obvious meaning.
  }
  \label{fig:rmc_setup}
\end{figure}
   %%%%%%%%%%%%%%%%

\begin{figure}[t]
\centering
\begin{subfigure}[b]{0.45\textwidth}
  \centering
  \begin{tikzpicture}
  [node distance=2cm,->,>=stealth,transform shape, scale=0.8]

  \tikzstyle{empty}=[]

  \node[state,initial] (s0) {};
  \node[state,right of=s0] (s1) {};
  %\node[state,above of=s1,yshift=-5mm] (s2) {};
  \node[state,below of=s1,yshift=5mm] (s3) {};
  \node[state,right of=s1] (s4) {};
  \node[state,right of=s4,accepting] (s5) {};

  \draw (s0) edge[in=125, out=55] node[above] {$\twotrack{a}{\pad}$} (s5);
  \draw (s0) edge node[above,xshift=2mm] {$\twotrack{x}{y}$} (s1);
  \draw (s0) edge node[above,right,yshift=2mm] {$\twotrack{a}{y}$} (s3);
  \draw (s1) edge node[above] {$\twotrack{a}{x}$} (s4);
  \draw (s3) edge node[above,yshift=-2mm,xshift=-4mm] {$\twotrack{x}{x}$} (s4);
  \draw (s3) edge[bend right] node[above] {$\twotrack{y}{\pad}$} (s5);
  \draw (s4) edge node[above] {$\twotrack{y}{\pad}$} (s5);
  %\draw (s2) edge[bend left] node[above] {$\twotrack{x}{\pad}$} (s5);
  \draw (s5) edge[loop above] node[above] {$\twotrack{\pad}{\pad}$} (s5);

\end{tikzpicture}
  \caption{$\aut_{\reachset_1}$}
  \label{fig:aut-ex-proc1}
\end{subfigure}
\hfill
\begin{subfigure}[b]{0.45\textwidth}
  \centering
  \begin{tikzpicture}
  [node distance=2cm,->,>=stealth,transform shape, scale=0.8]

  \tikzstyle{empty}=[]

  \node[state,initial] (s0) {};
  \node[state,right of=s0] (s1) {};
  \node[state,above of=s1,yshift=-10mm] (s2) {};
  \node[state,below of=s1,yshift=5mm] (s3) {};
  \node[state,right of=s1] (s4) {};
  \node[state,right of=s4,accepting] (s5) {};

  \draw (s0) edge[in=100, out=80] node[above] {$\twotrack{a}{\pad}$} (s5);
  \draw (s0) edge node[above] {$\twotrack{a}{x}$} (s2);
  \draw (s0) edge node[above,pos=0.9] {$\twotrack{x}{y}$} (s1);
  \draw (s0) edge node[above,right,yshift=2mm] {$\twotrack{a}{y}$} (s3);
  \draw (s1) edge node[above] {$\twotrack{a}{x}$} (s4);
  \draw (s3) edge node[above,yshift=-2mm,xshift=-4mm] {$\twotrack{x}{x}$} (s4);
  \draw (s3) edge[bend right] node[above] {$\twotrack{y}{\pad}$} (s5);
  \draw (s4) edge node[above,pos=0.2] {$\twotrack{y}{\pad}$} (s5);
  \draw (s2) edge[out=0] node[above] {$\twotrack{x}{\pad}$} (s5);
  \draw (s5) edge[loop above] node[above] {$\twotrack{\pad}{\pad}$} (s5);

\end{tikzpicture}
  \caption{$\aut_{\reachset_2}$}
  \label{fig:aut-ex-proc2}
\end{subfigure}
\caption{Examples of finite automata $\aut_{\reachset_1}$,
$\aut_{\reachset_2}$ representing sets of configurations reachable in one and
two steps, respectively, when solving the word equation $x a y = y x$
using~\cref{alg:rmc}.
  }
\end{figure}
   %%%%%%%%%%%%%%%%

\begin{example}
	Consider the word equation $\weq_{1}\colon x a y = y x$
	from~\cref{ssec:nielsenWordTransform}. We apply~\cref{alg:rmc} on the encoded
	equation $\I = \encodeof{x a y, y x} = \twotrack{x~a~y}{y~x~\pad} \twotrack
	\pad \pad ^{*}$.
  The inputs of the algorithm are in \cref{fig:rmc_setup}.
  In the first iteration of the main loop, the regular set
  $\procset$, represented by its minimal FA $\aut_{\reachset_1}$, is given
  in~\cref{fig:aut-ex-proc1} (the FA also corresponds to $\reachset_1$ and
  represents the set~$S_1$ from \cref{ex:lift_to_sets}).
  In particular, consider, e.g.,
	\[
	 	\Big(\twotrack{x~a~y}{y~x~\pad}, \twotrack{a~y~\pad}{y~\pad~\pad}\Big) \in \T^{\weq}_{\xtoepsilon}.
	\]
	After saturation we get that $\twotrack{a~y}{y~\pad} \in \reachset_1$.
	In the second (and also the last)
	iteration, the set $\procset$ is given by the minimal FA~$\aut_{\reachset_2}$
	in~\cref{fig:aut-ex-proc2} (which also corresponds to $\reachset_2$ and
	$\reachset_3$, and the set~$S_2$ from \cref{ex:lift_to_sets}).
  Since $\reachset_3 \subseteq \procset$, the algorithm
	terminates with $\false$, establishing the unsatisfiability of~$\weq_{1}$.
  \qed
\end{example}

%%%%%%%%%%%%%%%%%%%%%%%%%%%%%%%%%%%%%%%%%%%%%%%%%%%%%%%%%%%%%%%%%%%%%%%%%%%%%%%%
\vspace{-0.0mm}
\section{Solving a System of Word Equations using RMC}\label{section:conjunction}
\vspace{-0.0mm}
%%%%%%%%%%%%%%%%%%%%%%%%%%%%%%%%%%%%%%%%%%%%%%%%%%%%%%%%%%%%%%%%%%%%%%%%%%%%%%%%

In the previous section we described how to solve a single quadratic word equation in the
RMC framework. In this section we focus on an extension of this approach to
handle a system of word equations of the form
\begin{equation}
\sweq\colon \lhs^{1} =
\rhs^{1} \wedge \lhs^{2}=\rhs^{2} \wedge \ldots \wedge \lhs^{n} =\rhs^{n}.
\end{equation}
In the first step we need to encode the system $\sweq$ as a regular language.
For this we extend the $\encode$ function to a system of word equations by defining
\begin{equation}
	\encode(\sweq) = \encode(\lhs^{1}, \rhs^{1}). \big\{\twotrack{\deli}{\deli}\big\}.
		\ldots . \big\{\twotrack{\deli}{\deli}\big\}. \encode(\lhs^{n}, \rhs^{n}),
\end{equation}
where $\deli$ is a~delimiter symbol, $\deli \notin \msoalphpad$. From \cref{lem:eq-enc} we
know that $\encode(\lhs^{i}, \rhs^{i})$ is regular for all $1 \leq i\leq n$.
Moreover, since regular languages are closed under concatenation
(\cref{prop:rat-prop}), the set $\encode(\sweq)$ is also regular.
Because each equation is now separated by a~delimiter,
we need to extend the destination set to
$\left\{\twotrack{\pad}{\pad},\twotrack{\deli}{\deli}\right\}^{*}$.

For the transition relation, we need to extend $\tau^{\leq i}_{\xtoalphax}$ and
$\tau^{\leq i}_{\xtoepsilon}$ from the previous section to support delimiters.
An application of a~rule
$\xtoalphax$ on a~system of equations can be described as follows: the rule
$\xtoalphax$ is applied to the first non-empty equation and the rest of the equations
are modified according to the substitution $x\mapsto \alpha x$.
The substitution on the other equations
is performed regardless of their first symbols.
The procedure is analogous for the rule $\xtoepsilon$.
A~series of applications of the rules can reduce the number of
equations, which then leads to a~string in our encoding with a~prefix from
$\big\{\twotrack{\pad}{\pad},\twotrack{\deli}{\deli}\big\}^{*}$.
The relation implementing $\xtoalphax$ or $\xtoepsilon$ on an encoded system of
equations skips this prefix.
Formally, the rule $\xtoalphax$ for a~system of equations where every equation
has at most~$i$
occurrences of every variable is given by the following relation:
\begin{align}
	T_{\xtoalphax}^{\seq, i} &= T_{\pass}.T_{\xtoalphax}^{\leq i}.\Big(
		\big\{ \twotrack{\deli}{\deli}\mapsto\twotrack{\deli}{\deli} \big\} .
		(T_{\trim}\circ S_{\xsubalphax}^{\leq i}) \Big)^{*},
\end{align}
where $T_{\pass}=\big\{\twotrack{\pad}{\pad} \mapsto\twotrack{\pad}{\pad},
\twotrack{\deli}{\deli}\mapsto\twotrack{\deli}{\deli}\big\}^{*}$ and
$S_{\xsubalphax}^{\leq i}$ is the binary transducer representing the
formula $\psi'^{\leq i}_{x \mapsto \alpha x}$ from \cref{sssec:enc_rational}
(which does not look at the first symbol on the tape).
The relation
$T_{\xtoepsilon}^{\seq,i}$ is defined in a~similar manner.
By construction and from the closure properties of rational relations (cf.\
\cref{prop:rat-prop}), it is clear that~$T_{\xtoalphax}^{\seq,i}$
and~$T_{\xtoepsilon}^{\seq,i}$ are rational.

%%%%%%%%%%%%%%%%%%%%%%%%%%%%%%%%%%%%%%%%%%%%%%%%%%%%%%%%%%%%%%%%%%%%%%%%%%%%%%%%
\vspace{-0.0mm}
\subsection{Quadratic Case}\label{sec:quad-system}
\vspace{-0.0mm}
%%%%%%%%%%%%%%%%%%%%%%%%%%%%%%%%%%%%%%%%%%%%%%%%%%%%%%%%%%%%%%%%%%%%%%%%%%%%%%%%

When $\sweq$ is quadratic, its satisfiability problem can be reduced to an RMC
problem
$(\I^{\sqeq}_{\sweq},\T^{\sqeq}_{\xtoalphax}\cup\T^{\sqeq}_{\xtoepsilon},\D^{\sqeq})$
instantiating \cref{alg:rmc} where
\begin{align*}
	\I^{\sqeq}_{\sweq} = \encode(\sweq) \qquad \T^{\sqeq}_{\xtoalphax} &=
  \hspace*{-3.5mm}\bigcup_{y\in\vars, \alpha \in\svars}\hspace*{-3.5mm}T_{\ytoalphay}^{\seq, 2} \qquad \D^{\sqeq} = \big\{\twotrack{\pad}{\pad},\twotrack{\deli}{\deli}\big\}^{*}\\
	\T^{\sqeq}_{\xtoepsilon} &= \bigcup_{y\in\vars}T_{\ytoepsilon}^{\seq, 2}
\end{align*}
Rationality of $\T^{\sqeq}_{\xtoalphax}$ and $\T^{\sqeq}_{\xtoepsilon}$ follows
directly from \cref{prop:rat-prop}.
In addition, we also need to modify \cref{alg:rmc} such that in
Line~\ref{algln:tr}, we substitute~$\saturate$ with~$\saturatenew$ defined as follows:
\begin{equation*}
	\saturatenew(L) = \left\{ u_1.u_2\ldots u_k ~\middle|~ w \in L,
	w \in u_1.\twotrack \pad \pad^+ .u_2.\twotrack \pad \pad^+ \ldots u_k \text{ for some } k \in \nat  \right\}.
\end{equation*}
Intuitively, $\saturatenew(L)$ modifies $\saturate$ to take into account the
fact that we now work with several word equations encoded into a~single word,
where they are separated by delimiters.
Now, $\saturatenew$ does not only remove paddings at the end of the word (when
$k=1$), but also in the middle.

The soundness and completeness of
our procedure for a~system of quadratic word equations is summarized by the
following lemma.
\begin{lemma}\label{lem:quad-sweq-corr}
  \cref{alg:rmc} instantiated with
  $(\I^{\sqeq}_{\sweq},\T^{\sqeq}_{\xtoalphax}\cup\T^{\sqeq}_{\xtoepsilon},\D^{\sqeq})$
  is sound, complete, and terminating if $\sweq$ is quadratic.
\end{lemma}
\begin{proof}
	We encode the nodes of a Nielsen proof graph as strings. The initial node
	$\sweq$ corresponds to a string from $\I^{\sqeq}_{\sweq}$ (conjunction can be
	seen as the delimiter~$\deli$). Because of the padding, the final node $\epsilon =
	\epsilon\wedge\dots \wedge \epsilon = \epsilon$ corresponds to a string from
	$\D^{\sqeq}$. The relations $\T^{\sqeq}_{\xtoalphax}$ and
	$\T^{\sqeq}_{\xtoepsilon}$ implement the Nielsen rules $\xtoalphax$ and
	$\xtoepsilon$. From~\cref{lem:Nielsen} we have that the Nielsen proof graph is
	finite (and hence the potential final node is in a finite depth). Since our approach
	implements the BFS strategy, it is both sound, complete, and terminating.
\end{proof}

%%%%%%%%%%%%%%%%%%%%%%%%%%%%%%%%%%%%%%%%%%%%%%%%%%%%%%%%%%%%%%%%%%%%%%%%%%%%%%%%
\vspace{-0.0mm}
\subsection{General Case}\label{sec:gen-case}
\vspace{-0.0mm}
%%%%%%%%%%%%%%%%%%%%%%%%%%%%%%%%%%%%%%%%%%%%%%%%%%%%%%%%%%%%%%%%%%%%%%%%%%%%%%%%

Let us now consider the general case when the system $\sweq$ is not quadratic.
In this section, we show that this general case is also reducible to an extended
version of RMC.

We first apply \cref{alg:trans} to a~general system of string
constraints~$\sweq$ in order to get an equisatisfiable cubic system of word
equations~$\sweq'$.
If the input of the transformation is a~system of equations with~$n$
symbols, then the output of the transformation will, in the worst case, contain
$\frac n 2$ additional word equations and $\frac n 2$ additional literals, so
the transformation is linear.
Then, we can use
the transition relations $T_{\xtoalphax}^{\seq,3} $ and $T_{\xtoepsilon}^{\seq,3}$ to
construct transformations of the encoded system~$\sweq'$.

\begin{algorithm}[t]
  \caption{Transformation to a cubic system of equations}
  \label{alg:trans}
	\KwIn{System of word equations $\Phi$}
	\KwOut{Equisatisfiable cubic system of word equations $\Psi$}
  $\Psi := \Phi$\;
	\While{$\exists x \in \vars$ s.t.\ $x$ occurs more than three times in $\Psi$}{
		Replace two occurrences of~$x$ in $\sweq$ by a~fresh word variable~$x'$ to obtain a~new system $\Psi'$\;
		$\Psi := (\Psi' \land x = x')$\;
	}
	\Return{$\Psi$}\;
\end{algorithm}

\begin{lemma}
\label{lem:tocubic}
Any system of word equations can be transformed by \cref{alg:trans} to an equisatisfiable cubic
system of word equations.
\end{lemma}
\begin{proof}
	Let $\sweq$ be the input system of word equations. Observe that in every
	iteration of \cref{alg:trans}, the number of occurrences of
	a~variable~$x$ is decreased by one and a~new variable~$x'$ with exactly three
	occurrences is introduced.
\end{proof}

One more issue we need to solve is to make sure that we work with a~cubic system
of word equations in every step of our algorithm. It may happen that
a~transformation of the type $\xtoyx$ increases the number of occurrences of the
variable~$y$ by one, so if there had already been three occurrence of~$y$ before
the transformation, the result will not be cubic any more.
More specifically, assume a~cubic system of word equations $x\concat \lhs = y
\concat \rhs\wedge \sweq$, where $x$ and $y$ are distinct word variables and $\lhs$ and
$\rhs$ are word terms. If we apply the transformation~$\xtoyx$,  we will obtain
$x( \substof \lhs x {yx}) = \substof \rhs x {yx}\wedge \substof \sweq x {yx} $.
Observe that
\begin{inparaenum}[(i)]
  \item  the number of occurrences of $y$ is first \emph{reduced by one}
    because the first $y$ on the right-hand side of $x . \lhs = y . \rhs$ is
    removed and
  \item  the number of occurrences of~$y$ can be at most \emph{increased by
    two} because there exist at most two occurrences of~$x$ in~$\lhs$, $\rhs$,
    and~$\sweq$.
\end{inparaenum}
Therefore, after the transformation $\xtoyx$, a cubic system of
word equations might become a \emph{($y$-)quartic system of word equations} (at
most four occurrences of the variable~$y$ and at most three occurrences of any
other variable). For this reason, we need to apply the conversion to the cubic
system after each transformation.

Given a~fresh variable~$v$, we use $\cut_{v}$ to denote the transformation from
a~single-quartic system of word equations to a~cubic system of equations using
the fresh variable $v$.

\begin{lemma}
\label{lem:sq2cb}
	The relation $T_{\cut_{v}}$ performing the transformation $\cut_{v}$ on an encoded
  single-quartic system of equations is rational.
\end{lemma}
\begin{proof}
  We show how we can create a~transducer for the transformation
	from a~single-quartic system of word equations to a cubic system of word
	equations.

	In the first step, we create the transducer~$\T^{sq}_{x,x_{i}}$ that accepts only
	input that is an encoding of a $x$-quartic system of word equations. This can
	be done by using states to trace the number of occurrences of variables (we only
	need to count up to four). For an encoding of a~$x$-quartic system of word
	equations, the transducer~$\T^{sq}_{x,x_{i}}$ returns an encoding that is
	obtained by replacing the first two occurrences of~$x$ from the input to~$x_{i}$ and
	concatenating the language
	$\twotrack{\deli}{\deli}\twotrack{x}{x_{i}}\twotrack{\pad}{\pad}^{*}$ at the end.

	In the second step, we create the transducer $\T_{\mathit{cub}}$ that accepts
	only encodings of a cubic system of word equations and returns the same
	encodings.
	Now we have
	\begin{equation}
		T_{\cut_{v}} =\langof{\T_{\mathit{cub}}} \cup \bigcup_{x\in\vars} \langof{\T^{sq}_{x,v}}.
	\end{equation}
	The lemma then follows by \cref{prop:rat-prop}.
\end{proof}
To express solving a~system of string constraints~$\sweq$ in the terms of
a~(modified) RMC, we first convert~$\sweq$ (using \cref{alg:trans}) to
an equisatisfiable cubic system~$\sweq'$. The satisfiability of a~system of word
equations~$\sweq$ can be reduced to a~modified RMC problem
$(\I^{\seq}_{\sweq},\T^{v_{i},\seq}_{\xtoalphax}\cup\T^{v_{i},\seq}_{\xtoepsilon},\D^{\seq})$
instantiating \cref{alg:rmc} with the following components:
\begin{align*}
	\I^{\seq}_{\sweq} &= \encode(\Phi') &
  \T^{v, \seq}_{\xtoalphax} &= T_{\cut_{v}} \circ  \hspace*{-3.5mm}\bigcup_{y\in\vars,
    \alpha\in\svars}\hspace*{-3.5mm}T_{\ytoalphay}^{\seq, 3} &
	\D^{\seq} &= \big\{\twotrack{\pad}{\pad},\twotrack{\deli}{\deli}\big\}^{*}\\
  && \T^{v, \seq}_{\xtoepsilon} &= T_{\cut_{v}} \circ  \bigcup_{y\in\vars}T_{\ytoepsilon}^{\seq, 3}
\end{align*}

\begin{algorithm}[t]
  \caption{Solving a~general string constraint $\varphi$ using RMC}
  \label{alg:rmc-modified}
	\SetKwComment{ctriang}{\scalebox{1.3}{\textcolor{gray}{$\triangleright$\,}}}{}
	\KwIn{Encoding $\I$ of a formula $\varphi$ (the initial set),\\
	\hspace*{13mm}ordered set of indices $\varsV = \{ v_1, v_2, \dots \}$,\\
  \hspace*{13mm}parameterized transformers $\T^{v}_{\xtoalphax}$, $\T^{v}_{\xtoepsilon}$ where $v \in \varsV$, and \\
  \hspace*{13mm}the destination set $\D$}
	\KwOut{A model of $\varphi$ if $\varphi$ is satisfiable, $\false$ otherwise}
  $\reachset_0 := \I$\;
  $\procset := \emptyset$\;
	$\T^{v} := \T^{v}_{\xtoalphax} \cup \T^{v}_{\xtoepsilon}$\;
  $i := 0$\;
  \While{$\reachset_{i} \not\subseteq \procset$}
  {
    \If{$\D \cap \reachset_{i} \neq \emptyset$\label{algln:touch-modified}}
    {
      \Return{$\extractmodel(\{\T^{v_j}\}_{j \in \{ 0, \dots, i \}}, \D, \reachset_0, \ldots, \reachset_{i})$}\;
    }
    $\procset := \procset \cup \reachset_{i}$\;
    $\reachset_{i+1} := \saturatenew\circ\T^{v_{i}}(\reachset_{i})$\label{algln:tr-modified}\;
		$\vars := \vars \cup \{ v_{i} \}$\label{algln:update-modified}\;
    $i$++\;
  }
  \Return{$\false$}\;
\end{algorithm}

For the modified RMC algorithm, we need to assume that $\varsV = \{ v_1, v_2,
\dots \} \cap \svars = \emptyset$, where~$\varsV$ is a~set of \emph{fresh index variables}.
We also need to update Line~\ref{algln:sett} of \cref{alg:rmc}
to $\T^{v} := \T^{v}_{\xtoalphax} \cup \T^{v}_{\xtoepsilon}$ and
Line~\ref{algln:tr} to $\reachset_{i+1} := \T^{v_{i}}(\reachset_{i});~\vars :=
\vars\cup\{v_{i}\};$ to allow using a~new variable~$v_{i}$ in every iteration
(here, in every iteration of the algorithm, $\T^{v}$~will be instantiated with
a~new value of the~$v$ parameter).
The entire algorithm is shown in~\cref{alg:rmc-modified}. Rationality of $\T^{v,
\seq}_{\xtoalphax}$ and $\T^{v, \seq}_{\xtoepsilon}$ follows directly from
\cref{prop:rat-prop}.

\begin{lemma}
\label{lem:sweq-sound}
  \cref{alg:rmc-modified} instantiated with
  $(\I^{\seq}_{\sweq},\varsV,\T^{v,\seq}_{\xtoalphax}\cup\T^{v,\seq}_{\xtoepsilon},\D^{\seq})$ is
  sound if $\sweq$ is cubic.
\end{lemma}
\begin{proof}
	We again encode the nodes of a Nielsen proof graph as strings. The initial and
	final node correspond to strings from the encoded initial $\I^{\seq}_{\sweq}$ and
	final language $\D^{\seq}$, respectively. The
	relations $\T^{v,\seq}_{\xtoalphax}$ and $\T^{v,\seq}_{\xtoepsilon}$ implement
	the Nielsen rules $\xtoalphax$ and $\xtoepsilon$. For an arbitrary system of
	word equations the Nielsen proof graph may be infinite. However, since the
	transformation $\cut_{v}$ preserves satisfiability, the procedure is sound.
\end{proof}

\paragraph{Completeness}
Since the previous approach can in each step introduce a new equation (due to
the transducer $T_{C_v}$ transforming the system of equations to a~cubic one),
completeness is not guaranteed in general. In order to get a~sound and
complete procedure for cubic equations, it is necessary to use transducers with
an increasing number of symbols being rewritten. In particular, we need to use transducers implementing the following relations:
\begin{align*}
  \T^{i, \seq\star}_{\xtoalphax} &= \hspace*{-3.5mm}\bigcup_{y\in\vars,
    \alpha\in\svars}\hspace*{-3.5mm}T_{\ytoalphay}^{\seq, 2^{i+1}} &
	 \T^{i, \seq\star}_{\xtoepsilon} &= \bigcup_{y\in\vars}T_{\ytoepsilon}^{\seq, 2^{i+1}}
\end{align*}
with~$i \in \nat_1$ being the number of iteration
(moreover, it is also necessary to skip Line~\ref{algln:update-modified} in \cref{alg:rmc-modified}). After each
step, the maximum number of occurrences of a~variable in the system is
in the worst case multiplied by two.
Therefore, in the~$i$-th step, the number of occurrence of a~variable in the
system can be up to~$2^{i+1}$, which can be handled by
$\T^{i, \seq\star}_{\xtoalphax}$
and
$\T^{i, \seq\star}_{\xtoepsilon}$.

\begin{lemma}
\label{lem:sweq-complete}
  \cref{alg:rmc-modified} instantiated with
  $(\I^{\seq}_{\sweq}, \nat,\T^{i,\seq\star}_{\xtoalphax}\cup\T^{i,\seq\star}_{\xtoepsilon},\D^{\seq})$
  and modified as described above is sound and complete if~$\sweq$ is cubic.
\end{lemma}
\begin{proof}
	As in previous cases, we encode the nodes of a Nielsen proof graph as
	strings. Since the transducer $T_{\ytoalphay}^{\seq, \ell}$ correctly
	implements the Nielsen rule $\ytoalphay$ for systems where each variable
	occurs at most $\ell$~times, it suffices to show that in the $i$-th iteration
	in~\cref{alg:rmc-modified}, the number of occurrences of each variable is
	bounded by $2^{i+1}$. The soundness and completeness then follows from the BFS
	strategy of \cref{alg:rmc-modified} and \cref{lem:Nielsen}.

	Consider a system of equations $\sweq$ s.t.\ each variable occurs at most
	$n$~times. Application of the rule of the form $\ytoepsilon$ does not increase
	the number of occurrences of $y$. The situation is different for a rule of the
	form $\ytoalphay$ s.t. $\alpha \in \vars$.
  Furthermore, observe that there are at
	most $n$ occurrences of~$y$ and $(n-1)+(n-1) = 2n-2$ occurrences of~$\alpha$
	in the modified system (each~$y$ is replaced by~$\alpha y$ and the first
	occurrence of $\alpha$ is removed from the modified system). The maximum
	number of occurrences of variables in the $i$-th iteration of
	\cref{alg:rmc-modified} is hence bounded by $2^{i+1}$ (note that this is a 
	loose upper bound, since the number of occurrences increases to $2n-2\leq 2n$ 
	in the worst case), provided that~$\sweq$ is
	cubic. 
\end{proof}

%------------------------------------------------------------------------------
\paragraph{Termination}
Since the Nielsen transformation does not guarantee termination for the general case,
neither does our algorithm.
Investigation of possible symbolic encodings of complete algorithms, e.g.\
Makanin's algorithm (\cite{makanin1977problem}), is our future work.

%%%%%%%%%%%%%%%%%%%%%%%%%%%%%%%%%%%%%%%%%%%%%%%%%%%%%%%%%%%%%%%%%%%%%%%%%%%%%%%%
\vspace{-0.0mm}
\section{Handling Boolean Combination of String Constraints}\label{section:full}
\vspace{-0.0mm}
%%%%%%%%%%%%%%%%%%%%%%%%%%%%%%%%%%%%%%%%%%%%%%%%%%%%%%%%%%%%%%%%%%%%%%%%%%%%%%%%
%
In this section, we will extend the procedure from handling a~\emph{conjunction}
of word equations into a~procedure that handles their arbitrary Boolean
combination.
An obvious approach is by combining the solutions we have given in
\cref{section:we_rmc,section:conjunction} with standard DPLL(T)-based solvers
and use our procedure to handle the string theory.
We can, however, solve the whole formula with our procedure by using the
encoding proposed in this section.
Although we do not have hope that the presented solution can compete with the
highly-optimized DPLL(T)-based solvers, it (also taking into account the extensions
from \cref{sec:rmc-extension}) makes our framework more robust, by having
a~homogeneous automata-based encoding for a~quite general class of constraints.
When one, e.g., tries to extend the approach by using abstraction
(cf.~\cite{BouajjaniHRV12}) to accelerate
termination or by \emph{learning} the invariant in the spirit
of~\cite{NeiderJ13}, they can then still treat the encoding of the whole system of
constraints uniformly within the framework of RMC.

The negation of word equations can be handled in the standard way. For instance,
we can use the approach given by~\cite{abdulla2014string} to convert a~negated word
equation $\lhs \neq \rhs$ to the following string constraint:
\begin{equation}\label{eq:neq}
	\bigvee_ {c\in \Sigma} (\lhs= \rhs \cdot c x  \vee \lhs \cdot c x = \rhs)
	\vee \bigvee_ {c_1, c_2\in \Sigma, c_1\neq c_2} (\lhs = y c_1 x_1 \wedge
	\rhs = y c_2 x_2)
\end{equation}
The first part of the constraint says that either $\lhs$ is a strict prefix of
$\rhs$ or the other way around. The second part says that $\lhs$ and $\rhs$
have a~common prefix~$y$ and start to differ in the next characters $c_1$ and
$c_2$.
For word equations connected using $\land$ and $\lor$, we apply distributive
laws to obtain an equivalent formula in the conjunctive normal form (CNF) whose
size is at worst exponential in the size of the original formula.
Note that we cannot use the Tseitin transformation~\cite{tseitin1983complexity},
since it may introduce fresh negated variables and their removal
using~\cref{eq:neq} would destroy the CNF form.

Let us now focus on how to express solving a string constraint $\sweq$ composed
of an arbitrary Boolean combination of word equations using a~(modified) RMC.
We start by
removing inequalities in $\sweq$ using \cref{eq:neq}, then we convert the system without
inequalities into CNF, and, finally, we apply the procedure in
\cref{lem:tocubic} to convert the CNF formula to an equisatisfiable and
cubic CNF~$\sweq'$.
For deciding satisfiability of $\sweq'$ in the terms of RMC, both the transition
relations and the destination set remain the same as in the previous
section (general case). The only difference is the initial configuration because
the system is not a conjunction of terms any more but rather a~general formula in
CNF.
For this, we extend the definition of $\encode$ to a~clause $c \colon
(\lhs^1=\rhs^1 \vee \ldots \vee \lhs^{n} = \rhs^{n})$ as
$\encodeof{c}=\bigcup_{1 \leq j \leq n }\encodeof{\lhs^j,\rhs^j}$.
Then the initial configuration for~$\sweq'$ is given as
\begin{equation}
 \I^{\sct}_{\sweq'} = \encode(c_1) . \big\{\twotrack{\deli}{\deli}\big\} . \ldots . \big\{\twotrack{\deli}{\deli}\big\} . \encode(c_m),
\end{equation}
where $\sweq'$ is 
of the form $\sweq'\colon c_1 \land \ldots \land c_m$ and each clause $c_{i}\colon
(\lhs^1=\rhs^1 \vee \ldots \vee \lhs^{n_{i}} = \rhs^{n_{i}})$.
We obtain the following lemma directly from \cref{prop:rat-prop}.
\begin{lemma}
	The initial set $\I^{\sct}_{\sweq'}$ is regular.
\end{lemma}
The transition relation and the destination set are the same as the ones in the
previous section, i.e., $\T^{v,\sct}_{\xtoalphax} =
\T^{v,\seq}_{\xtoalphax}$, $\T^{v,\sct}_{\xtoepsilon}=\T^{v,\seq}_{\xtoepsilon}$,
and $\D^{\sct}= \D^{\seq}$. The soundness of our procedure for 
a~Boolean combination of word equations is summarized by the following lemma.
The completeness can be achieved, as in~\cref{sec:gen-case}, by transducers
with an increasing number of rewritten symbols.
\begin{lemma}
	Given a Boolean combination of word equations $\sweq$,
  \cref{alg:rmc-modified} instantiated with
  $(\I^{\sct}_{\sweq'},\varsV,\T^{v,\sct}_{\xtoalphax}\cup\T^{v,\sct}_{\xtoepsilon},\D^{\sct})$ is
  sound.
\end{lemma}
\begin{proof}
  A system of full word equations can be converted according to the
	steps described above to an equisatisfiable system in CNF $\Psi \colon \bigwedge_{i
	= 1}^{n} c_{i}$ where every~$c_{i}$ is a~disjunction of equalities.
  Then, $\Psi$ is satisfiable
	if there is some $\phi\colon \bigwedge_{i = 1}^{n} \lhs^{i} = \rhs^{i}$ where $(\lhs^{i} =
	\rhs^{i}) \in c_{i}$ for all~$i$.
  Moreover, we have $\encodeof{\phi} \in \I^{\sct}_{\Phi}$. From
	\cref{lem:sweq-sound} (and from the BFS strategy of RMC), we get that our
	algorithm is sound in proving $\sweq$ is satisfiable.
\end{proof}

Completeness is guaranteed if we consider the transducers $\T^{i, \seq\star}$ for $\sweq$ in the same way as in \cref{sec:gen-case}.
Regarding termination, it cannot be guaranteed, due to the corresponding result in \cref{sec:gen-case} that holds for a special case of the Boolean combination of string equations we consider here.

%%%%%%%%%%%%%%%%%%%%%%%%%%%%%%%%%%%%%%%%%%%%%%%%%%%%%%%%%%%%%%%%%%%%%%%%%%%%%%%%
\vspace{-0.0mm}
\section{Extensions}
\label{sec:rmc-extension}
\vspace{-0.0mm}
%%%%%%%%%%%%%%%%%%%%%%%%%%%%%%%%%%%%%%%%%%%%%%%%%%%%%%%%%%%%%%%%%%%%%%%%%%%%%%%%
In this section, we discuss how to extend our RMC-based framework to support the following
two types of \emph{atomic constraints}:
\begin{enumerate}[(i)]
  \item  A \emph{length constraint} $\varphi_{i}$ is a formula of Presburger
    arithmetic over the values of $|x|$ for $x \in \vars$, where $|\,\cdot\,|
    \colon \vars \to \nat$ is the word length function (to simplify the notation
    we use a formula of Presburger arithmetic with free variables~$\vars$ and we
    keep in mind that the value assigned to $x\in\vars$ corresponds in fact to
    $|x|$).

  \item  A \emph{regular constraint} $\varphi_r$ is a conjunction of atoms of the
		form $x \in \langof{\aut}$ (or their negation) where $x$ is a word variable
		and $\aut$ is an FA representing a regular language.
\end{enumerate}

%*******************************************************************************
\vspace{-0.0mm}
\subsection{Length Constraints}
\label{sec:length-constraints}
\vspace{-0.0mm}
%*******************************************************************************
%
In order to extend our framework to solve word equations with length
constraints, we encode them as regular languages, and we encode the effect of
Nielsen transformations on the lengths of variables as regular relations.
Let us start with defining \emph{atomic length constraints}:
\begin{align*}
  \varphi_{\mathit{len}} ::=
  a_1 x_1 + \cdots + a_n x_n \leq c
\end{align*}
for string variables $x_1, \ldots, x_n \in \vars$ and integers $a_1, \ldots,
a_n, c \in \mathbb{Z}$ (we will also use formulae in a~less restricted form,
which can always be translated to the defined one using standard arithmetic
rules).
Given a~variable assignment $I\colon \vars \to \Sigma^{*}$, it~holds that~$I$ is a~model
of~$\varphi_{\mathit{len}}$, written as $I \models \varphi_{\mathit{len}}$, iff
$a_1 \cdot |I(x_1)| + \cdots + a_n \cdot |I(x_n)| \leq c$.
We note that the satisfiability of a~string constraint with only atomic length constraints connected via
Boolean connectives (i.e., no word equations) corresponds to the satisfiability of
a~Boolean combination of constraints in \emph{integer linear programming},
which is an \clNP-complete problem~(\cite{Karp72}).

We will show that length constraints can be encoded into our framework using
standard automata-based techniques for dealing with constraints
in Presburger arithmetic~(\cite{Presburger29,WolperB00,BoigelotW02}).
First, let us define how a~first-order variable ranging over~$\nat$ is
represented in \msostr.
Let $\lsbf\colon \nat \to 2^\nat$ be a~function representing the
\emph{least-significant bit first} binary encoding of a~number such that for $n
\in \nat$, we define $\lsbfof n$ to be the finite set~$S \subseteq \nat$ for
which $n = \sum_{i\in S} 2^{i}$.
For instance, $\lsbfof{42} = \{1, 3, 5\}$ because $42 = 2^1 + 2^3 + 2^5$.
Moreover, we define the \emph{positional} least-significant bit first binary encoding of a number $n\in\nat$ as $\lsbfp(n) = \{ \ell + 1 \mid \ell \in \lsbf(n) \}$.

\begin{proposition}\label{prop:presb-as-mso}
Let $\varphi_{\mathit{len}}(x_1, \ldots, x_n)$ be an atomic length constraint.
Then there exists an \msostr formula $\psi_{\mathit{len}}(X_1, \ldots, X_n)$
with free position variables $X_1, \ldots, X_n$
such that an assignment $\sigma\colon \{x_1, \ldots, x_n\} \to \nat$ is
a~model of~$\varphi_{\mathit{len}}$ iff the assignment $\sigma' =
\{X_{i} \mapsto \lsbfp(v_{i}) \mid \sigma(x_{i}) = v_{i}\}$ is a~model
of~$\psi_{\mathit{len}}$.
\end{proposition}

\begin{proof}
A~possible encoding of $\varphi_{\mathit{len}}$ into $\psi_{\mathit{len}}$ is
given, e.g., in~\cite{GlennG96}.
\end{proof}

Recall that in automata-based approaches to Presburger arithmetic (such as
\cite{GlennG96,WolperB00,BoigelotW02}), a~formula~$\varphi$ with~$k$ free variables is
translated into an automaton~$\autof \varphi$ over the alphabet~$\bbB^k$ for
$\bbB = \{0,1\}$.
A~model of~$\varphi$ is represented as a~word $w \in \big(\bbB^k\big)^{\!*}$ in the
language of~$\autof \varphi$ such that projecting the track for variable~$x_{i}$
from~$w$ gives us the $\lsbf$ encoding of the value of~$x_{i}$ in the model.
For instance, if $\bintrackfst x y {1~0~0~1}{1~0~1~0} \in \langof{\autof
\varphi}$, then the assignment $\{x \mapsto 9 , y \mapsto 5\}$ is a~model
of~$\varphi$ because $1001$ is a~LSBF binary encoding of the number~$9$ and
$1010$ encodes the number~$5$.

In order to encode dealing with length constraints into our framework,
\cref{prop:presb-as-mso} is not sufficient: we also need to be able to represent
how the transformations modify those constraints and how the constraints
restrict the space of possible solutions.
In the following paragraphs, we provide the details about our approach.

Consider a~word~$w_\sigma$ encoding an~assignment~$\sigma \colon \vars \to \nat$. 
The transformation
$\xtoyx$ for $x,y\in\vars$ applied to $w_\sigma$ produces a word $w_{\sigma'}$
encoding the assignment $\sigma' = \sigma \triangleleft \{ x \mapsto \sigma(x) -
\sigma(y) \}$ if $\sigma(x) \geq \sigma(y)$, where $\sigma' = \sigma \triangleleft \{x \mapsto n\}$ is defined as $\sigma'(x) = n$ and $\sigma'(y) = \sigma(y)$ for all $y \neq x$.
The transformation $\xtoax$, for
$a\in\Sigma$ produces a word $w_{\sigma'}$ encoding the assignment $\sigma' =
\sigma \triangleleft \{ x \mapsto \sigma(x) - 1 \}$ if $\sigma(x) \geq 1$.
Finally, the transformation $\xtoepsilon$ does not change the word $w_\sigma$,
but imposes the restriction $\sigma(x) = 0$.
Formally, the transformations are described using the following formulae:
\begin{align}
  \begin{split}
    \varphi^{\mathit{len}}_{\xtoyx}(x, y, x') \defiff{}& x \geq y \land x' = x - y, \\
    \varphi^{\mathit{len}}_{\xtoax}(x, x') \defiff{}& x \geq 1 \land x' = x -1, \text{ and} \\
    \varphi^{\mathit{len}}_{\xtoepsilon}(x) \defiff {}& x = 0.
  \end{split}
\end{align}

From \cref{prop:presb-as-mso} and \cref{prop:msostr-power}, it follows that the
relations denoted by the formulae are regular.
We will denote the transducers encoding those relations as
$T^{\mathit{len}}_{\xtoyx}$, $T^{\mathit{len}}_{\xtoax}$, and
$T^{\mathit{len}}_{\xtoepsilon}$ respectively.

Let us now focus on how to adjust the initial and destination sets for an
equation with a~length constraint~$\varphi_{i}(\vars)$ with free variables~$\vars$.
The initial set is extended by all encoded models of~$\varphi_{i}$.
Formally, the part of the initial set related
to the length constraint is given as $\I_{\varphi_{i}} = \langof{\varphi_{i}}$
(which is a~subset of~$\big(\bbB^{|\vars|}\big)^{\!*}$) and the part of the
destination set as $\D_{\mathit{len}} = \big(\bbB^{|\vars|}\big)^{\!*}$.

Satisfiability of a~quadratic equation $\weq\colon \lhs = \rhs$ with the length
constraint~$\varphi_{i}$ can then be expressed as the RMC problem
$(\I_{\varphi_{i}}^{\len}, \T^{\len}_{\xtoalphax}\cup \T^{\len}_{\xtoepsilon},
\D_{\varphi_{i}}^{\len})$ instantiating \cref{alg:rmc} with items given
as follows (note the use of a~fresh delimiter~$\lensep$ for length constraints):
\begin{align*}
  \I_{\varphi_{i}}^{\len} ={} & \I^{\weq}\concat\{ \lensep \}\concat\I_{\varphi_{i}} \quad\quad\quad\quad
		\D_{\varphi_{i}}^{\len} = \D^{\weq} \concat \{ \lensep \} \concat \big(\bbB^{|\vars|}\big)^{\!*}\\
  \T^{\len}_{\xtoalphax} ={} & \hspace*{-3.5mm}\bigcup_{y\in\vars, z\in\vars}\hspace*{-3.5mm}
  T_{\ytozy}^{\leq 2} \concat \{ \lensep \mapsto\lensep \}
  \concat T^{\mathit{len}}_{\ytozy} \cup {}\\
	&\hspace*{-3.5mm}\bigcup_{y\in\vars, a\in\Sigma}\hspace*{-3.5mm}
  T_{\ytoay}^{\leq 2} \concat \{ \lensep \mapsto\lensep \} \concat T^{\mathit{len}}_{\ytoay} \\
  \T^{\len}_{\xtoepsilon} ={} & \bigcup_{y\in\vars} T_{\ytoepsilon}^{\leq 2} \concat \{ \lensep \mapsto\lensep \} \concat T^{\mathit{len}}_{\ytoepsilon}
\end{align*}
Extensions to a~system of equations and (more generally) a Boolean combination of constraints can
be done in the same manner as in \cref{section:conjunction,section:full}.

Rationality of $\T^{\len}_{\xtoalphax}$ and $\T^{\len}_{\xtoepsilon}$
follows directly from \cref{prop:rat-prop}.
The soundness and completeness of our algorithm is summarized by \cref{lem:len-sound}.
Termination is an open problem even for quadratic
equations (cf.~\cite{buchi1990definability}).
\begin{lemma}\label{lem:len-sound}
  Given a quadratic word equation $\weq\colon \lhs = \rhs$ with the length
  constraint~$\varphi_{i}$, \cref{alg:rmc} instantiated with
  $(\I_{\varphi_{i}}^{\len}, \T^{\len}_{\xtoalphax}\cup \T^{\len}_{\xtoepsilon},
  \D_{\varphi_{i}}^{\len})$ is sound and complete.
\end{lemma}
\begin{proof}
  We can generalize nodes of the Nielsen proof graph to pairs of the
	form $(\lhs' = \rhs', f)$ where~$f$ is a mapping assigning lengths to
	variables from~$\vars$. The transformation rules can be straightforwardly
	generalized to take into account also the lengths. The initial nodes are pairs
	$(\lhs = \rhs, f)$ where $f$ is a~model of~$\varphi_{i}$. The final nodes are nodes
	$(\epsilon = \epsilon, g)$ where $g$ is arbitrary. Note that the
	generalized graph is not necessarily finite even for quadratic equations.
	Nevertheless, if the equation is satisfiable then there is a finite path from an
	initial node to a final node.

	Directly from the definition of $\I_{\varphi_{i}}^{\len}$ we have that the
	initial nodes of the generalized proof graph are encoded strings from
	$\I_{\varphi_{i}}^{\len}$ and the final nodes correspond to
	$\D_{\varphi_{i}}^{\len}$. We can also see that the transformation rules
	correspond to the encoded relations $\T^{\len}_{\xtoalphax}$ and
	$\T^{\len}_{\xtoepsilon}$. Since the search in \cref{alg:rmc}
	implements a BFS strategy, we get that our (semi-)algorithm is sound and complete in proving
	satisfiability.
\end{proof}

For the general (non-quadratic) case and the case of a~Boolean combination of
constraints, we can obtain a~sound and complete (though non-terminating)
procedure by using the transducers $\T^{i, \seq\star}$ in the same
way as in \cref{sec:gen-case} and modifying the proof of \cref{lem:len-sound}
accordingly.

%*******************************************************************************
\vspace{-0.0mm}
\subsection{Regular Constraints}
\label{sec:regular-constraints}
\vspace{-0.0mm}
%*******************************************************************************
%

Our second extension of the framework is the support of regular constraints as a
conjunction of atoms of the form $x\in\langof{\A}$ for an FA~$\A$ over~$\Sigma$ (note that
the negation of an atom $x\notin\langof{\A}$ can be converted to the positive
atom $x\in\mathcal{L}(\A^\complement)$ where $\A^\complement$ is a~complement of the FA~$\A$).
In particular, we assume that regular
constraints are represented by a~conjunction~$\varphi_r$ of~$\ell$ atoms of the
form
\begin{equation}
	\varphi_r \defiff \bigwedge_{i=1}^\ell (x_{i}\in \langof{\aut_{i}}),
\end{equation}
where~$\aut_{i}$ is an FA for each $1\leq i \leq \ell$. Without loss of
generality, we assume that the automata occurring in~$\varphi_r$ have pairwise
disjoint sets of states and, further, we use $\A_r = (Q_r,\Sigma,\delta_r,I_r,F_r)$ to
denote the automaton constructed as the disjoint union of all automata
occurring in formula~$\varphi_r$. 
Note that the disjoint union of two FAs 
$\aut_1 = (Q_1, \Sigma, \Delta_1, Q_i^1, Q_f^1) $ and 
$\aut_1 = (Q_2, \Sigma, \Delta_2, Q_i^2, Q_f^2)$ is the FA 
$\aut_1 \uplus \aut_2 = (Q_1 \uplus Q_2, \Sigma, \Delta_1 \uplus \Delta_2, Q_i^1 \uplus Q_i^2, Q_f^1 \uplus Q_f^2)$.

The approach we developed here is inspired by the approach in Norn (cf.\
\cite{abdulla2015norn}), but the idea needed to be significantly modified to
fit in our (more proof-based) framework.
In particular, we encode regular constraints as words over symbols of the form $\regsym{x,p,q}$
where $x\in\vars$ and $p,q \in Q_r$. We denote the set of all such symbols
as~$\allsymreg$.
Moreover, we treat the words as sets of symbols and hence we
assume a fixed linear order~$\preccurlyeq$ over symbols to allow a~unique
representation. In particular, for a word $w\in\allsymreg^{*}$ we use
$w_\preccurlyeq$ to denote the string containing symbols sorted by
$\preccurlyeq$ with no repetitions of symbols.
A single atom $x\in\langof{\A_{i}}$ for $\aut = (Q_{i}, \Sigma, \delta_{i}, I_{i},
F_{i})$ can be encoded as a set of words
$\encode(x\in\langof{\A_{i}}) = \{ \regsym{x,p,q} \mid p\in I_{i}, q \in
F_{i} \}$. 
The set represents all possible combinations of initial and final states in~$\A_{i}$. 
The initial set $\I_{\varphi_r}$ is then defined as
\begin{equation}
	\I_{\varphi_r} = \{ w_\preccurlyeq \in \allsymreg^{*} \mid w\in\encode(x_1\in \langof{\aut_1})\dots\encode(x_\ell\in \langof{\aut_\ell})\}.
\end{equation}
Note that~$\I_{\varphi_r}$ is finite for a~finite~$\vars$, therefore it is a
regular language.

Let us now describe the effect of the Nielsen transformation on the regular
constraint part.
Consider a~word~$w$ encoding a~set of symbols from~$\allsymreg$.
Then, the transformation $\xtoyx$ for $x,y\in\vars$ applied to~$w$
produces words $w'$ encoding sets where each occurrence of a~symbol
$\regsym{x,p,q}$ is replaced with all possible pairs of symbols~$\regsym{y,p,s}$
and~$\regsym{x,s,q}$ where $p\leadsto s$ and $s\leadsto q$ in~$\A_r$ (we use $q
\leadsto q'$ to denote that there is a~path from~$q$ to~$q'$ in the transition
diagram of~$\A_r$).
Similarly,
the transformation $\xtoax$ for $x\in\vars, a \in\Sigma$ applied to $w$ produces
words $w'$ encoding sets where each occurrence of a~symbol $\regsym{x,p,q}$ is
replaced with all possible symbols $\regsym{x,r,q}$ where $p\ltr{a} r$ in
$\A_r$.
Finally, by the transformation $\xtoepsilon$ we obtain a string $w' = w$
only if all symbols of $w$ related to the variable~$x$ are of the form
$\regsym{x,q,q}$ for $q\in Q$.
Formally, we first define the~function expanding a~single symbol for variables
$x$ and $y$ as
\begin{equation}
\expfuncof x y (\sigma) =
\begin{cases}
  \{\,\regsym{y,p,r}{.}\regsym{x,r,q} \mid p\runto r\runto q \text{ in }\A_r \,\} & \text{if } \sigma= \regsym{x,p,q}, \\
  \{\sigma\} & \text{otherwise.}
\end{cases}
\end{equation}
Similarly, we define the expansion
\begin{equation}
  \expfuncof x a (\sigma) =
  \begin{cases}
    \{\,\regsym{x,r,q} \mid p\ltr{a} r\runto q \text{ in }\A_r \,\} & \text{if } \sigma= \regsym{x,p,q}, \\
    \{\sigma\} & \text{otherwise.}
  \end{cases}
\end{equation}
Then, the transformations
$\xtoyx$, $\xtoax$, and $\xtoepsilon$ can be described by the following relations:
\begin{align}
  \begin{split}
    T_{\xtoyx}^{\mathit{reg}} & = \{\,(w,u_\preccurlyeq) \mid
    	u\in\expfuncof x y(w[1])\dots\expfuncof x y (w[|w|])\,\},\\
		T_{\xtoax}^{\mathit{reg}} & = \{\,(w,u_\preccurlyeq) \mid
	    u\in\expfuncof x a(w[1])\dots\expfuncof x a (w[|w|])\,\}, \text{ and}\\
    T_{\xtoepsilon}^{\mathit{reg}} & = \Big\{\,(w,w) \mid
    \forall 1\leq i \leq |w|\colon \forall p,q \in Q\colon w[i] = \regsym{x,p,q} \Rightarrow p=q \,\Big\} .
  \end{split}
\end{align}

\begin{example}
  Consider the regular constraint $x \in \langof{\aut_x}$ with $\aut_x$
  given below.
  \begin{center}
    \begin{tikzpicture}
  [node distance=2cm,->,>=stealth,transform shape, scale=0.8]

  \tikzstyle{empty}=[]

  \node[state,initial,accepting,circle] (q1) {1};
  \node[state,right of=q1,circle] (q2) {2};

  \draw (q1) edge[bend right] node[below] {$a$} (q2);
  \draw (q2) edge[bend right] node[above] {$b$} (q1);

\end{tikzpicture}

  \end{center}
  Then, the corresponding values will be as follows:
  \begin{align*}
    \encode(x\in\langof{\A_x}) = {}& \{ \regsym{x, 1, 1} \} \\
    \expfuncof x y (\regsym{x,1,1}) = {}& \{
      \regsym{y, 1, 1}.\regsym{x, 1, 1},~
      \regsym{y, 1, 2}.\regsym{x, 2, 1} \}\\
    \expfuncof x a (\regsym{x,1,1}) = {}& \{\regsym{x, 2, 1}\}
  \end{align*}
  Moreover, 
    $T_{\xtoyx}^{\mathit{reg}}$ will, e.g., contain pairs
    $(\regsym{x, 1, 1}.\regsym{y, 1, 2}, u_\preccurlyeq)$ with
    \[
    u_\preccurlyeq \in \left\{
        \regsym{x, 1, 1}.\regsym{y, 1, 1}.\regsym{y, 1, 2},~
        \regsym{x, 2, 1}.\regsym{y, 1, 2}
    \right\}
    \]
    (we assume that~$\preccurlyeq$ is a~lexicographic ordering on the
    components).
    Note that symbols in the words are sorted by~$\preccurlyeq$ with duplicates
    removed.
    Similarly, $T_{\xtoax}^{\mathit{reg}}$ will, e.g., contain the pair
    $(\regsym{x, 1, 1}.\regsym{y, 1, 2},~\regsym{x, 2, 1}.\regsym{y, 1, 2})$,
    and $T_{\xtoepsilon}^{\mathit{reg}}$ will contain
   $(\regsym{x, 1, 1}.\regsym{y, 1, 2},~\regsym{x, 1, 1}.\regsym{y, 1, 2})$.
\qed
\end{example}

The following lemma shows that the transformations are rational. In the proof, we
first construct \msostr formulae realizing necessary set operations on strings
and the effect of the expanding function. Based on them, we construct formulae
realizing the transformations, by means of the relation $\padfnc_{\pad}(T)$ 
appending to~$w$ and~$w'$ with $(w,w') \in T$ an arbitrary number of symbols $\pad$ 
(cf.~\cref{ssec:symbolicEncodingQuadratictoRMC}).
\begin{lemma}
	The relations $\padfnc_{\pad}(T_{\xtoyx}^{\mathit{reg}})$,
	$\padfnc_{\pad}(T_{\xtoax}^{\mathit{reg}})$, and
	$\padfnc_{\pad}(T_{\xtoepsilon}^{\mathit{reg}})$ are rational.
\end{lemma}
\begin{proof}
	In this proof, we extend the total order $\preccurlyeq$ on $\allsymreg$ to a
	total order on~$\allsymreg \cup \{\pad\}$ such that $\forall \sigma\in\allsymreg\colon
	\sigma\preccurlyeq\pad$.
  (We note that encoding $\preccurlyeq$ in \msostr{} is simple.)
  We define the relations $T_{\xtoyx}^{\mathit{reg}}$
	and $T_{\xtoepsilon}^{\mathit{reg}}$ using \msostr{}. The relation
	$T_{\xtoax}^{\mathit{reg}}$ can be defined analogously to
	$T_{\xtoyx}^{\mathit{reg}}$.
	\begin{align}
    \begin{split}
		    \psireg_{\xtoyx}(w,w') \defiff&~\existsw u_{1},u_{2},u_3\Big(\filter_x(u_{1},u_{2},w) \land \mathit{expand}_x^y(u_{1},u_3) \land {}\\
			     &\hspace*{2.15cm}\union(u_{2},u_3,w') \wedge \ordset(w')\Big)
    \end{split}
    \\
		\psireg_{\xtoepsilon}(w,w') \defiff&~\forallp i\Big(w[i] = w'[i] \wedge
			\bigvee_{\substack{\xi\in \allsymregx \cup\\
			\{ (x,q,q) \mid q \in Q_r \} }} \hspace*{-5mm}w'[i]=\xi\Big)
	\end{align}
	where $\filter_x(u,v,w)$ partitions the symbols of $w$ to $u$ and $v$ such that $u$
	contains symbols that are of the form 
	$\regsym{x,-,-}$, i.e., of the form $\regsym{x,q,s}$ for arbitrary $q$ and $s$,
	and $v$ contains the remaining ones; 
	$\mathit{expand}_x^y(u,v)$ replaces each symbol
	$\regsym{x,p,q}$ in $u$ with $\regsym{y,p,s}$ and $\regsym{x,s,q}$ in~$v$; 
	and
	$\union$ is a~set-like union.
  These predicates (including auxiliary predicates) are defined as follows:
	\begin{align}
		\sigma \in w \defiff&~\existsp i(w[i] = \sigma)\\
		\set(u) \defiff& \neg\existsp i,j(i\neq j \land  u[i] = u[j])\\
		\ordset(u) \defiff&~\set(u) \wedge \forallp i,j(i < j \rightarrow u[i] \preccurlyeq u[j])
	\end{align}
	\begin{align}
    \begin{split}
		   \filter_x(u,v,w) \defiff&~\forallp i\Bigg(\bigwedge_{q,s\in Q_r}\hspace*{-3mm}(w[i]=\regsym{x,q,s} \rightarrow (u[i] = \regsym{x,q,s} \land v[i] = \pad)) \\&
			   \wedge \bigwedge_{\substack{z\in\vars \setminus \{x\}\\ q,s\in Q_r}}\hspace*{-3mm}(w[i]=\regsym{z,q,s} \rightarrow (u[i] = \pad \land v[i] = \regsym{z,q,s}))\Bigg)
    \end{split}
	\end{align}
	\begin{align}
    \begin{split}
		\mathit{expand}_x^y(u,v) \defiff& \bigwedge_{\substack{s,q\in Q_r, s\leadsto q\\\xi'=\regsym{x,s,q}}}
			\Big( \xi' \in v \rightarrow \bigvee_{\substack{p\in Q_r, p\leadsto s\\ \xi = \regsym{x,p,q}\\
			\xi''=\regsym{y,p,s}}} \xi \in u \wedge \xi''\in v \Big) \land {}\\
			&\bigwedge_{\substack{p,s\in Q_r,p\leadsto s \\\xi''=\regsym{y,p,s}}}
				\Big( \xi'' \in v \rightarrow \bigvee_{\substack{q\in Q_r, s\leadsto q \\\xi = \regsym{x,p,q}\\
				\xi'=\regsym{x,s,q}}} \xi \in u \wedge \xi'\in v \Big) \land {}\\
			&\bigwedge_{\substack{p,q\in Q_r\\\xi = \regsym{x,p,q}}}
				\Big( \xi \in u \rightarrow \bigvee_{\substack{s\in Q_r, p\leadsto s\leadsto q\\\xi''=\regsym{y,p,s}\\
				\xi'=\regsym{x,s,q}}} \xi'' \in u \wedge \xi'\in v \Big)
    \end{split}
    \\
		\union(u,v,w) \defiff& \bigwedge_{\xi \in \allsymreg}
			\xi\in w \leftrightarrow (\xi\in u \vee \xi\in v)
	\end{align}
  Intuitively, in the definition of $\filter_x(u,v,w)$, the first part picks
  from~$w$ symbols containing~$x$ and adds them into~$u$ and the second part
  picks from~$w$ the other symbols and adds them into~$v$.
  On the other hand, in the definition of $\mathit{expand}_x^y(u,v)$, we
  use~$s$ to denote the \emph{splitting} state on the path from state~$p$ to
  state~$q$.
  Then, the first and the second parts of the formula denote that 
	$\regsym{y,p,s}$ and $\regsym{x,s,q}$ are in~$v$ while the last part denotes that
  $\regsym{x,p,q}$ is in~$u$.

	We further consider the relations $\transductof{+\padfnc} = \{ (w,w') \mid
	w\in(\allsymreg\cup\{\pad\})^{*}, w' \in w.\{\pad\}^{*} \}$ and
	$\transductof{-\padfnc} = \{ (w,w') \mid w'\in(\allsymreg\cup\{\pad\})^{*}, w
	\in w'.\{\pad\}^{*} \}$ appending and removing padding, respectively. These
	relations are rational. Then, observe that
	$\padfnc_{\pad}(T_{\xtoyx}^{\mathit{reg}}) =
	\transductof{+\padfnc}\circ\transductof{-\padfnc}\circ
	\langof{\psireg_{\xtoyx}}\circ \transductof{+\padfnc}$. 
	Recall the relation $\padfnc_{\pad}(T)$ 
	appends to $w'$ in $(w,w') \in T$ an arbitrary number of symbols $\pad$ 
	(cf. \cref{ssec:symbolicEncodingQuadratictoRMC}).
	From
	\cref{prop:rat-prop,prop:msostr-power}, we have that
	$\padfnc_{\pad}(T_{\xtoyx}^{\mathit{reg}})$ is rational (the same for
	$\padfnc_{\pad}(T_{\xtoepsilon}^{\mathit{reg}})$).
\end{proof}

The last missing piece is a definition of the destination set containing all
satisfiable regular constraints. For a~variable $x\in\vars$, we define the set
of satisfiable $x$-constraints as $L^x=\{\, w_\preccurlyeq \mid w =
\regsym{x,q_1,r_1}\cdots\regsym{x,q_n,r_n}\in\allsymreg^{*}, \bigcap_{i=1}^{n}
\langinof{\A_r}{q_{i}, r_{i}} \neq\emptyset \,\}$. Then, the destination set for
a~set of variables $\vars = \{ x_1,\dots,x_k \}$ is given as $\D_{\mathit{reg}}
= \{\, w_\preccurlyeq \mid w\in L^{x_1}\cdots L^{x_k}\,\}$. As in the case of
$\I_{\varphi_r}$, the set $\D_{\mathit{reg}}$ is finite and hence regular as well.

Satisfiability of a~quadratic word equation $\weq\colon \lhs = \rhs$ with a~regular
constraint~$\varphi_r$ can be expressed in the RMC framework as
$(\I_{\varphi_r}^{\reg},\T^{\reg}_{\xtoalphax}\cup\T^{\reg}_{\xtoepsilon},\D_{\varphi_r}^{\reg})$
instantiating \cref{alg:rmc} with items given 
as follows (note that we use a~fresh delimiter~$\regsep$):
\begin{align*}
	\I_{\varphi_r}^{\reg} &= \I^{\weq}\concat\{ \regsep \}\concat\padfnc_{\pad}(\I_{\varphi_r})\quad\quad\quad\quad\D_{\varphi_r}^{\reg} = \D^{\weq} \concat \{ \regsep \} \concat \padfnc_{\pad}(\D_{\mathit{reg}})\\
	\T^{\reg}_{\xtoalphax} &= \hspace*{-3.5mm}\bigcup_{y\in\vars, z \in\vars}\hspace*{-3.5mm}
  T_{\ytozy}^{\leq 2} \concat \{ \regsep \mapsto\regsep
  \} \concat \padfnc_{\pad}(T^{\mathit{reg}}_{\ytozy}) \cup{}\\
	& \bigcup_{y\in\vars, a \in\Sigma}\hspace*{-3.5mm} T_{\ytoay}^{\leq 2} \concat \{ \regsep \mapsto\regsep
  \} \concat \padfnc_{\pad}(T^{\mathit{reg}}_{\ytoay}) \\
	\T^{\reg}_{\xtoepsilon} &= \bigcup_{y\in\vars} T_{\ytoepsilon}^{\leq 2} \concat \{ \regsep \mapsto\regsep \} \concat \padfnc_{\pad}(T^{\mathit{reg}}_{\ytoepsilon})
\end{align*}
The rationality of $\T^{\reg}_{\xtoalphax}$ and $\T^{\reg}_{\xtoepsilon}$
follows directly from \cref{prop:rat-prop}. The soundness and
completeness of our procedure is summarized by the following lemma.
\begin{lemma}\label{lem:reg-sound}
	Given a quadratic word equation $\weq\colon \lhs = \rhs$ with a~regular constraint
	$\varphi_r$, \cref{alg:rmc} instantiated with
	$(\I_{\varphi_r}^{\reg},\T^{\reg}_{\xtoalphax}\cup\T^{\reg}_{\xtoepsilon},\D_{\varphi_r}^{\reg})$
	is sound, complete, and terminating.
\end{lemma}
\begin{proof}
  Similarly to proof of \cref{lem:len-sound}, we can
	generalize nodes of the Nielsen proof graph to pairs of the form $(\lhs' =
	\rhs', S)$ where $S\subseteq\allsymreg$. The transformation rules can be
	straightforwardly generalized to take into account also the regular
	constraints represented by a subset of $\allsymreg$. Since $\allsymreg$ is
	finite and $\weq$ is quadratic, the generalized proof graph is finite. The
	initial nodes of the generalized proof graph are exactly encoded strings from
	$\I_{\varphi_r}^{\reg}$, the final nodes correspond to $\D_{\varphi_r}^{\reg}$, and
	the transformation rules correspond to the encoded relations
	$\T^{\reg}_{\xtoalphax}$ and $\T^{\reg}_{\xtoepsilon}$. Since our RMC
	framework implements the BFS strategy, from the previous we get that our procedure
	is sound, complete, and terminating in proving satisfiability.
\end{proof}

As in the case of length constraints, the satisfiability of a word equation
$\weq$ with regular constraints can be generalized to a system of equations
$\sweq$ with regular constraints. The languages/relations corresponding to
$\weq$ are replaced by languages/relations corresponding to $\sweq$.
For a~system of word equations with regular constraints our algorithm is still
sound (and complete if we consider the transducers $\T^{i, \seq\star}$ for $\sweq$).

\paragraph{Discussion}
In the previous sections, we elaborated the proposed RMC framework for various kinds of string constraints including 
length and regular extensions. In \cref{tab:summary}, we give a summary of the achieved results. 
If the system of equations is quadratic (and even enriched with regular extensions), then our RMC approach 
is sound, complete, and terminating. It is basically due to the fact that the language $\procset$ 
in \cref{alg:rmc} for the transformations tailored for quadratic equations becomes saturated after 
a~finite number of steps. In all other cases, our RMC approach is sound and complete (but not generally terminating). 
For suitable encoded transformations, we are able to reach a solution after a finite number of steps 
(if the system is satisfiable). But in general, the language $\procset$ in \cref{alg:rmc-modified}
for these transformations is not guaranteed to eventually become saturated.

\begin{table}[t]
	\caption{Summary of the proposed approach on various types of string constraints with extensions. 
	\texttt{S} stands for \emph{sound}, \texttt{C} stands for \emph{complete}, and \texttt{T} stands for 
	\emph{terminating}.}
	\centering
	\begin{tabular}{|l|l|l|l|}
		\hline
			& \emph{Quadratic system} & \emph{General system} & \emph{Boolean combination} \\\hline\hline
	\emph{pure}    & \texttt{SCT} (\cref{sec:quad-system})      & \texttt{SC} (\cref{sec:gen-case})      & \texttt{SC} (\cref{section:full})              \\\hline
	\emph{length}  & \texttt{SC}  (\cref{sec:length-constraints})      & \texttt{SC} (\cref{sec:length-constraints})     & \texttt{SC} (\cref{sec:length-constraints})                 \\\hline
	\emph{regular} & \texttt{SCT} (\cref{sec:regular-constraints})      & \texttt{SC} (\cref{sec:regular-constraints})     & \texttt{SC} (\cref{sec:regular-constraints}) \\      \hline         
	\end{tabular}
	\label{tab:summary}
\end{table}

\newcommand{\figFRT}[0]{
\begin{figure}[t]
	\begin{subfigure}{0.99\textwidth}
		\centering
		\begin{tikzpicture}
    [node distance=4.0cm,->,>=stealth,transform shape,scale=0.6]
  
    \tikzstyle{state} = [draw=black,minimum height=0.5cm, inner sep=4pt,rectangle, rounded corners=1mm]
  
    \node[state, initial below, accepting, initial text={}] (s0) at (0,0) {$\langle \rangle$};
    \node[state, right of=s0] (s1) {$\langle {r} \rangle^\downarrow$};
    \node[state, accepting, node distance=4.0cm, right of=s1] (s2) {$\langle \rangle^\downarrow$};
    \node[state, left of=s0] (s3) {$\langle {r} \rangle^\uparrow$};
    \node[state, accepting, node distance=4.0cm, left of=s3] (s4) {$\langle \rangle^\uparrow$};

    \draw (s0) edge [loop above] node[above] {$\twotrack{\pad}{\pad}/\twotrack{\pad}{\pad}; \top; \nil$} (s0);
    \draw (s0) to node[above] {$\twotrack{\alpha}{\beta}/\emptyword; \alpha = x; r \gets\beta$} (s1);
    \draw (s0) to node[above] {$\twotrack{\alpha}{\beta}/\emptyword; \beta = x; r \gets\alpha$} (s3);

    \draw (s1) edge [loop above] node[above, text width=3.5cm, align=center] {$\twotrack{\alpha}{\beta}/\twotrack{\alpha}{r};$ \\$\beta\neq\pad \wedge \alpha \neq x \wedge \beta\neq x;$ \\$r \gets\beta$} (s1);
    \draw (s1) to node[above, text width=2.5cm, align=center] {$\twotrack{\alpha}{\beta}/\twotrack{\alpha}{r};$ \\$\beta=\pad \wedge \alpha \neq x;$\\$\nil$} (s2);
    \draw (s2) edge [loop above] node[above, text width=2.5cm, align=center] {$\twotrack{\alpha}{\beta}/\twotrack{\alpha}{\beta};$ \\$\beta=\pad \wedge \alpha \neq x;$ \\$\nil$} (s2);

    \draw (s3) edge [loop above] node[above, text width=3.5cm, align=center] {$\twotrack{\alpha}{\beta}/\twotrack{r}{\beta};$\\$\alpha\neq\pad \wedge \alpha \neq x \wedge \beta\neq x;$\\$r \gets\alpha$} (s3);
    \draw (s3) to node[above,text width=2.5cm, align=center] {$\twotrack{\alpha}{\beta}/\twotrack{r}{\beta};$\\$\alpha=\pad \wedge \beta \neq x;$\\$\nil$} (s4);
    \draw (s4) edge [loop above] node[above, text width=2.5cm, align=center] {$\twotrack{\alpha}{\beta}/\twotrack{\alpha}{\beta}; $\\$\alpha=\pad \wedge \beta \neq x; $\\$\nil$} (s4);
  
  \end{tikzpicture}
    \vspace{-2mm}

    \caption{An FRT for the relation $T_{x\mapsto \epsilon}^{\leq 1}$ for an
    arbitrary set of variables $\vars$ and alphabet $\Sigma$}
		\label{fig:trans-a}
	\end{subfigure}

  \vspace{2mm}
	\begin{subfigure}{0.99\textwidth}
      \centering
			\begin{tikzpicture}
    [node distance=4.0cm,->,>=stealth,transform shape,scale=0.6]
  
    \tikzstyle{state} = [draw=black,minimum height=0.5cm, inner sep=4pt,rectangle, rounded corners=1mm]
  
    \node[state, initial below, accepting, initial text={}] (s0) at (0,0) {$\langle \rangle$};
    \node[state, right of=s0] (s1) {$\langle {r_1} \rangle^\downarrow$};
    \node[state, accepting, node distance=4.0cm, right of=s1] (s2) {$\langle \rangle^\downarrow$};
    \node[state, left of=s0] (s3) {$\langle {r_1} \rangle^\uparrow$};
    \node[state, accepting, node distance=4.0cm, left of=s3] (s4) {$\langle \rangle^\uparrow$};

    \draw (s0) edge [loop above] node[above] {$\twotrack{\pad}{\pad}/\twotrack{\pad}{\pad}; \top; \nil$} (s0);
    \draw (s0) to node[above,text width=2.5cm, align=center] {$\twotrack{\alpha}{\beta}/\emptyword;$ \\$\alpha \in \vars;$ \\$r_1 \gets\beta, r_2 \gets \alpha$} (s1);
    \draw (s0) to node[above,text width=2.5cm, align=center] {$\twotrack{\alpha}{\beta}/\emptyword;$ \\$\beta \in \vars;$ \\$r_1 \gets\alpha, r_2\gets\beta $} (s3);

    \draw (s1) edge [loop above] node[above, text width=3.5cm, align=center] {$\twotrack{\alpha}{\beta}/\twotrack{\alpha}{r_1};$ \\$\beta\neq\pad \wedge \alpha \neq r_2 \wedge \beta\neq r_2;$ \\$r_1 \gets\beta$} (s1);
    \draw (s1) to node[above, text width=2.5cm, align=center] {$\twotrack{\alpha}{\beta}/\twotrack{\alpha}{r_1};$ \\$\beta=\pad \wedge \alpha \neq r_2;$\\$\nil$} (s2);
    \draw (s2) edge [loop above] node[above, text width=2.5cm, align=center] {$\twotrack{\alpha}{\beta}/\twotrack{\alpha}{\beta};$ \\$\beta=\pad \wedge \alpha \neq r_2;$ \\$\nil$} (s2);

    \draw (s3) edge [loop above] node[above, text width=3.5cm, align=center] {$\twotrack{\alpha}{\beta}/\twotrack{r_1}{\beta};$\\$\alpha\neq\pad \wedge \alpha \neq r_2 \wedge \beta\neq r_2;$\\$r_1 \gets\alpha$} (s3);
    \draw (s3) to node[above,text width=2.5cm, align=center] {$\twotrack{\alpha}{\beta}/\twotrack{r_1}{\beta};$\\$\alpha=\pad \wedge \beta \neq r_2;$\\$\nil$} (s4);
    \draw (s4) edge [loop above] node[above, text width=2.5cm, align=center] {$\twotrack{\alpha}{\beta}/\twotrack{\alpha}{\beta}; $\\$\alpha=\pad \wedge \beta \neq r_2; $\\$\nil$} (s4);
  
  \end{tikzpicture}
      \vspace{-2mm}
			\caption{An FRT for the relation $\bigcup_{x\in\vars}T_{x\mapsto \epsilon}^{\leq 1}$ 
	for an arbitrary set of variables $\vars$ and alphabet $\Sigma$}
			\label{fig:trans-b}
		\end{subfigure}

    \caption{FRT~(\subref{fig:trans-a}) 
	implementing the relation $T_{x\mapsto \epsilon}^{\leq 1}$ for an arbitrary set of variables $\vars$ and alphabet $\Sigma$.
	The register $r$ is used to store a shifted symbol.
	FRT~(\subref{fig:trans-b}) implementing the relation $\bigcup_{x\in\vars}T_{x\mapsto \epsilon}^{\leq 1}$ 
	for an arbitrary set of variables $\vars$ and alphabet $\Sigma$. The register
  $r_1$ is used to store a shifted symbol and the register $r_2$ is used 
	to store a variable that should be removed. In the figures, $\alpha$ and $\beta$ are
    variables representing the input symbol, $r, r_1$, and $r_2$ are registers, and the
    transitions are of the form \emph{action;condition;register update}.}
\end{figure}
}

%%%%%%%%%%%%%%%%%%%%%%%%%%%%%%%%%%%%%%%%%%%%%%%%%%%%%%%%%%%%%%%%%%%%%%%%%%%%%%%%
\section{Implementation}\label{sec:rmc-implementation}
%%%%%%%%%%%%%%%%%%%%%%%%%%%%%%%%%%%%%%%%%%%%%%%%%%%%%%%%%%%%%%%%%%%%%%%%%%%%%%%%

We created a~prototype Python tool called \retro\footnote{available at \texttt{\url{https://github.com/VeriFIT/retro}}}, where we implemented the
symbolic procedure for solving systems of word equations. \retro implements
a~modification of the RMC loop from \cref{alg:rmc}. In particular,
instead of standard transducers defined in \cref{sec:preliminaries}, it
uses the so-called \emph{finite-alphabet register transducers} (FRTs), which
allow a~more concise representation of a~rational relation.

Informally, an FRT is a~register automaton (in the sense of~\cite{KaminskiF94,DemriL09})
where the alphabet is finite. The finiteness of the alphabet implies that the
expressive power of FRTs coincides with the class of rational languages, but the
advantage of using FRTs is that they allow a~more concise representation than
ordinary transducers.
One can consider FRTs to be the restriction of symbolic transducers with
registers of \cite{VeanesHLMB12} to finite alphabets.
Operations for dealing with these transducers are then straightforward
restrictions of the operations considered by \cite{VeanesHLMB12} and therefore
do not elaborate on it here.

In particular, transducers (without registers) corresponding to the transformers
$\T_{\xtoalphax}$ and~$\T_{\xtoepsilon}$ contain branching at the beginning for
each choice of~$x$ and~$\alpha$. Especially in the case of huge alphabets, this
yields huge transducers (consider for instance the Unicode alphabet with over 1
million symbols). The use of FRTs yields much smaller automata because the
choice of~$x$ and~$\alpha$ is stored into registers and then processed
symbolically. 
To illustrate the effect of using registers, consider the transducer shown in \cref{fig:transducer-ex} implementing the
encoded relation $T_{x\mapsto \epsilon}^{\leq 1}$ for $\vars = \{ x\}$ and $\Sigma = \{ a\}$. 
The full transducer for large alphabets would
require a branching for each $\twotrack{u}{v}$, with $u,v \in \msoalph$, and a~lot of states to store
the concrete shifted symbols. In particular, it requires at least one pair of states $\langle v \rangle^\uparrow$ 
and $\langle v \rangle^\downarrow$ for each $v\in\msoalph$, which is unfeasible for very large~$\msoalph$.
On the other hand, the FRT in
\cref{fig:trans-a} stores the shifted symbols in the register~$r$,
the branching is replaced by a symbolic transition, and hence it requires just a couple of
states and transitions. Moreover, using an additional register $r_2$ to store the variable to replace, we are able to 
efficiently represent the relation $\bigcup_{x\in\vars}T_{x\mapsto
\epsilon}^{\leq 1}$, as shown in \cref{fig:trans-b}.

%%%%%%%%%%%%%%%%%%

\begin{figure}[t]
	\begin{subfigure}{0.99\textwidth}
		\centering
		\begin{tikzpicture}
    [node distance=4.0cm,->,>=stealth,transform shape,scale=0.6]
  
    \tikzstyle{state} = [draw=black,minimum height=0.5cm, inner sep=4pt,rectangle, rounded corners=1mm]
  
    \node[state, initial below, accepting, initial text={}] (s0) at (0,0) {$\langle \rangle$};
    \node[state, right of=s0] (s1) {$\langle {r} \rangle^\downarrow$};
    \node[state, accepting, node distance=4.0cm, right of=s1] (s2) {$\langle \rangle^\downarrow$};
    \node[state, left of=s0] (s3) {$\langle {r} \rangle^\uparrow$};
    \node[state, accepting, node distance=4.0cm, left of=s3] (s4) {$\langle \rangle^\uparrow$};

    \draw (s0) edge [loop above] node[above] {$\twotrack{\pad}{\pad}/\twotrack{\pad}{\pad}; \top; \nil$} (s0);
    \draw (s0) to node[above] {$\twotrack{\alpha}{\beta}/\emptyword; \alpha = x; r \gets\beta$} (s1);
    \draw (s0) to node[above] {$\twotrack{\alpha}{\beta}/\emptyword; \beta = x; r \gets\alpha$} (s3);

    \draw (s1) edge [loop above] node[above, text width=3.5cm, align=center] {$\twotrack{\alpha}{\beta}/\twotrack{\alpha}{r};$ \\$\beta\neq\pad \wedge \alpha \neq x \wedge \beta\neq x;$ \\$r \gets\beta$} (s1);
    \draw (s1) to node[above, text width=2.5cm, align=center] {$\twotrack{\alpha}{\beta}/\twotrack{\alpha}{r};$ \\$\beta=\pad \wedge \alpha \neq x;$\\$\nil$} (s2);
    \draw (s2) edge [loop above] node[above, text width=2.5cm, align=center] {$\twotrack{\alpha}{\beta}/\twotrack{\alpha}{\beta};$ \\$\beta=\pad \wedge \alpha \neq x;$ \\$\nil$} (s2);

    \draw (s3) edge [loop above] node[above, text width=3.5cm, align=center] {$\twotrack{\alpha}{\beta}/\twotrack{r}{\beta};$\\$\alpha\neq\pad \wedge \alpha \neq x \wedge \beta\neq x;$\\$r \gets\alpha$} (s3);
    \draw (s3) to node[above,text width=2.5cm, align=center] {$\twotrack{\alpha}{\beta}/\twotrack{r}{\beta};$\\$\alpha=\pad \wedge \beta \neq x;$\\$\nil$} (s4);
    \draw (s4) edge [loop above] node[above, text width=2.5cm, align=center] {$\twotrack{\alpha}{\beta}/\twotrack{\alpha}{\beta}; $\\$\alpha=\pad \wedge \beta \neq x; $\\$\nil$} (s4);
  
  \end{tikzpicture}
    \vspace{-2mm}

    \caption{An FRT for the relation $T_{x\mapsto \epsilon}^{\leq 1}$ for an
    arbitrary set of variables $\vars$ and alphabet $\Sigma$}
		\label{fig:trans-a}
	\end{subfigure}

  \vspace{2mm}
	\begin{subfigure}{0.99\textwidth}
      \centering
			\begin{tikzpicture}
    [node distance=4.0cm,->,>=stealth,transform shape,scale=0.6]
  
    \tikzstyle{state} = [draw=black,minimum height=0.5cm, inner sep=4pt,rectangle, rounded corners=1mm]
  
    \node[state, initial below, accepting, initial text={}] (s0) at (0,0) {$\langle \rangle$};
    \node[state, right of=s0] (s1) {$\langle {r_1} \rangle^\downarrow$};
    \node[state, accepting, node distance=4.0cm, right of=s1] (s2) {$\langle \rangle^\downarrow$};
    \node[state, left of=s0] (s3) {$\langle {r_1} \rangle^\uparrow$};
    \node[state, accepting, node distance=4.0cm, left of=s3] (s4) {$\langle \rangle^\uparrow$};

    \draw (s0) edge [loop above] node[above] {$\twotrack{\pad}{\pad}/\twotrack{\pad}{\pad}; \top; \nil$} (s0);
    \draw (s0) to node[above,text width=2.5cm, align=center] {$\twotrack{\alpha}{\beta}/\emptyword;$ \\$\alpha \in \vars;$ \\$r_1 \gets\beta, r_2 \gets \alpha$} (s1);
    \draw (s0) to node[above,text width=2.5cm, align=center] {$\twotrack{\alpha}{\beta}/\emptyword;$ \\$\beta \in \vars;$ \\$r_1 \gets\alpha, r_2\gets\beta $} (s3);

    \draw (s1) edge [loop above] node[above, text width=3.5cm, align=center] {$\twotrack{\alpha}{\beta}/\twotrack{\alpha}{r_1};$ \\$\beta\neq\pad \wedge \alpha \neq r_2 \wedge \beta\neq r_2;$ \\$r_1 \gets\beta$} (s1);
    \draw (s1) to node[above, text width=2.5cm, align=center] {$\twotrack{\alpha}{\beta}/\twotrack{\alpha}{r_1};$ \\$\beta=\pad \wedge \alpha \neq r_2;$\\$\nil$} (s2);
    \draw (s2) edge [loop above] node[above, text width=2.5cm, align=center] {$\twotrack{\alpha}{\beta}/\twotrack{\alpha}{\beta};$ \\$\beta=\pad \wedge \alpha \neq r_2;$ \\$\nil$} (s2);

    \draw (s3) edge [loop above] node[above, text width=3.5cm, align=center] {$\twotrack{\alpha}{\beta}/\twotrack{r_1}{\beta};$\\$\alpha\neq\pad \wedge \alpha \neq r_2 \wedge \beta\neq r_2;$\\$r_1 \gets\alpha$} (s3);
    \draw (s3) to node[above,text width=2.5cm, align=center] {$\twotrack{\alpha}{\beta}/\twotrack{r_1}{\beta};$\\$\alpha=\pad \wedge \beta \neq r_2;$\\$\nil$} (s4);
    \draw (s4) edge [loop above] node[above, text width=2.5cm, align=center] {$\twotrack{\alpha}{\beta}/\twotrack{\alpha}{\beta}; $\\$\alpha=\pad \wedge \beta \neq r_2; $\\$\nil$} (s4);
  
  \end{tikzpicture}
      \vspace{-2mm}
			\caption{An FRT for the relation $\bigcup_{x\in\vars}T_{x\mapsto \epsilon}^{\leq 1}$ 
	for an arbitrary set of variables $\vars$ and alphabet $\Sigma$}
			\label{fig:trans-b}
		\end{subfigure}

    \caption{FRT~(\subref{fig:trans-a}) 
	implementing the relation $T_{x\mapsto \epsilon}^{\leq 1}$ for an arbitrary set of variables $\vars$ and alphabet $\Sigma$.
	The register $r$ is used to store a shifted symbol.
	FRT~(\subref{fig:trans-b}) implementing the relation $\bigcup_{x\in\vars}T_{x\mapsto \epsilon}^{\leq 1}$ 
	for an arbitrary set of variables $\vars$ and alphabet $\Sigma$. The register
  $r_1$ is used to store a shifted symbol and the register $r_2$ is used 
	to store a variable that should be removed. In the figures, $\alpha$ and $\beta$ are
    variables representing the input symbol, $r, r_1$, and $r_2$ are registers, and the
    transitions are of the form \emph{action;condition;register update}.}
\end{figure}

%%%%%%%%%%%%%%%%%%

Concretely, \retro implements the decision procedure for a~system of general
word equations and length constraints (i.e., the procedures covered in
\cref{section:quadratic,section:conjunction,sec:length-constraints}).
It does not implement
\begin{inparaenum}[(i)]
  \item  Boolean combinations of constraints (\cref{section:full}) and
  \item  regular constraints (\cref{sec:regular-constraints}),
\end{inparaenum}
which are quite inefficient and are provided in order for our approach to have
a~more robust theoretical formal basis for a~large fragment of input constraints
(cf.\ the discussion at the beginning of \cref{section:full}).

As another feature, \retro uses deterministic finite automata to represent
configurations in \cref{alg:rmc}. It also uses eager automata
minimization, since it has a big impact on the performance, especially on
checking the termination condition of the RMC algorithm, which is done by
testing language inclusion between the current configuration and all so-far
processed configurations.

%%%%%%%%%%%%%%%%%%%%%%%%%%%%%%%%%%%%%%%%%%%%%%%%%%%%%%%%%%%%%%%%%%%%%%%%%%%%%%%%
\section{Experimental Evaluation}\label{sec:rmc-experiments}
%%%%%%%%%%%%%%%%%%%%%%%%%%%%%%%%%%%%%%%%%%%%%%%%%%%%%%%%%%%%%%%%%%%%%%%%%%%%%%%%

\newcommand{\figretroscatter}[0]{
\begin{figure}[t]
  \begin{subfigure}[b]{0.49\linewidth}
  \begin{center}
  \includegraphics[width=\linewidth,keepaspectratio]{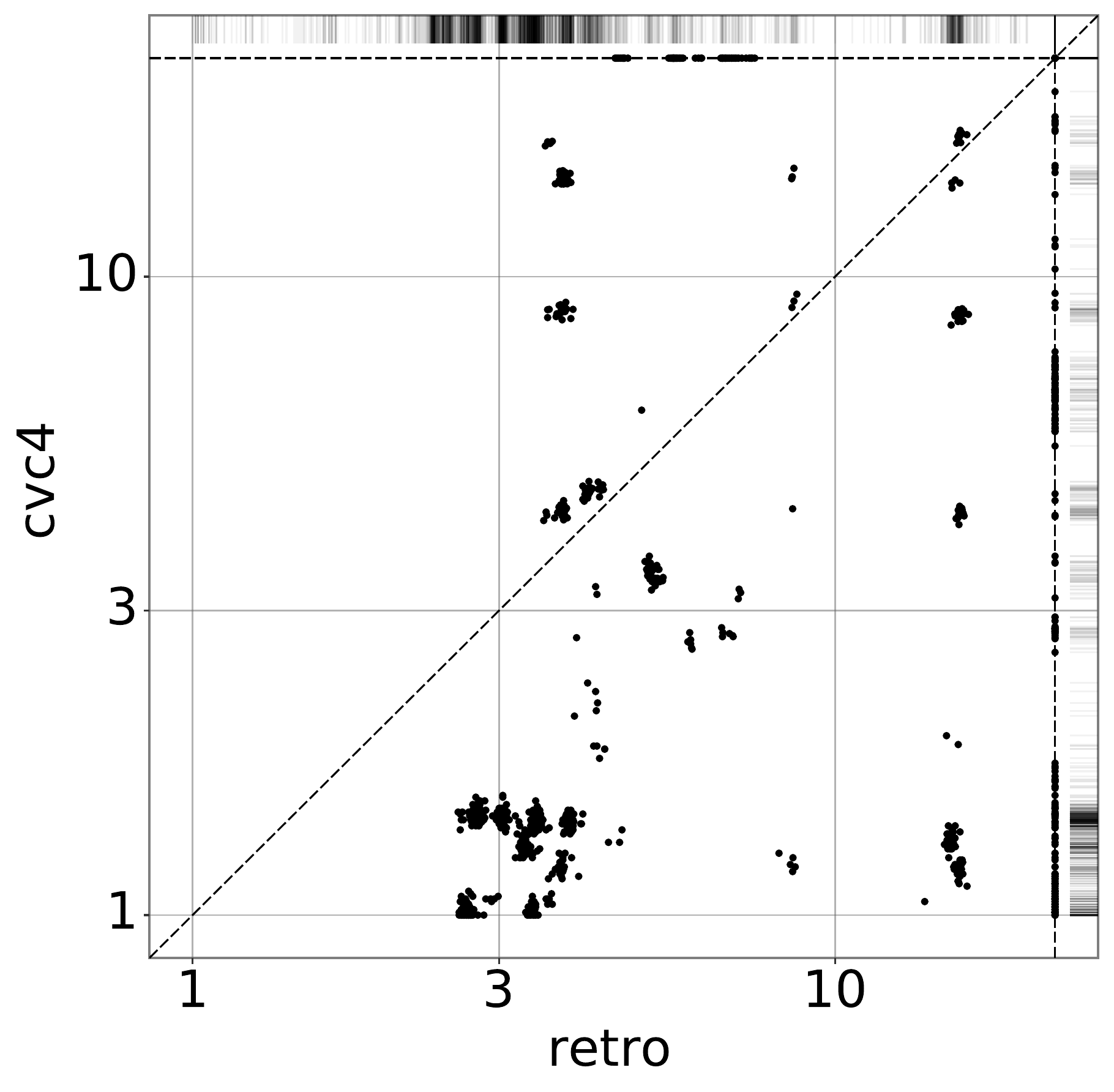}
  \end{center}
  \vspace{-0mm}
  \caption{\retro vs CVC4}
  \label{fig:schewe}
  \end{subfigure}
  \begin{subfigure}[b]{0.49\linewidth}
  \begin{center}
  \includegraphics[width=\linewidth,keepaspectratio]{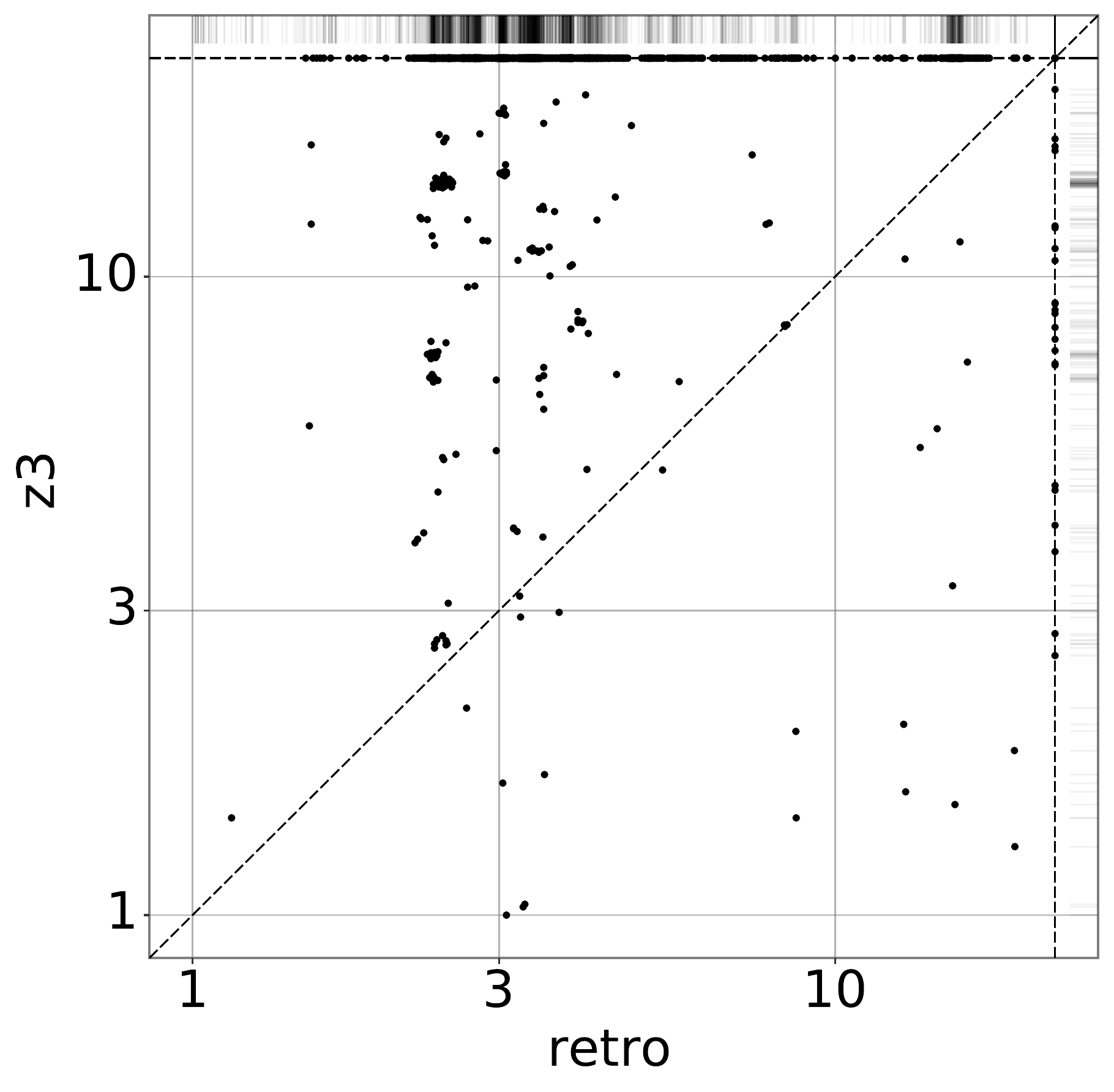}
  \end{center}
  \vspace{-0mm}
  \caption{\retro vs Z3}
  \label{fig:rankerold}
  \end{subfigure}
  \vspace{-0mm}
  \caption{Comparison of \retro with CVC4 and Z3 on \pyexhard. We show difficult 
  instances that took more than 1\,s to finish. Times are given in seconds, axes
  are logarithmic.}
\label{fig:retroscatter}
\end{figure}
}

\newcommand{\figvbs}[0]{
\begin{figure}[t]
   \centering
   \begin{subfigure}[b]{0.49\linewidth}
    \begin{center}
      \includegraphics[width=1.0\linewidth,keepaspectratio]{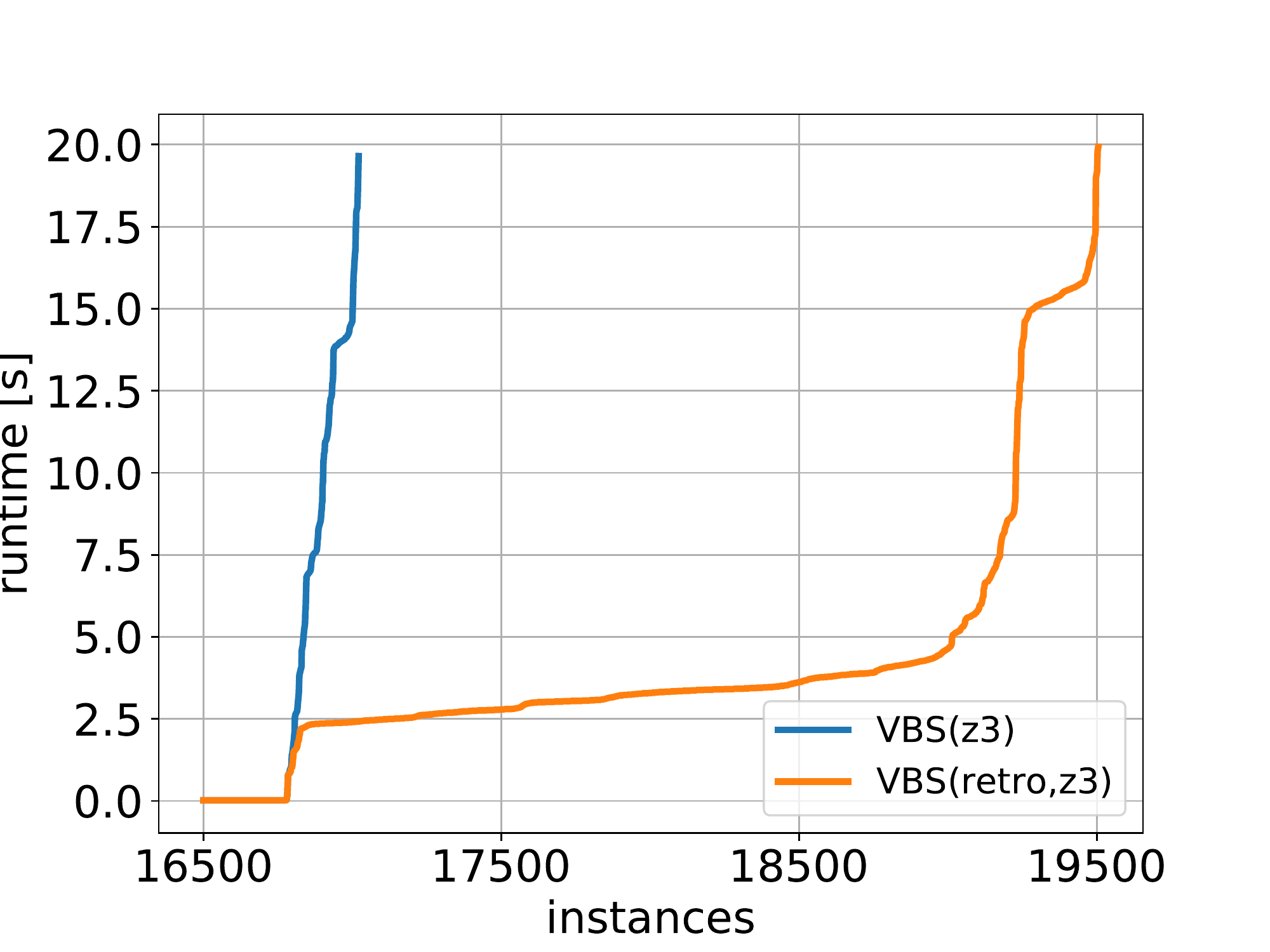}
    \end{center}
    \vspace{-0mm}
    \caption{\textit{VBS}(\retro, Z3)}
    \label{fig:vbs-z3}
    \end{subfigure}
    \begin{subfigure}[b]{0.49\linewidth}
    \begin{center}
      \includegraphics[width=1.0\linewidth,keepaspectratio]{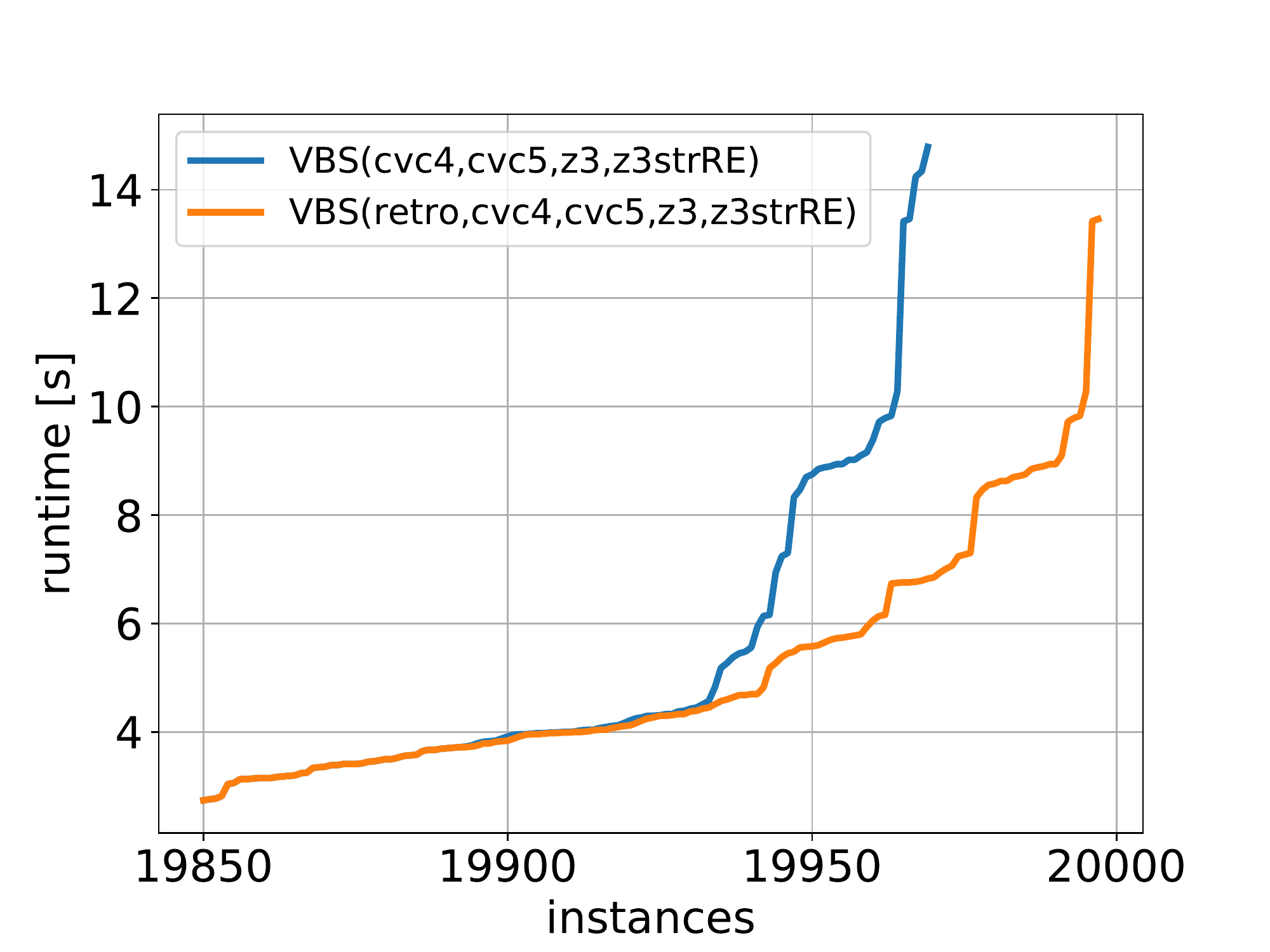}
    \end{center}
    \vspace{-0mm}
    \caption{\textit{VBS}(\retro, CVC4, CVC5, Z3, \zstrre)}
    \label{fig:vbs}
    \end{subfigure}
    \vspace{-0mm}
   \caption{A cactus plot comparing the Virtual Best Solver of
   (\subref{fig:vbs-z3}) Z3 with and without \retro  and (\subref{fig:vbs}) all
   tools with and without \retro  on the \pyexhard benchmark.
    We show only the most difficult benchmarks (out of 20,020).}
  \label{fig:vbsplots}
\end{figure}
}

\newcommand{\figkepler}[0]{
\begin{figure}[t]
   \centering
   \includegraphics[width=0.8\linewidth,keepaspectratio]{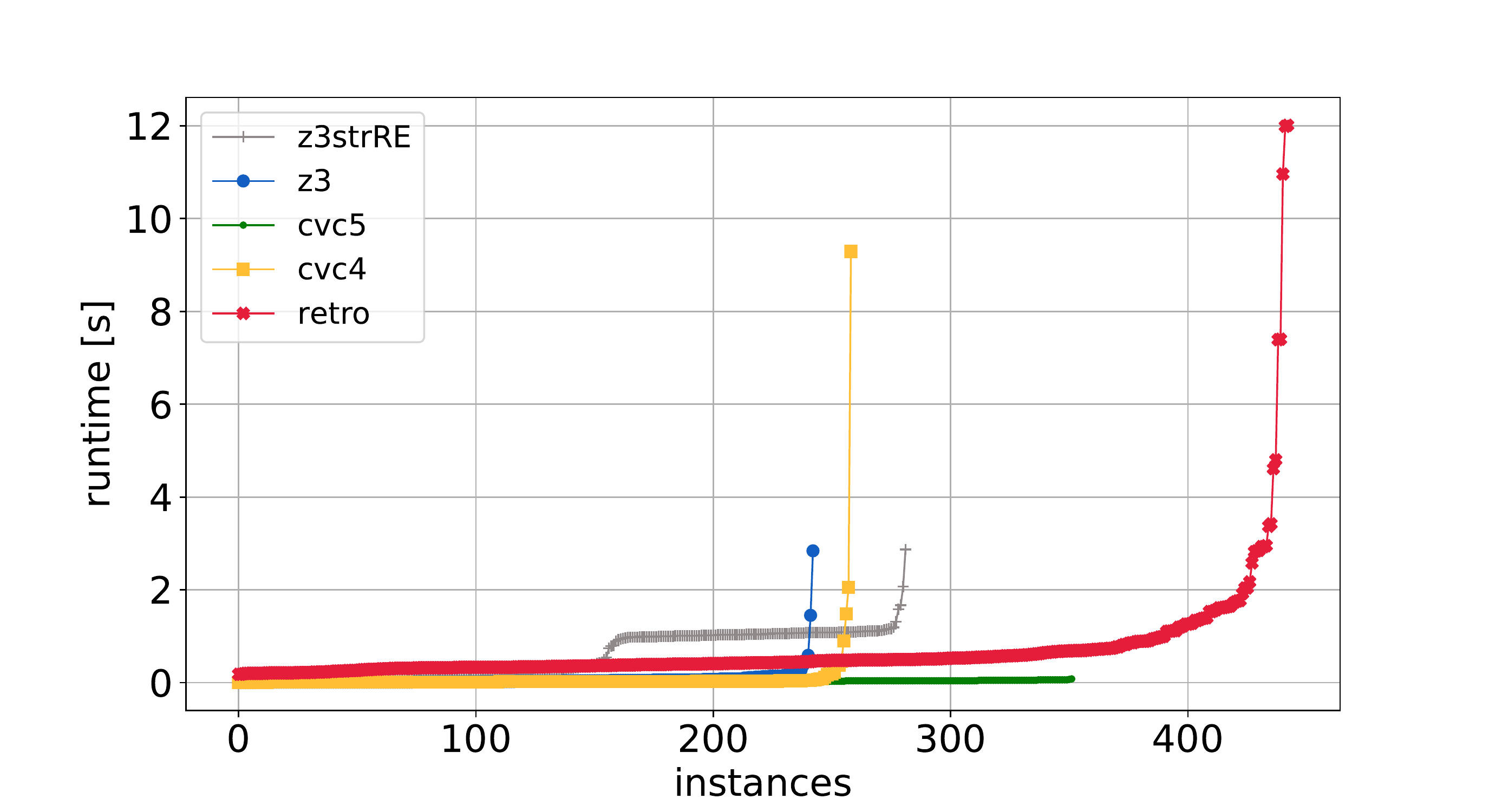}
   \caption{A cactus plot comparing \retro, CVC4, CVC5, Z3, and \zstrre on \kepler}
   \label{fig:cactus-kepler}
\end{figure}
}

\newcommand{\tabkepler}[0]{
\begin{table}[b]
\caption{A break-down of solved cases on the \kepler benchmark.  A~number on row~$X$ in column~$Y$ denotes the number of cases that
  solver~$X$ could solve and solver~$Y$ could not solve.
  }
  \vspace{-5mm}
\begin{center}
\small
\begin{tabular}{lrrrrr}
  \toprule
  & \retro & CVC4 & CVC5 & Z3 & \zstrre \\
  \midrule
  \retro &  \multicolumn{1}{c}{---} & 185 & 92 & 201 & 161 \\
  CVC4 & 1 &\multicolumn{1}{c}{---}   & 2 & 104 & 63  \\
  CVC5 &1 & 95 & \multicolumn{1}{c}{---} & 197 & 145 \\
  Z3 & 1 & 88 & 88 & \multicolumn{1}{c}{---} & 23 \\
  \zstrre & 0 & 86 & 75 & 62 & \multicolumn{1}{c}{---}\\
\bottomrule
\end{tabular}

\end{center}
\label{tab:kepler}
\end{table}
}

We compared the performance of our approach (implemented in \retro) with four
current state-of-the-art SMT solvers that support the string theory:
Z3~4.8.14 (\cite{MouraB08}), \zstrre (\cite{berzish2021smt}), CVC4~1.8 (\cite{cvc4Tool}), and CVC5~1.0.1 (\cite{cvc5}).
Regarding other solvers that we are aware of, the performance of \norn from
\cite{abdulla2015norn} and \ostrich from \cite{chen2019decision} was much worse than the
considered tools, the performance of \zthreestrfour (\cite{MoraBKNG21}) was similar to that of
\zstrre, and \sloth of \cite{HolikJLRV18} was unsound on the considered fragment
(it supports only the so-called \emph{straight-line fragment}).

The first set of benchmarks is \kepler, obtained from~\cite{LeH18}.
\kepler contains 600 hand-crafted string constraints composed of quadratic word
equations with length constraints.
In \cref{fig:cactus-kepler}, we give a~cactus plot of the results of the
solvers on the \kepler benchmark set with the timeout of 20\,s.
In cactus plots, the closer a solver's plot is to the right and bottom borders, the better is the corresponding solver.
The total numbers of solved benchmarks within the timeout were:
243 for Z3,
282 for \zstrre,
259 for CVC4,
352 for CVC5, and
443 for \retro.
We can refine these numbers by comparing the cases solved by each pair of tools, as shown in \cref{tab:kepler}.
From the table we can see that \retro solves significantly more benchmarks than all other state-of-the-art tools, from 92 when compared with CVC5 to 201 with Z3.
Only in one case \retro failed (as did \zstrre) while CVC4, CVC5, and Z3 succeeded.
Except for CVC4 vs.\@ CVC5, the comparison between the other tools is inconclusive, as the entries $(X,Y)$ and $(Y,X)$ of the table both contain large values.

%%%%%%%%%%%%%%
\figkepler
%%%%%%%%%%%%%%

%%%%%%%%%%%%%%

\begin{table}[b]
\caption{A break-down of solved cases on the \kepler benchmark.  A~number on row~$X$ in column~$Y$ denotes the number of cases that
  solver~$X$ could solve and solver~$Y$ could not solve.
  }
  \vspace{-5mm}
\begin{center}
\small
\begin{tabular}{lrrrrr}
  \toprule
  & \retro & CVC4 & CVC5 & Z3 & \zstrre \\
  \midrule
  \retro &  \multicolumn{1}{c}{---} & 185 & 92 & 201 & 161 \\
  CVC4 & 1 &\multicolumn{1}{c}{---}   & 2 & 104 & 63  \\
  CVC5 &1 & 95 & \multicolumn{1}{c}{---} & 197 & 145 \\
  Z3 & 1 & 88 & 88 & \multicolumn{1}{c}{---} & 23 \\
  \zstrre & 0 & 86 & 75 & 62 & \multicolumn{1}{c}{---}\\
\bottomrule
\end{tabular}

\end{center}
\label{tab:kepler}
\end{table}

%%%%%%%%%%%%%%

The other set of benchmarks that we tried is \pyexhard. Here we wanted to see the
potential of integrating \retro with DPLL(T)-based string solvers, like Z3 or
CVC4, as a~specific string theory solver.
The input of this component is a conjunction of atomic string formulae (e.g.,
$xy=zb \wedge z=ax$) that is a~model of the Boolean structure of the top-level
formula.
The conjunction of atomic string formulae is then, in several layers, processed
by various string theory solvers, which either add more conflict clauses or
return a model.
To~evaluate whether \retro is suitable to be used as ``one of the layers'' of
Z3 or CVC4's string solver, we analyzed the PyEx
benchmarks from \cite{reynolds2017scaling} and extracted from it 967 difficult instances that neither CVC4 nor Z3 could solve in 10 seconds.
From those instances, we obtained 20,020 conjunctions of word equations that
Z3's DPLL(T) algorithm sent to its string theory solver when trying to solve
them.
We call those 20,020 conjunctions of word equations \pyexhard.
We then evaluated the other solvers on \pyexhard with the timeout of 20\,s.
Out of these, Z3 could not solve 3,001, \zstrre 814, CVC4 152, CVC5 171, and
\retro could not solve 3,079 instances.

\figretroscatter

Let us now look closely at the hard instances in the \pyexhard benchmark set,
in particular, on the instances that the other tools could not solve.
These benchmarks cannot be handled by the (several layers of) fast heuristics
implemented in these tools, which are sufficient to solve many benchmarks
without the need to start applying the case-split rule.%
\footnote{%
For instance, when Z3 receives the word equation $xy=yax$, it infers the length
constraint $|x|+|y| = |y|+1+|x|$, which implies unsatisfiability of the word
equation without the need to start applying the case-split rule at all.%
}
In \cref{fig:retroscatter} we give a comparison of the running times of \retro 
with CVC4 and Z3 with a particular focus on the difficult instances 
(running time above 1\,s). 
The plots about CVC5 and \zstrre show similar trends.
From the figure we can see that on many hard instances \retro can 
provide the answer much faster than the other tool, in particular for Z3.
When we compare the solvers on the examples that the other tools failed to solve, 
we have that
\retro could solve 86 examples (56.2\,\%) out of those where CVC4 failed, 
130 examples (76.02\,\%) where CVC5 failed,
2,484 examples (82.7\,\%) where Z3 failed, and
519 examples (63.75\,\%) where \zstrre failed. 
Moreover, \retro solved 28 instances as the only tool. 
Lastly, we consider how \retro affects the \emph{Virtual Best Solver}: 
given a set of solvers~$S$, we use $\mathit{VBS(S)}$ to denote the solver that would be
obtained by taking, for each benchmark, the fastest solver on the given benchmark.
In \cref{fig:vbsplots}, we provide cactus plots showing the impact of \retro on two instances of Virtual Best
Solvers; 
the plots for \zstrre, CVC4, and CVC5 are similar, with the one for \zstrre
being closer to \cref{fig:vbs-z3} while the ones for CVC4/5 being closer to \cref{fig:vbs}.
As we can see, the plots show that our approach can significantly help solvers deal with hard equations.

%%%%%%%%%%%%%%%%
\figvbs
%%%%%%%%%%%%%%%%

%------------------------------------------------------------------------------
\paragraph{Discussion}

From the obtained results, we see that our approach works well in \emph{difficult
cases}, where the fast heuristics implemented in state-of-the-art solvers are
not sufficient to quickly discharge a~formula, which happens in particular when the
(un)satisfiability proof is complex.
Our approach can exploit the symbolic representation of the proof tree and
use it to reduce the redundancy of performing transformations.
Note that we can still beat the heavily optimized Z3, \zstrre, CVC4, and CVC5 written in C++ by
a~Python prototype in those cases.
We believe that implementing our symbolic algorithm as a~part of a~state-of-the-art SMT
solver would push the applicability of string solving even further, especially
for cases of string constraints with a~complex structure, which need to solve
multiple DPLL(T) queries in order to establish the (un)satisfiability of
a~string formula.

%%%%%%%%%%%%%%%%%%%%%%%%%%%%%%%%%%%%%%%%%%%%%%%%%%%%%%%%%%%%%%%%%%%%%%%%%%%%%%%%
\vspace{-0.0mm}
\section{Related Work}\label{sec:related}
\vspace{-0.0mm}
%%%%%%%%%%%%%%%%%%%%%%%%%%%%%%%%%%%%%%%%%%%%%%%%%%%%%%%%%%%%%%%%%%%%%%%%%%%%%%%%

The study of solving string constraint traces back to 1946, when
\cite{quine1946concatenation} showed that the first-order theory of word
equations is undecidable.
\cite{makanin1977problem} achieved a~milestone result
by showing that the class of quantifier-free word equation is decidable.
Since then, several works, e.g.,
\cite{plandowski1999satisfiability,plandowski2006efficient,matiyasevich2008computation,robson1999quadratic,schulz1990makanin,ganesh2012word,ganesh2016undecidability,abdulla2014string,barcelo2013graph,lin2016string,ChenCHLW18WhatIsDecidable,chen2019decision,abdulla2019chain},
consider the decidability and complexity of different classes of string
constraints.
Efficient solving of satisfiability of string constraints is a~challenging problem.
Moreover, decidability of the problem of satisfiability of word equations
combined with length constraints of the form $|x|=|y|$ has already been open for
over 20 years~\cite{buchi1990definability}.

The strong practical motivation led to the rise of several string constraint
solvers that concentrate on solving practical problem instances.
The typical procedure implemented within \emph{DPLL(T)-based} string
solvers
is to split the constraints into simpler sub-cases based on how the solutions
are aligned, combining with powerful techniques for Boolean reasoning to
efficiently explore the resulting exponentially-sized search space.
The case-split rule is usually performed explicitly.
Examples of solvers implementing this approach are 
\norn (\cite{abdulla2014string,abdulla2015norn}),
\trau (\cite{abdulla2018trau}),
\ostrich (\cite{chen2019decision}),
\sloth (\cite{HolikJLRV18}),
CVC4 (\cite{cvc4Tool}),
CVC5 (\cite{cvc5}),
\textsc{Z3str2} (\cite{zheng2017z3str2}),
\textsc{Z3str3} (\cite{BerzishGZ17Z3str3}),
\zthreestrfour (\cite{MoraBKNG21}),
\zstrre (\cite{berzish2021smt}),
S3 (\cite{trinh2014s3}),
S3P (\cite{trinh2016progressive}).
In contrast, our approach performs case-splits symbolically.

Automata and transducers have been used in many approaches and tools for string
solving, such as in
\norn (\cite{abdulla2014string,abdulla2015norn}),
\trau (\cite{abdulla2018trau}),
\ostrich (\cite{chen2019decision}),
\sloth (\cite{HolikJLRV18}),
\slog (\cite{WangTLYJ16StringAutomata}),
\slent (\cite{WangCYJ18}), or
\zstrre (\cite{berzish2021smt}), and also in
string solvers for analyzing string-manipulating programs, such as
ABC (\cite{aydin2018parameterized}) and Stranger (\cite{yu2010stranger}), which
soundly over-approximate string constraints using
transducers~(\cite{yu2016optimal}).
The main difference of these approaches to ours is that they use transducers to
encode possible models (solutions) to the string constraints, while we use
automata and transducers to encode the string constraint transformations.

Other approaches for solving string constraints include reducing the constraints
to the SMT theory of bit vectors (e.g.,
\cite{zheng2017z3str2,BerzishGZ17Z3str3,MoraBKNG21,KiezunGAGHE12HAMPI}),
the theory of arrays (e.g., \cite{LiG13}), or SAT-solving
(e.g., \cite{DayEKMNP19,AmadiniGST17,ScottFPS17}).
Not so many approaches are based on algebraic approaches, such as the Nielsen
transformation.
In addition to our approach, it is also used as the basis of the work of \cite{LeH18}.
On the other hand, the Nielsen transformation (\cite{Nielsen17}) is used by some
tools that implement different approaches to discharge
quadratic equations (e.g., \ostrich of \cite{chen2019decision}).
Complex rewriting rules are used, e.g., when dealing with regular constraints in
CVC5 (\cite{NotzliRBBT22}).

%%%%%%%%%%%%%%%%%%%%%%%%%%%%%%%%%%%%%%%%%%%%%%%%%%%%%%%%%%%%%%%%%%%%%%%%%%%%%%%%
\vspace{-0.0mm}
\section*{Acknowledgements}\label{sec:label}
\vspace{-0.0mm}
%%%%%%%%%%%%%%%%%%%%%%%%%%%%%%%%%%%%%%%%%%%%%%%%%%%%%%%%%%%%%%%%%%%%%%%%%%%%%%%%

We thank
Mohamed Faouzi Atig for discussing the topic.
This work has been partially supported by
the National Natural Science Foundation of China (grant no.\@ 61836005),
the CAS Project for Young Scientists in Basic Research (grant no.\@ YSBR-040),
the Czech Science Foundation project 19-24397S,
the FIT BUT internal project FIT-S-20-6427, and
the project of Ministry of Science and Technology, Taiwan (grant no.\ MOST-109-2628-E-001-001-MY3).
% the Czech Ministry of Education, Youth and Sports project LL1908 of the ERC.CZ programme,
\newline\protect\includegraphics[height=8pt]{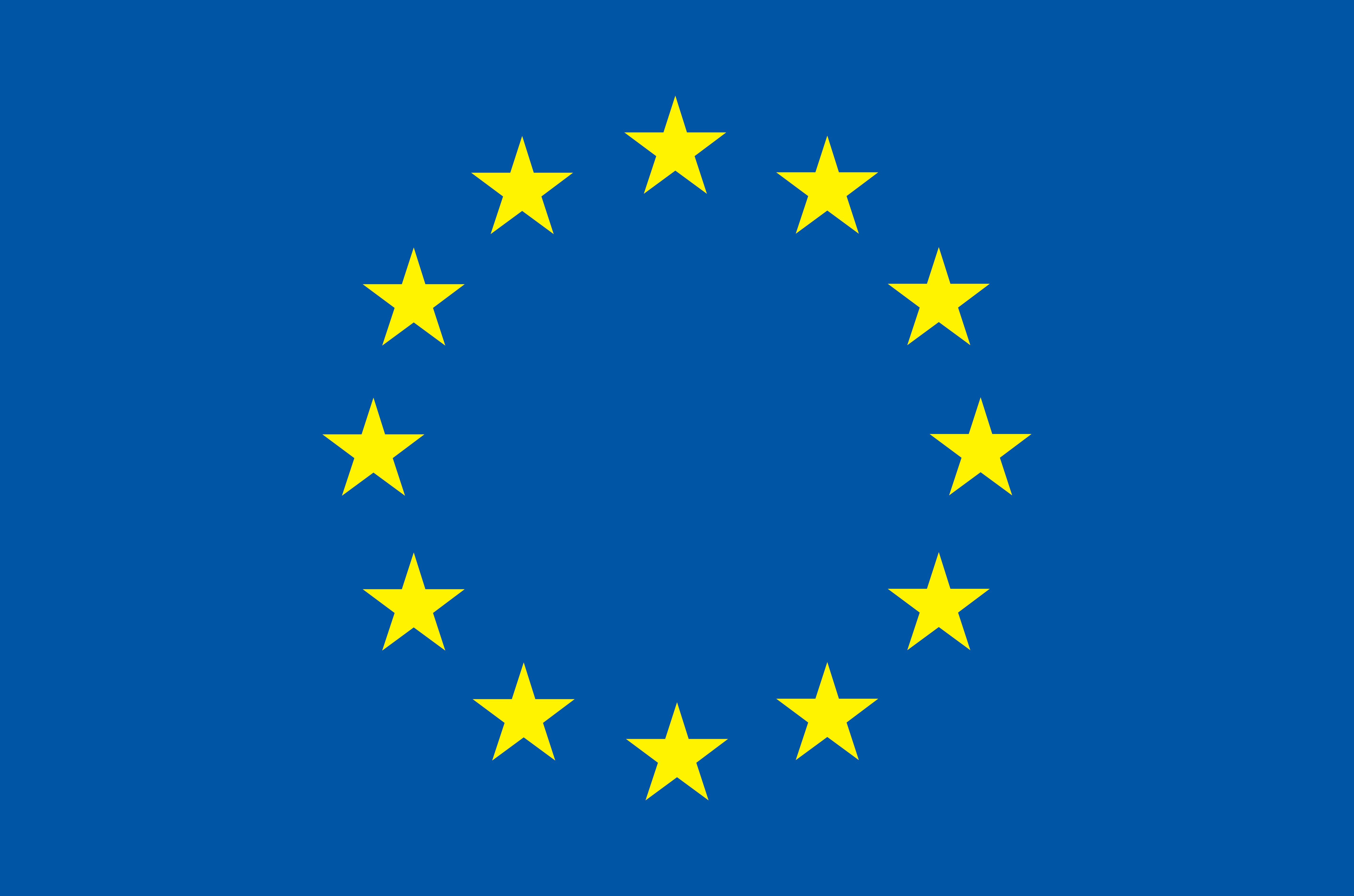} This work is part of the European Union’s 
Horizon 2020 research and innovation programme under the 
Marie Sk\l{}odowska-Curie 
% Marie Sklodowska-Curie 
grant no.\@ 101008233.

%%%%%%%%%%%%%%%%%%%%%%%%%%%%%%%%%%%%%%%%%%%%%%%%%%%%%%%%%%%%%%%%%
\bibliographystyle{elsarticle-harv}
\bibliography{bibliography}
%%%%%%%%%%%%%%%%%%%%%%%%%%%%%%%%%%%%%%%%%%%%%%%%%%%%%%%%%%%%%%%%%

\end{document}